\documentclass{amsart}
\usepackage{amssymb,delarray,graphicx,isomath,subfig,float,slashed,diagrams}
\usepackage[text={14cm,22cm},centering]{geometry}
\usepackage[pdftex,bookmarksopen,bookmarksnumbered,
            unicode,plainpages=false,pdfpagelabels,
            pdfauthor={Domenic Denicola, Matilde Marcolli, and Ahmad Zainy Al-Yasry},
            pdftitle={Spin foams and noncommutative geometry}]{hyperref}
\usepackage{aliascnt}
\usepackage{modifiedcleveref}

\newcommand{\setuptheorem}[5]
{
    \newaliascnt{#1}{thm}
    \newtheorem{#1}[#1]{#3}
    \aliascntresetthe{#1}
    \crefname{#1}{#2}{#4}
    \Crefname{#1}{#3}{#5}
}

\theoremstyle{plain}
    \newtheorem{thm}{Theorem}[section]
    \crefname{thm}{theorem}{theorems}
    \Crefname{thm}{Theorem}{Theorems}

    \setuptheorem{prop}{proposition}{Proposition}{propositions}{Propositions}
    \setuptheorem{cor}{corollary}{Corollary}{corollaries}{Corollaries}
    \setuptheorem{lem}{lemma}{Lemma}{lemmas}{Lemmas}
    \setuptheorem{ques}{question}{Question}{questions}{Questions}
    \setuptheorem{defn}{definition}{Definition}{definitions}{Definitions}
    \setuptheorem{rem}{remark}{Remark}{remarks}{Remarks}
    \setuptheorem{conj}{conjecture}{Conjecture}{conjectures}{Conjectures}

\theoremstyle{definition}
    \setuptheorem{ex}{example}{Example}{examples}{Examples}

\numberwithin{equation}{section}

\def\bA{{\mathbb A}}

\def\bH{{\mathbb H}}

\def\bO{{\mathbb O}}

\def\C{{\mathbb C}}

\renewcommand{\H}{{\mathbb H}}
\def\N{{\mathbb N}}

\def\Q{{\mathbb Q}}
\def\Z{{\mathbb Z}}
\def\R{{\mathbb R}}

\def\fW{{\mathfrak W}}

\def\fs{{\mathfrak s}}
\def\ft{{\mathfrak t}}

\def\cA{{\mathcal A}}
\def\cB{{\mathcal B}}
\def\cC{{\mathcal C}}
\def\cD{{\mathcal D}}
\def\cE{{\mathcal E}}

\def\cG{{\mathcal G}}
\def\cH{{\mathcal H}}

\def\cL{{\mathcal L}}
\def\cM{{\mathcal M}}

\def\cS{{\mathcal S}}

\newcommand{\ii}{\mathrm{i}}
\newcommand{\ee}{\mathrm{e}}

\def\Aut{{\rm Aut}}

\def\id{{\rm id}}

\def\SL{{\rm SL}}

\def\Tr{{\rm Tr}}

\newcommand{\Vol}{\mathrm{Vol}}
\newcommand{\Iso}{\mathrm{Iso}}
\newcommand{\Mor}{\mathrm{Mor}}
\newcommand{\MorTwo}{\Mor^{\smash[t]{(2)}}}
\newcommand{\Obj}{\mathrm{Obj}}
\newcommand{\Diff}{\mathrm{Diff}}

\crefname{enumi}{}{}
\Crefname{enumi}{}{}
\creflabelformat{enumi}{(#2#1#3)}

\title{Spin foams and noncommutative geometry}
\author{Domenic Denicola}
\author{Matilde Marcolli}
\author{Ahmad Zainy al-Yasry}
\email{domenic\@@domenicdenicola.com}
\email{matilde\@@caltech.edu}
\email{azainy79\@@yahoo.com}
\address{MSC \#240 \\
California Institute of Technology \\ 1200 E.\ California Blvd.\\
Pasadena, CA 91126, USA}
\address{Division of Physics, Mathematics, and Astronomy\\
California Institute of Technology \\ 1200 E.\ California Blvd.\\
Pasadena, CA 91125, USA}
\address{Abdus Salam International Center for Theoretical Physics \\
Strada Costiera 11 \\ Trieste \\ I-34151 \\ Italy}

\begin{document}
\maketitle

\begin{abstract}
    We extend the formalism of embedded spin networks and spin foams to include topological data that encode the underlying three-manifold or four-manifold as a branched cover. These data are expressed as monodromies, in a way similar to the encoding of the gravitational field via holonomies. We then describe convolution algebras of spin networks and spin foams, based on the different ways in which the same topology can be realized as a branched covering via covering moves, and on possible composition operations on spin foams. We illustrate the case of the groupoid algebra of the equivalence relation determined by covering moves and a 2-semigroupoid algebra arising from a 2-category of spin foams with composition operations corresponding to a fibered product of the branched coverings and the gluing of cobordisms. The spin foam amplitudes then give rise to dynamical flows on these algebras, and the existence of low temperature equilibrium states of Gibbs form is related to questions on the existence of topological invariants of embedded graphs and embedded two-complexes with given properties.  We end by sketching a possible approach to combining the spin network and spin foam formalism with matter within the framework of spectral triples in noncommutative geometry.
\end{abstract}

\tableofcontents

\section{Introduction}

In this paper we extend the usual formalism of spin networks and spin foams, widely used in the context of loop quantum gravity, to encode the additional information on the topology of the ambient smooth three- or four-manifold, in the form of branched covering data. In this way, the usual data of \textit{holonomies} of connections, which provide a discretization of the gravitational field in LQG models, is combined here with additional data of \textit{monodromies}, which encode in a similar way the smooth topology.

The lack of uniqueness in the description of three-manifolds and four-manifolds as branched coverings determines an equivalence relation on the set of our topologically enriched spin networks and foams, which is induced by the covering moves between branch loci and monodromy representations. One can associate to this equivalence relation a groupoid algebra. We also consider other algebras of functions on the space of all possible spin networks and foams with topological data, and in particular a 2-semigroupoid algebra coming from a 2-category which encodes both the usual compositions of spin foams as cobordisms between spin networks and a fibered product operation that parallels the KK-theory product used in D-brane models.

The algebras obtained in this way are associative noncommutative algebras, which can be thought of as noncommutative spaces parameterizing the collection of all topologically enriched spin foams and networks with the operations of composition, fibered product, or covering moves. The lack of covering-move invariance of the spin foam amplitudes, and of other operators such as the quantized area operator on spin networks, generates a dynamical flow on these algebras, which in turn can be used to construct equilibrium states. The extremal low temperature states can be seen as a way to dynamically select certain spin foam geometries out of the parameterizing space. This approach builds on an analogy with the algebras of $\Q$-lattices up to commensurability arising in arithmetic noncommutative geometry.

\section{Spin networks and foams enriched with topological data}

The formalism of spin foams and spin networks was developed to provide a background-independent framework in loop quantum gravity. In the case of spin networks, the gravitational field on a three-dimensional manifold $M$ is encoded by a graph $\Gamma$ embedded in $M$, with representation-theoretic data attached to the edges and vertices giving the holonomies of the gravitational connection. Similarly, spin foams represent the evolution of the gravitational field along a cobordism $W$ between three-manifolds $M$ and $M'$; they are given by the geometric realization of a simplicial two-complex $\Sigma$ embedded in $W$, with similar representation-theoretic data attached to the faces and edges.

In this way, spin networks give the quantum states of three-dimensional geometries, while the spin foams give cobordisms between spin networks and are used to define partition functions and transition amplitudes as ``sums over histories'' \cite{BaezSpinFoam,RovelliQG}. The background independence then arises from the fact that, in this setting, one does not have to fix a background metric on $M$ or on $W$, and represent the gravitational field as perturbations of this fixed metric. One does, however, fix the background topology of $M$ or on $W$.

In this section we describe a way to extend the formalism of spin networks (respectively spin foams) to include ``topological data" as additional labeling on a graph embedded in the three-sphere $S^3$ (respectively, a two-complex
embedded in $S^3\times [0,1]$). These additional labels encode the topology of a three-manifold $M$ (respectively four-manifold $W$ with boundary). This is achieved by representing three-manifolds and four-manifolds as branched coverings, respectively of the three-sphere or the four-sphere, branched along an embedded graph or an embedded two-complex. This means that we need only consider graphs embedded in the three-sphere $S^3$ and two-complexes embedded in $S^3\times [0,1]$; from these we obtain both the topological information needed to construct $M$ or $W$, as well as the metric information---all from the labeling attached to faces, edges, and vertices of these simplicial data.

In essence, while the metric information is encoded in the spin network and spin foam formalism by holonomies, the topological information will be encoded similarly by monodromies.

\subsection{Spin networks}

A spin network is the mathematical representation of the quantum state of the gravitational field on a compact smooth three-dimensional manifold $M$, thought of as a three-dimensional hypersurface in a four-dimensional spacetime. In other words, spin networks should be thought of as ``quantum three-geometries." In this way, spin networks form a basis for the kinematical state space of loop quantum gravity.

Mathematically, spin networks are directed embedded graphs with edges labeled by representations of a compact Lie group, and vertices labeled by intertwiners of the adjacent edge representations. We recall the definition of spin networks given in \cite{BaezSpinFoam}.

\begin{defn} \label{def:spinNetwork}
    A spin network over a compact Lie group $G$ and embedded in a three-manifold $M$ is a triple $(\Gamma, \rho, \iota)$ consisting of:
    \begin{enumerate}
        \item   an oriented graph (one-complex) $\Gamma \subset M$;
        \item   a labeling $\rho$ of each edge $e$ of $\Gamma$ by a representation $\rho_e$ of $G$;
        \item   a labeling $\iota$ of each vertex $v$ of $\Gamma$ by an intertwiner
                \begin{equation*}
                    \iota_v : \rho_{e_1} \otimes \cdots \otimes \rho_{e_n} \to \rho_{e_1'} \otimes \cdots \otimes \rho_{e_m'},
                \end{equation*}
                where $e_1, \ldots, e_n$ are the edges incoming to $v$ and $e_1', \ldots, e_m'$ are the edges outgoing from $v$.
    \end{enumerate}
\end{defn}

\begin{center}
    \begin{figure}
        \includegraphics{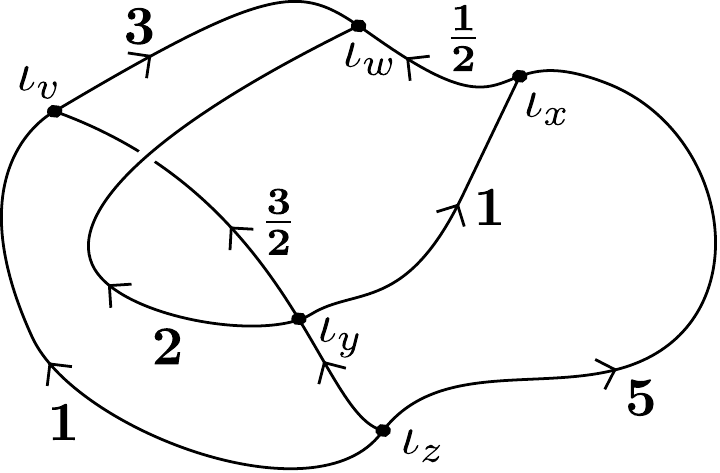}
        \caption{A sample spin network, labeled with representations of $\mathrm{SU}(2)$}
    \end{figure}
\end{center}

Notice that, in the loop quantum gravity literature, often one imposes the additional condition that the representations $\rho_e$ are irreducible. Here we take the less restrictive variant used in the physics literature, and we do not require irreducibility. Another point where there are different variants in the loop quantum gravity literature is whether the graphs $\Gamma$ should be regarded as combinatorial objects or as spatial graphs embedded in a specified three-manifold. We adopt here the convention of regarding graphs as embedded. This will be crucial in order to introduce the additional monodromy data that determine the underlying three-manifold topology, as we discuss at length in the following sections.

We can intuitively connect this to a picture of a quantum three-geometry as follows \cite{RovelliQG}. We think of such a geometry as a set of ``grains of space,'' some of which are adjacent to others. Then, each vertex of the spin network corresponds to a grain of space, while each edge corresponds to the surfaces separating two adjacent grains. The quantum state is then characterized by the quantum numbers given in our collections $\rho$ and $\iota$: in fact, the label $\iota_v$ determines the quantum number of a grain's volume, while the label $\rho_e$ determines the quantum number of the area of a separating surface. The area operator and its role in our results is is discussed further in \cref{sec:dynamicsFromQuantizedAreaOperators}.

The Hilbert space of quantum states associated to spin networks is spanned by the \textit{ambient isotopy classes} of embedded graphs $\Gamma \subset M$, with labels of edges and vertices as above; see \cite{RovelliQG} for more details. In fact, for embedded graphs, as well as for knots and links, being related by ambient isotopy is the same as being related by an orientation-preserving piecewise-linear homeomorphic change of coordinates in the ambient $S^3$, so that is the natural equivalence relation one wants to impose in the quantum gravity setting.

As in the case of knots and links, ambient isotopy is also equivalent to all planar projections being related by a generalization of Reidemeister moves: see for instance Theorem 2.1 of \cite{KauffmanInvariants}, or Theorems 1.3 and 1.7 of \cite{YetterKnottedGraphs}. These Reidemeister moves for graphs are listed in \cref{fig:reidemeisterMoves}, and discussed further in \cref{sec:wirtingerRelations}.

\begin{figure}
    \centering
    \includegraphics[scale=0.65]{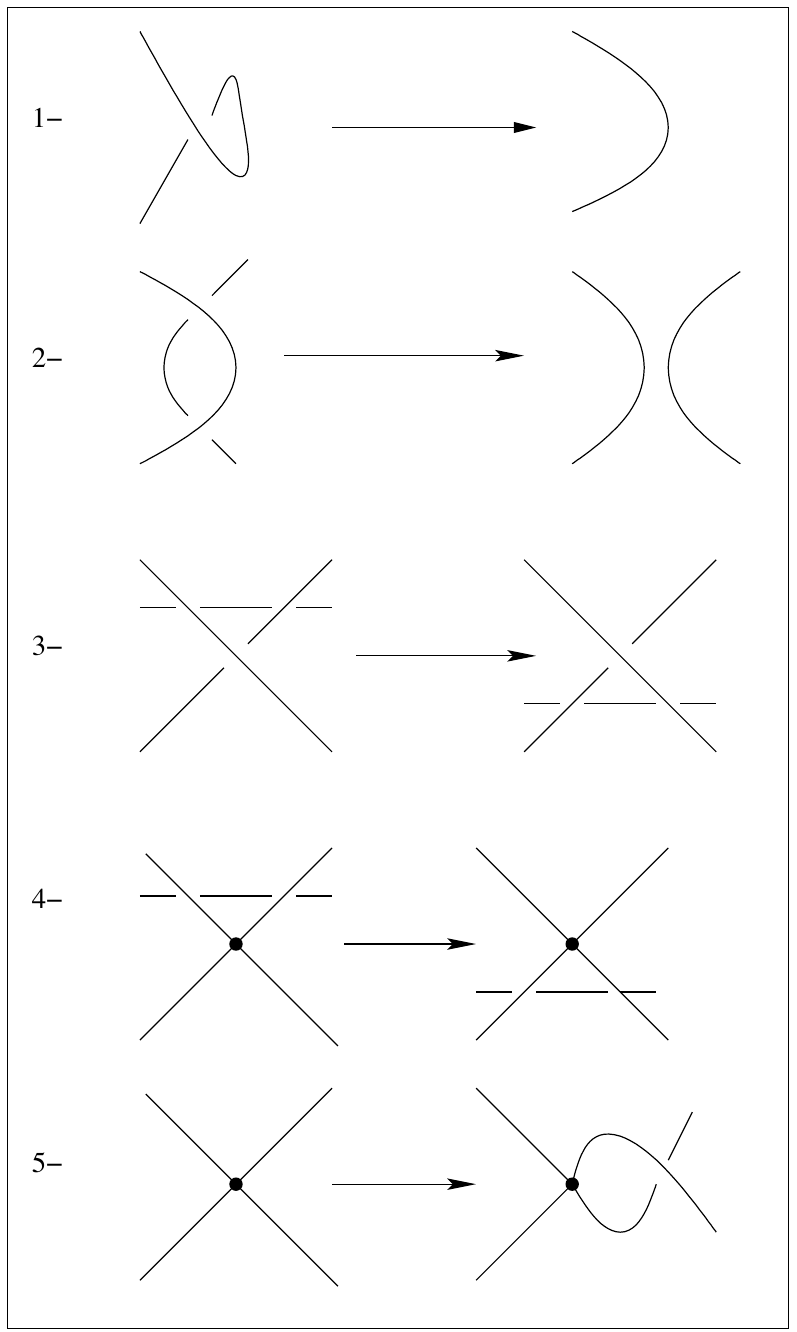}
    \caption{Reidemeister moves for embedded graphs}
    \label{fig:reidemeisterMoves}
\end{figure}

\subsection{Three-manifolds and cobordisms as branched coverings}

It is a well-known topological fact \cite{AlexanderRiemannSpaces} that every compact oriented three-manifold $M$ can be described as a branched covering of $S^3$. By a \textit{branched covering}, we mean a submersion $p : M \to S^3$ such that the restriction $p_\rvert : M \setminus p^{-1}(\Gamma) \to S^3 \setminus \Gamma$ to the complement of an embedded graph $\Gamma \subseteq S^3$ is an ordinary covering of some degree $n$. We call $\Gamma$ the \textit{branch locus}, and we say that $p$ is an order-$n$ covering of $M$, branched along $\Gamma$. The result is formulated in the piecewise-linear (PL) category, but in dimension three this is equivalent to working in the smooth category, so we will simply talk about smooth three-manifolds.

Furthermore, since the branched covering map $p$ is completely determined (up to PL homeomorphism) by $p_\rvert$, it in fact suffices to specify $\Gamma$ along with a representation $\sigma : \pi_1(S^3 \smallsetminus \Gamma) \to S_n$ of the fundamental group of the complement of the branch locus in the group of permutations on $n$ elements. The intuition is that $\Gamma$ describes where $p$ fails to be an ordinary covering, and the representation $\sigma$ completely determines how to stitch together the different branches of the covering over the branch locus \cite{AlexanderRiemannSpaces,FoxKnotTheory}. The use of the fundamental group of the complement implies that the embedding of $\Gamma$ in $S^3$ is important: that is, it is not only the structure of $\Gamma$ as an abstract combinatorial graph that matters.

Notably, the correspondence between three-manifolds and branched coverings of the three-sphere is not bijective: in general for a given manifold there will be multiple pairs $(\Gamma, \sigma)$ that realize it as a branched covering. The conditions for two pairs $(\Gamma, \sigma)$ and $(\Gamma', \sigma')$ of branching loci and fundamental group representations to give rise to the same three-manifold (up to PL homeomorphism) are discussed further in \cref{sec:coveringMoves}.

As an example of this lack of uniqueness phenomenon, the Poincar\'e homology sphere $M$ can be viewed as a fivefold covering of $S^3$ branched along the trefoil $K_{2,3}$, or as a threefold covering branched along the $(2, 5)$ torus knot $K_{2,5}$, among others \cite{PraSoIntroToInvariants}.

There is a refinement of the branched covering description of three-manifolds, the Hilden--Montesinos theorem \cite{ThreeFoldBranchedCoverings,ThreeManifoldsAsBranchedCovers}, which shows that one can in fact always realize three-manifolds as threefold branched coverings of $S^3$, branched along a knot. Although this result is much stronger, for our quantum gravity applications it is preferable to work with the weaker statement given above, with branch loci that are embedded graphs and arbitrary order of covering. We comment further in \cref{sec:bicategories} on the case of coverings branched along knots or links.

\smallskip

Another topological result we will be using substantially in the following is the analogous branched covering description for four-manifolds, according to which all compact PL four-manifolds can be realized as branched coverings of the four-sphere $S^4$, branched along an embedded simplicial two-complex. In this case also one has a stronger result \cite{FourManifoldsAsBranchedCovers}, according to which one can always realize the PL four-manifolds as fourfold coverings of $S^4$ branched along an embedded surface. However, as in the case of three-manifolds, we work with the more general and weaker statement that allows for coverings of arbitrary order, along an embedded two-complex as the branch locus. Using the fact that in dimension four there is again no substantial difference between the PL and the smooth category, in the following we work directly with smooth four-manifolds.

We also recall here the notion of \textit{branched cover cobordism} between three-manifolds realized as branched coverings of $S^3$. We use the same terminology as in \cite{CoveringsCorrespondencesNCG}.

First consider three-manifolds $M_0$ and $M_1$, each realized as a branched covering $p_i : M_i \to S^3$ branched along embedded graphs $\Gamma_i \subset S^3$. Then a branched cover cobordism is a smooth four-manifold $W$ with boundary $\partial W = M_0 \cup \bar{M}_1$ and with a branched covering map $q: W \to S^3 \times [0,1]$, branched along an embedded two-complex $\Sigma \subset S^3 \times [0,1]$, with the property that $\partial \Sigma = \Gamma_0 \cup \bar{\Gamma}_1$ and the restrictions of the covering map $q$ to $S^3 \times \{ 0 \}$ and $S^3 \times \{1 \}$ agree with the covering maps $p_0$ and $p_1$, respectively.

For the purpose of the 2-category construction we present later in the paper, we also consider cobordisms $W$ that are realized in two different ways as branched cover cobordisms between three-manifolds, as described in \cite{CoveringsCorrespondencesNCG}. This version will be used as a form of geometric correspondences providing the 2-morphisms of our 2-category.

To illustrate this latter concept, let $M_0$ and $M_1$ be two closed smooth three-manifolds, each realized in two ways as a branched covering of $S^3$ via respective branchings over embedded graphs $\Gamma_i$ and $\Gamma'_i$ for $i \in \{0,1\}$. We represent this with the notation
\begin{equation*}
    \Gamma_i \subset S^3 \stackrel{p_i}{\leftarrow} M_i \stackrel{p_i'}{\rightarrow} S^3 \supset \Gamma_i'.
\end{equation*}
A branched cover cobordism $W$ between $M_0$ and $M_1$ is a smooth four-dimensional manifold with boundary $\partial W = M_0 \cup \bar{M}_1$, which is realized in two ways as a branched cover cobordism of $S^3\times[0,1]$ branched along two-complexes $\Sigma$ and $\Sigma'$ embedded in $S^3\times[0,1]$, with respective boundaries $\partial \Sigma = \Sigma \cap (S^3\times \{0,1\}) = \Gamma_0 \cup \bar\Gamma_1$ and $\partial \Sigma' = \Gamma_0' \cup \bar\Gamma_1'$. We represent this with a similar notation,
\begin{equation*}
    \Sigma \subset S^3 \times [0,1] \stackrel{q}{\leftarrow} W \stackrel{q'}{\rightarrow} S^3\times [0,1] \supset \Sigma'.
\end{equation*}
The covering maps $q$ and $q'$ have the property that their restrictions to $S^3\times \{0\}$ agree with the maps $p_0$ and $p_0'$, respectively, while their restrictions to $S^3\times \{1\}$ agree with the maps $p_1$ and $p_1'$.

\subsection{Wirtinger relations} \label{sec:wirtingerRelations}

As we mentioned in the previous section, a branched covering $p: M\to S^3$, branched along an embedded graph $\Gamma$, is completely determined by a group representation $\sigma: \pi_1(S^3\smallsetminus \Gamma) \to S_n$. Here we pause to note that the fundamental group $\pi_1(S^3\smallsetminus \Gamma)$ of the complement of an embedded graph in the three-sphere has an explicit presentation, which is very similar to the usual Wirtinger presentation for the fundamental group of knot complements.

The advantage of describing the representation $\sigma$ in terms of an explicit presentation is that it will allow us to encode the data of the branched covering $p: M\to S^3$ completely in terms of labels attached to edges of a planar projection of the graph, with relations at vertices and crossings in the planar diagram. Furthermore, if two embedded graphs $\Gamma$ and $\Gamma'$ in $S^3$ are ambient-isotopic, then any given planar diagrams $D(\Gamma)$ and $D(\Gamma')$ differ by a finite sequence of moves that generalize to graphs the usual Reidemeister moves for knots and links, as shown in \cref{fig:reidemeisterMoves}. (For more details, see Theorem 2.1 of \cite{KauffmanInvariants} and Theorem 1.7 of \cite{YetterKnottedGraphs}.)

For the rest of this discussion, we use the following terminology. An \textit{undercrossing} in a planar diagram is the line that passes underneath at a crossing, while an overcrossing is the line that passes above the other.  Thus an arc of a planar diagram $D(\Gamma)$ is either an edge of the graph $\Gamma$, if the edge does not appear as an undercrossing in the diagram, or a half edge of $\Gamma$, when the corresponding edge is an undercrossing. Thus, an edge of $\Gamma$ always corresponds to a number $N+1$ of arcs in $D(\Gamma)$, where $N$ is the number of undercrossings that edge exhibits in the planar diagram $D(\Gamma)$. The following result is well known, but we recall it here for convenience.

\begin{lem} \label{Wirtinger}
    Let $(\Gamma, \sigma)$ be a pair of an embedded graph $\Gamma \subset S^3$ and a representation $\sigma: \pi_1(S^3\smallsetminus \Gamma) \to S_n$. Let $D(\Gamma)$ be a choice of a planar diagram for $\Gamma$. Then the representation $\sigma$ is determined by a set of permutations $\sigma_i \in S_n$ assigned to the arcs of $D(\Gamma)$, which satisfy the Wirtinger relations at crossings:
    \begin{align}
        \sigma_j & = \sigma_k \sigma_i \sigma_k^{-1}, \label{eq:wirtingerRel1}\\
        \sigma_j & = \sigma_k^{-1} \sigma_i \sigma_k. \label{eq:wirtingerRel2}
    \end{align}
    Here $\sigma_k$ is the permutation assigned to the arc of the overcrossing edge, while $\sigma_i$ and $\sigma_j$ are the permutations assigned to the two arcs of the undercrossing edge, with \cref{eq:wirtingerRel1} holding for negatively-oriented crossings and \cref{eq:wirtingerRel2} holding for positively-oriented crossings. The permutations also satisfy an additional relation at vertices, namely
    \begin{equation}
        \prod_i \sigma_i \prod_j \sigma_j^{-1} = 1, \label{eq:vertexRelation}
    \end{equation}
    where $\sigma_i$ are the permutations associated to the incoming arcs at the given vertex and $\sigma_j$ are the permutations of the outgoing arcs at the same vertex.
\end{lem}


The statement is an immediate consequence of the Wirtinger presentation of the group  $\pi_1(S^3\smallsetminus \Gamma)$, in terms of monodromies along loops around the edges, compatible with the graph orientation \cite{ClassTopAndCombinGroupTheory}. In fact, the group $\pi_1(S^3 \smallsetminus \Gamma)$ has a presentation with generators $\gamma_i$ for each strand of a planar diagram $D(\Gamma)$, modulo relations of the form \cref{eq:wirtingerRel1,eq:wirtingerRel2,eq:vertexRelation}, so that the above defines a group homomorphism $\sigma: \pi_1(S^3\smallsetminus \Gamma) \to S_n$, with $\sigma(\gamma_i) = \sigma_i$.

Notably, a similar presentation can be given for the fundamental groups $\pi_1(S^{k+2}\smallsetminus \Sigma^k)$ of the complement of a $k$-dimensional cycle embedded in the $(k+2)$-dimensional sphere, in terms of diagrams $D(\Sigma^k)$ associated to projections on an $(k+1)$-dimensional Euclidean space \cite{AlexPolynomialsAsIsoInvariants}.

\subsection{Covering moves for embedded graphs}\label{sec:coveringMoves}

As promised above, we now describe the conditions necessary for two labeled branch loci $(\Gamma, \sigma)$ and $(\Gamma', \sigma')$ to represent the same manifold (up to PL homeomorphism, or equivalently, up to diffeomorphism). The result is proved and discussed in greater detail in \cite{CoveringMovesAndKirbyCalculus}. We recall it here, since we will use it in what follows to construct our class of spin networks with topological data.

A \textit{covering move} is a non-isotopic modification of a labeled branch locus $(\Gamma, \sigma)$ representing a branched covering $p : M \to S^3$ that results in a new labeled branch locus $(\Gamma', \sigma')$ giving a different branched covering description $p' : M \to S^3$ of the same manifold $M$.

Up to stabilization, one can always assume that the two coverings have the same degree. In fact, if two labeled branch loci are of different degrees, then we can easily modify them to be of the same degree by adding trivial links to the appropriate branch locus, labeled by a transposition. Adding a sheet to the covering and an unlinked circle to the branch locus, labeled by a transposition exchanging the extra sheet with one of the others, then gives rise to a new manifold that is homeomorphic to a connected sum of the previous one with the base. The base being $S^3$, the resulting manifold is still homeomorphic to the original branched covering. Thus, we can pass to coverings of equal order by stabilization.

With this in mind, we see that to show that two arbitrary-degree labeled branch loci describe the same manifold, it suffices to have a complete and explicit description of all possible covering moves. This is done in \cite{CoveringMovesAndKirbyCalculus} in terms of the planar diagrams of the previous section, i.e.\ by describing a pair $(\Gamma, \sigma)$ by its planar diagram $D(\Gamma)$ and a set of permutations $\{\sigma_i\}$ assigned to the arcs of the diagram and satisfying the Wirtinger relations \cref{eq:wirtingerRel1,eq:wirtingerRel2,eq:vertexRelation}. All the moves are local, in the sense that they only depend on the intersection of the planar diagram with a ball, and can be performed without affecting the rest of the diagram outside of this cell. Then, as shown in \cite{CoveringMovesAndKirbyCalculus}, the following four covering moves suffice, in the sense that any two diagrams giving rise to the same three-manifold can be related by a finite sequence of these moves (after stabilization).

\begin{figure}[H]
    \centering
    \subfloat[][At vertices]
               {\label{fig:topCoveringMoves1}
                \begin{tabular}{m{0.16\textwidth} m{0.05\textwidth} m{0.16\textwidth}}
                  \includegraphics[width=0.15\textwidth]{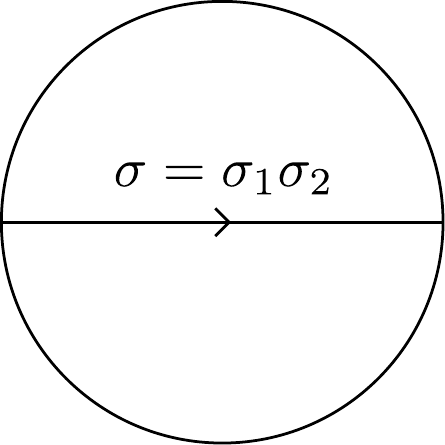} & \center{$\stackrel{V_1}{\leftrightsquigarrow}$} & \includegraphics[width=0.15\textwidth]{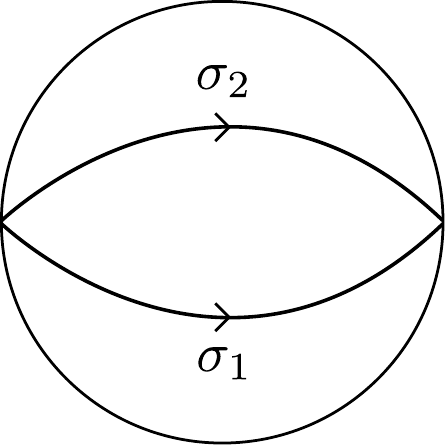} \\
                  \includegraphics[width=0.15\textwidth]{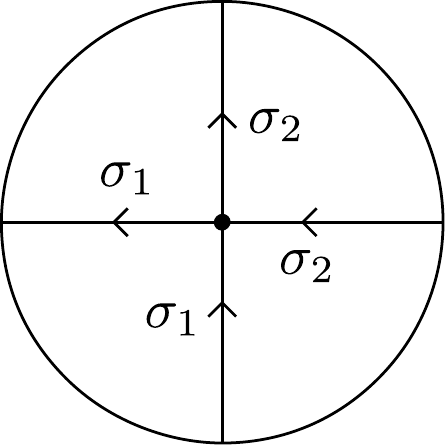} & \center{$\stackrel{V_2}{\leftrightsquigarrow}$} & \includegraphics[width=0.15\textwidth]{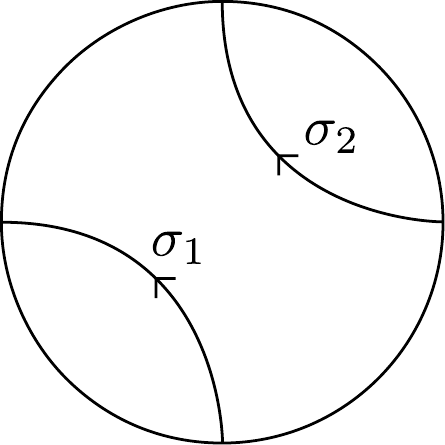}
                \end{tabular}
                }
    \hspace{0.05\textwidth}
    \subfloat[][At crossings]
               {\label{fig:topCoveringMoves2}
                \begin{tabular}{m{0.16\textwidth} m{0.05\textwidth} m{0.16\textwidth}}
                  \includegraphics[width=0.15\textwidth]{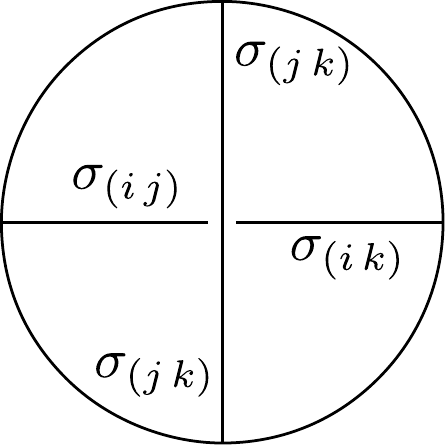} & \center{$\stackrel{C_1}{\leftrightsquigarrow}$} & \includegraphics[width=0.15\textwidth]{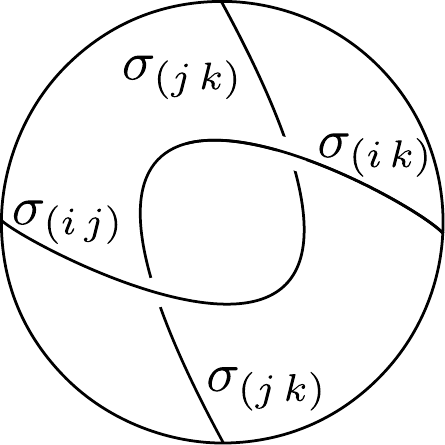} \\
                  \includegraphics[width=0.15\textwidth]{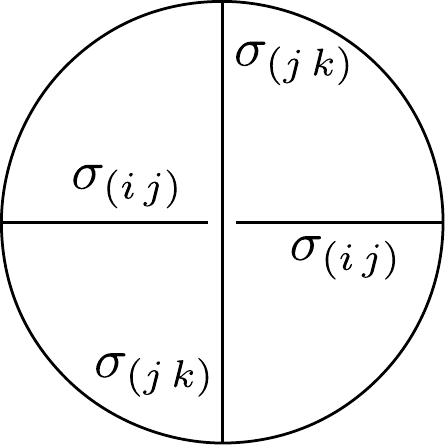} & \center{$\stackrel{C_2}{\leftrightsquigarrow}$} & \includegraphics[width=0.15\textwidth]{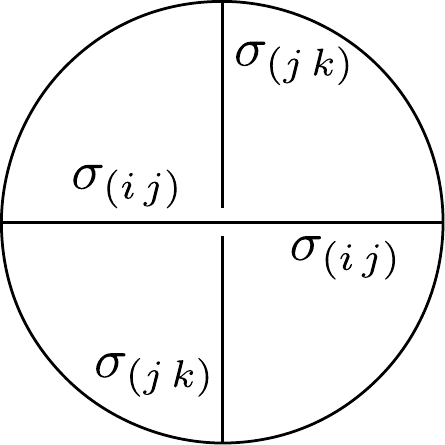}
                \end{tabular}
                }
    \caption{Topological covering moves}
\end{figure}

The covering moves of are of two types: the first two moves, those of \cref{fig:topCoveringMoves1}, involve a change in edges attached to vertices. In the move $V_1$, an edge decorated by a permutation $\sigma \in S_n$ is replaces by two parallel edges decorated by permutations $\sigma_1$ and $\sigma_2$ such that $\sigma = \sigma_1 \sigma_2$. This alters the valence of the vertices adjacent to the given edge, but the Wirtinger relations at these vertices are clearly preserved. The move $V_2$ removes a valence-four vertex with two incoming edges labeled by permutations $\sigma_1$ and $\sigma_2$ and two outgoing edges with the same labels, replacing it with only two disjoint edges with labels $\sigma_1$ and $\sigma_2$. This changes the number of vertices by one and the number of edges by two while preserving the Wirtinger relations.

The two other moves, reproduced in \cref{fig:topCoveringMoves2}, affect crossings in the planar diagram of the embedded graph, without altering the number of edges and vertices. In the move $C_1$, one considers a crossing where the two undercrossing strands are decorated by permutations $\sigma_{(i\,j)}$ and $\sigma_{(i\,k)}$ that exchange two elements, while the overcrossing is decorated by a similar permutation $\sigma_{(j\,k)}$. The move replaces the crossing by a twist, and the Wirtinger relations \cref{eq:wirtingerRel1,eq:wirtingerRel2} continue to hold since before the move we have $\sigma_{(i\,j)} = \smash{\sigma_{(j\,k)}^{-1}} \sigma_{(i\,k)} \sigma_{(j\,k)}$, while after the move we have $\sigma_{(j\,k)} = \sigma_{(i\,j)} \sigma_{(i\,k)} \smash{\sigma_{(i\,j)}^{-1}}$ and $\sigma_{(i\,j)} = \smash{\sigma_{(i\,k)}^{-1}} \sigma_{(j\,k)} \sigma_{(i\,k)}$. The move $C_2$ applies to the case of a crossing where the strands are labeled by permutations that exchange two elements and the two strands of the undercrossing carry the same label $\sigma_{(i\,j)}$, with the overcrossing strand labeled by $\sigma_{(k\,l)}$. In this case one can change the sign of the crossing. The Wirtinger relation before the move is $\sigma_{(i\,j)} = \smash{\sigma_{(k\,l)}^{-1}} \sigma_{(i\,j)} \sigma_{(k\,l)}$, while after the move it is $\sigma_{(k\,l)} = \sigma_{(i\,j)} \sigma_{(k\,l)} \smash{\sigma_{(i\,j)}^{-1}}$, and these are clearly equivalent.

\subsection{Topologically enriched spin networks}

We are now ready to define topologically enriched spin networks, or \textit{topspin networks}:

\begin{defn}\label{def:topspinNetwork}
    A topspin network over a compact Lie group $G$ is a tuple $(\Gamma, \rho, \iota, \sigma)$ of data consisting of
    \begin{enumerate}
        \item  a spin network  $(\Gamma, \rho, \iota)$ in the sense of \cref{def:spinNetwork}, with $\Gamma \subset S^3$;
        \item  a representation $\sigma : \pi_1(S^3 \smallsetminus \Gamma) \to S_n$.
    \end{enumerate}
\end{defn}

The key insight here is that the branch locus $(\Gamma, \sigma)$ corresponds, as explained above, to a unique manifold $M$. Thus, encapsulated in our topspin network is not only the geometric data given by the labelings $\rho$ and $\iota$, but also the topological data given by the representation $\sigma$. This is in stark contrast to the usual picture of (embedded) spin networks, wherein $\Gamma$ is a graph embedded into a specific manifold $M$ with fixed topology, and thus only gives geometrical data on the gravitational field.

The data of a topspin network can equivalently be described in terms of planar projections $D(\Gamma)$ of the embedded graph $\Gamma \subset S^3$, decorated with labels $\rho_i$ and $\sigma_i$ on the strands of the diagram, where the $\rho_i$ are the representations of $G$ associated to the corresponding edges of the graph $\Gamma$ and $\sigma_i \in S_n$ are permutations satisfying the Wirtinger relations \cref{eq:wirtingerRel1,eq:wirtingerRel2,eq:vertexRelation}. The vertices of the diagram $D(\Gamma)$ are decorated with the intertwiners $\iota_v$.

\smallskip

It is customary in the setting of loop quantum gravity to take graphs $\Gamma$ that arise from triangulations of three-manifolds and form directed systems associated to families of nested graphs, such as barycentric subdivisions of a triangulation. One often describes quantum operators in terms of direct limits over such families \cite{RovelliQG}.

In the setting of topspin networks described here, one can start from a graph $\Gamma$ which is, for instance, a triangulation of the three-sphere $S^3$, which contains as a subgraph $\Gamma' \subset \Gamma$ the branch locus of a branched covering describing a three-manifold $M$. By pullback, one obtains from it a triangulation of the three-manifold $M$. Viewed as a topspin network, the diagrams $D(\Gamma)$ will carry nontrivial topological labels $\sigma_i$ on the strands belonging to the subdiagram $D(\Gamma')$ while the strands in the rest of the diagram are decorated with $\sigma_i=1$; all the strands can carry nontrivial $\rho_i$. In the case of a barycentric subdivision, the new edges belonging to the same edge before subdivision maintain the same labels $\sigma_i$, while the new edges in the barycentric subdivision that do not come from edges of the previous graph carry trivial topological labels. Thus, all the arguments usually carried out in terms of direct limits and nested subgraphs can be adapted to the case of topspin networks without change. Working with data of graphs $\Gamma$ containing the branch locus $\Gamma'$ as a subgraph is also very natural in terms of the fibered product composition we describe in \cref{fiberSec}, based on the construction of \cite{CoveringsCorrespondencesNCG}.

\smallskip

Also in loop quantum gravity, one considers a Hilbert space generated by the spin networks \cite{RovelliQG}, where the embedded graphs are taken up to ambient isotopy, or equivalently up to a PL change of coordinates in the ambient $S^3$. (We deal here only with spin networks embedded in $S^3$, since as explained above this is sufficiently general after we add the topological data $\sigma$.) This can be extended to a Hilbert space $\cH$ of topspin networks by requiring that two topspin networks are orthogonal if they are not describing the same three-manifold, that is, if they are not related (after stabilization) by covering moves, and by defining the inner product $\langle \psi, \psi'\rangle$ of  topspin networks that are equivalent under covering moves to be the usual inner product of the underlying spin networks obtained by forgetting the presence of the additional topological data $\sigma$, $\sigma'$.

\subsection{Geometric covering moves}\label{sec:consistencyConditions}

The addition of topological data $\sigma$ to spin networks, on top of the preexisting geometric labelings $\rho$ and $\iota$, necessitates that we ensure these two types of data are compatible. The essential issue is the previously-discussed fact that a given three-manifold $M$ can have multiple descriptions as a labeled branch locus, say as both $(\Gamma, \sigma)$ and $(\Gamma', \sigma')$. If $(\Gamma, \rho, \iota, \sigma)$ is a topspin network representing a certain geometric configuration over $M$, then we would like to be able to say when another topspin network $(\Gamma', \rho', \iota', \sigma')$, also corresponding to the manifold $M$ via the data $(\Gamma', \sigma')$, represents the same geometric configuration. That is, what are the conditions relating $\rho'$ and $\iota'$ to $\rho$ and $\iota$ that ensure geometric equivalence?

Since $(\Gamma, \sigma)$ and $(\Gamma', \sigma')$ represent the same manifold, they can be related by the covering moves of \cref{sec:coveringMoves}. These covering moves can be interpreted as answering the question ``if one makes a local change to the graph, how must the topological labeling change?" In this section, our task is to answer a very similar question, namely, ``if one makes a local change to the graph, how must the geometric labeling change?" Phrased this way, it is easy to see that we simply need to give an account of what happens to the geometric labelings under the same covering moves as before.

The first thing to note is that, since in \cref{def:topspinNetwork} we do not demand that the edge-labels are given by \textit{irreducible} representations, there is a certain type of trivial equivalence between different geometric labelings that emerges. To see this, consider the following two very simple spin networks:
\begin{figure}[H]
    \begin{tabular}{>{\raggedleft}m{0.4\textwidth} m{0.05\textwidth} m{0.4\textwidth}}
        \includegraphics[width=0.25\textwidth]{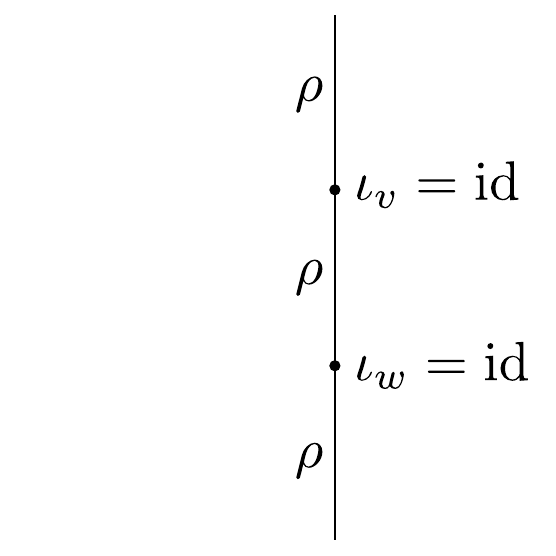} & \center{\text{vs.}} & \includegraphics[width=0.25\textwidth]{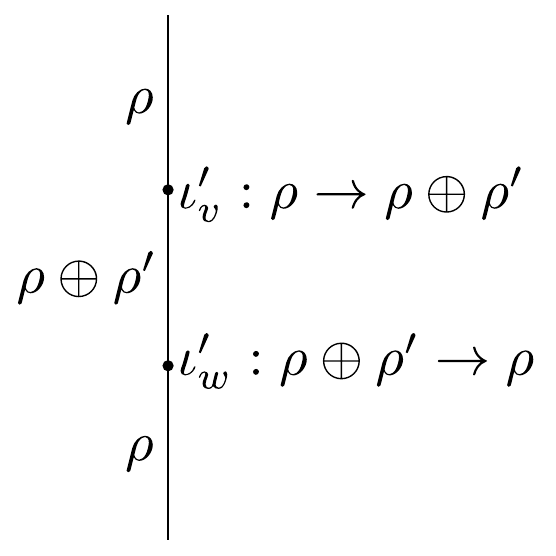}
    \end{tabular}
\end{figure}
Now, if we demand that $\iota'_w \circ \iota'_v = \iota_w \circ \iota_v = \id$, then these networks are essentially the same. The addition of $\rho'$ to the middle edge, and the compensating changes to the intertwiners, added no geometric meaning: passing through the two vertices in sequence still gives the same result. In the same way, if we have two topologically-equivalent topspin networks related by their covering moves, we could always obtain geometric equivalence by falling back on the demand that the composition of the relevant intertwiners is the same. For example, consider the move $V_1$, which, with geometric labels associated to it, is as follows:
\begin{figure}[H]
    \begin{tabular}{>{\raggedleft}m{0.4\textwidth} m{0.05\textwidth} m{0.4\textwidth}}
        \includegraphics[width=0.25\textwidth]{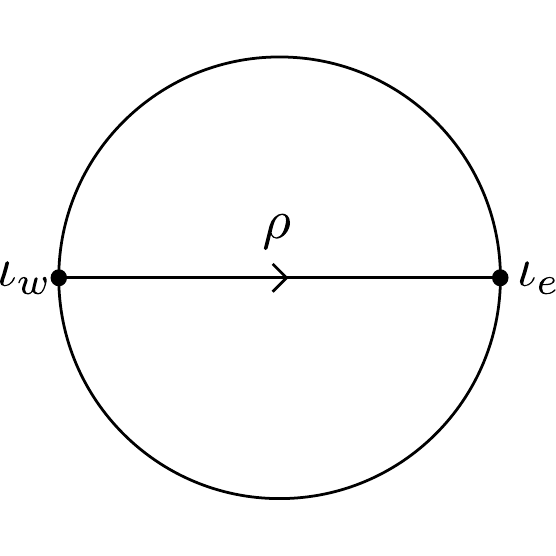} & \center{$\stackrel{V_1}{\leftrightsquigarrow}$} & \includegraphics[width=0.25\textwidth]{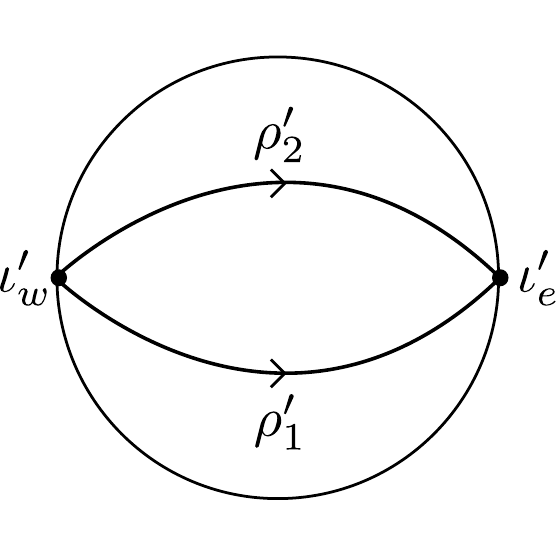}
    \end{tabular}
\end{figure}
The simplest condition we could impose on $\rho'$ and $\iota'$ is only that the following diagram commutes:
\begin{diagram}
\xi_w & \rTo^{\iota_w}        & \rho                    & \rTo^{\iota_e}        & \xi_e \\
      & \rdTo(2,1)<{\iota'_w} & \rho'_1 \otimes \rho'_2 & \ruTo(2,1)>{\iota'_e} &
\end{diagram}
where we denote by $\xi_w$ and $\xi_e$, respectively, the tensor product of the representations associated to all the edges outside of the cell that connect to the vertices $w$ and $e$.

This, however, is not very enlightening. We can gain more insight into the geometric structure under covering moves by discarding this trivial equivalence from our consideration. This can be done either by working only with irreducible representations, or equivalently by requiring that the vertex intertwiners stay the same under covering moves.

\medskip

With this in mind, consider again the move $V_1$ depicted above. If we demand that the intertwiners stay the same on each side, i.e.\ that $\iota'_w = \iota_w$ and $\iota'_e = \iota_e$, then we obtain the requirement that $\rho = \rho'_1 \otimes \rho'_2$ if the inputs and outputs of each diagram's intertwiners are to match. In other words, we can present the geometric covering move as follows:
\begin{figure}[H]
    \begin{tabular}{>{\raggedleft}m{0.4\textwidth} m{0.05\textwidth} m{0.4\textwidth}}
        \includegraphics[width=0.25\textwidth]{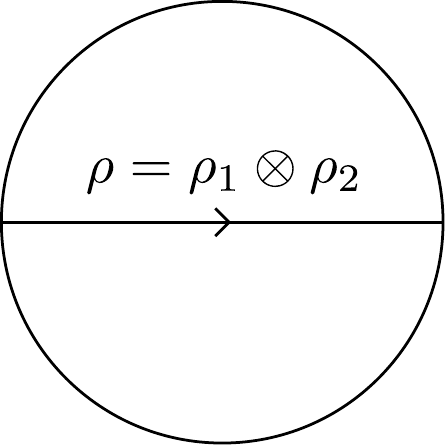} & \center{$\stackrel{V_1}{\leftrightsquigarrow}$} & \includegraphics[width=0.25\textwidth]{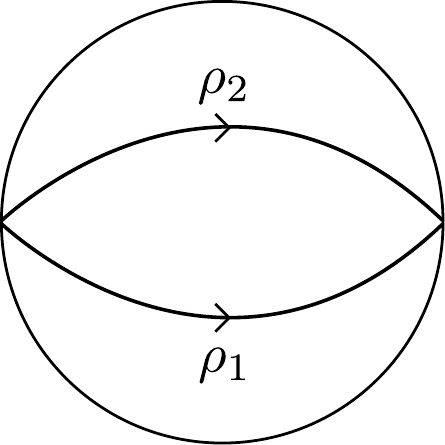}
    \end{tabular}
\end{figure}
The move $V_1$ thus allows us to split a tensor-product edge label into two edges labeled by the factors.

\smallskip

In the same way, we can analyze the move $V_2$, which is given with as-yet-unrelated geometric labels as follows:
\begin{figure}[H]
    \begin{tabular}{>{\raggedleft}m{0.4\textwidth} m{0.05\textwidth} m{0.4\textwidth}}
        \includegraphics[width=0.25\textwidth]{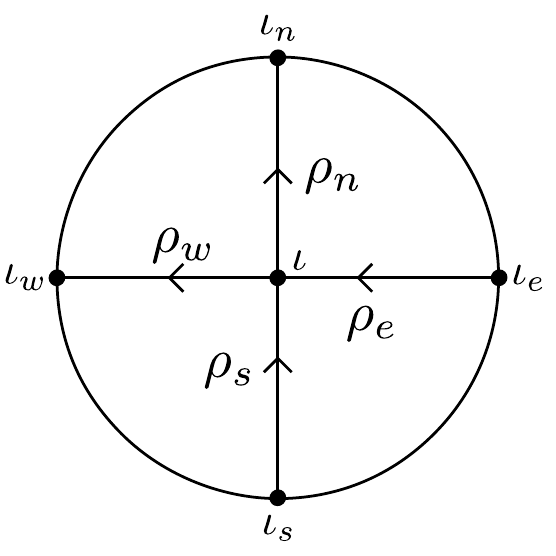} & \center{$\stackrel{V_2}{\leftrightsquigarrow}$} & \includegraphics[width=0.25\textwidth]{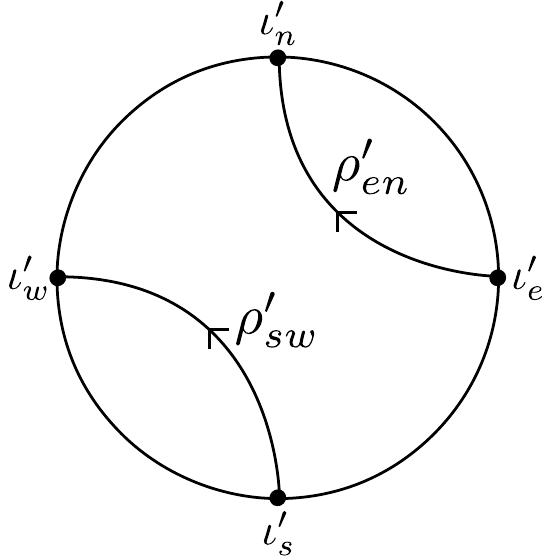}
    \end{tabular}
\end{figure}
By demanding equality of the intertwiners, we see that since $\iota_w$ expects $\rho_w$ as input, we must have that $\rho'_{sw}$, the input to $\iota'_w$, is equal to $\rho_w$. In the same way, considering the output of $\iota_s$ versus $\iota'_s$, we obtain $\rho'_{sw} = \rho_s$. Similar arguments for the other two intertwiners and their counterparts give us $\rho'_{en} = \rho_e = \rho_n$. Finally, under equality of primed and unprimed intertwiners, requiring that the composition of intertwiners along the possible paths is the same, i.e.\ that $\iota'_w \circ \iota'_s = \iota_w \circ \iota \circ \iota_s$ and $\iota'_n \circ \iota'_e = \iota_n \circ \iota \circ \iota_e$, gives us the requirement that $\iota = \id$. The resulting geometric covering move is then given by
\begin{figure}[H]
    \begin{tabular}{>{\raggedleft}m{0.4\textwidth} m{0.05\textwidth} m{0.4\textwidth}}
        \includegraphics[width=0.25\textwidth]{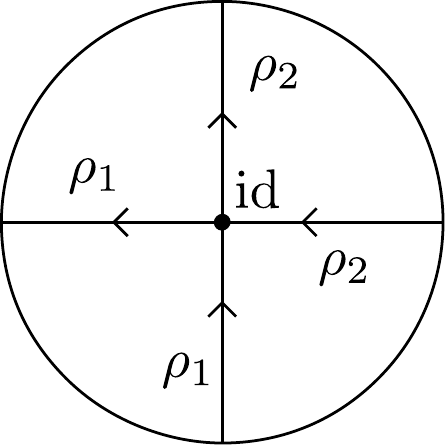} & \center{$\stackrel{V_2}{\leftrightsquigarrow}$} & \includegraphics[width=0.25\textwidth]{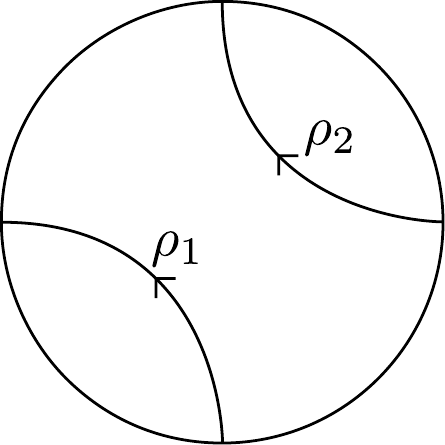}
    \end{tabular}
\end{figure}
Thus the move $V_2$ simply says that we can remove an identity intertwiner.

\medskip

The covering moves that change relations at crossings are simpler to analyze. For topological labels, these moves are interesting since we associate a topological label to each strand. In contrast, geometric labels are associated to each \textit{edge}: thus if a given edge has two strands, both strands must have the same geometric labels in any given diagram. In this way, the covering moves that modify crossings will have a much less interesting effect on the geometric labels than they do on the topological ones.

\smallskip

More concretely, let us consider move $C_1$, as usual letting the geometric labels on each side be entirely general for now:
\begin{figure}[H]
    \begin{tabular}{>{\raggedleft}m{0.4\textwidth} m{0.05\textwidth} m{0.4\textwidth}}
        \includegraphics[width=0.25\textwidth]{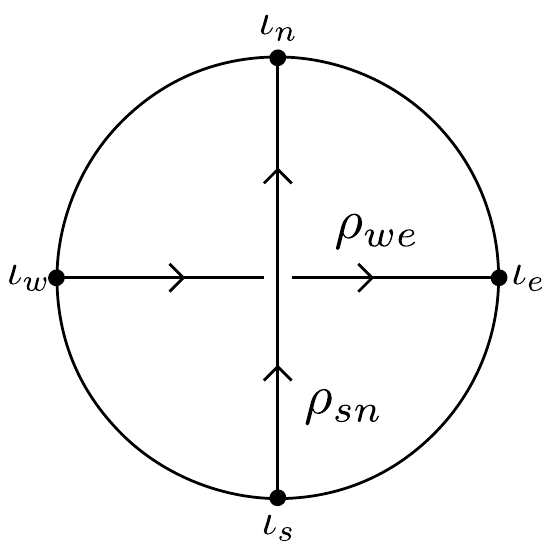} & \center{$\stackrel{C_1}{\leftrightsquigarrow}$} & \includegraphics[width=0.25\textwidth]{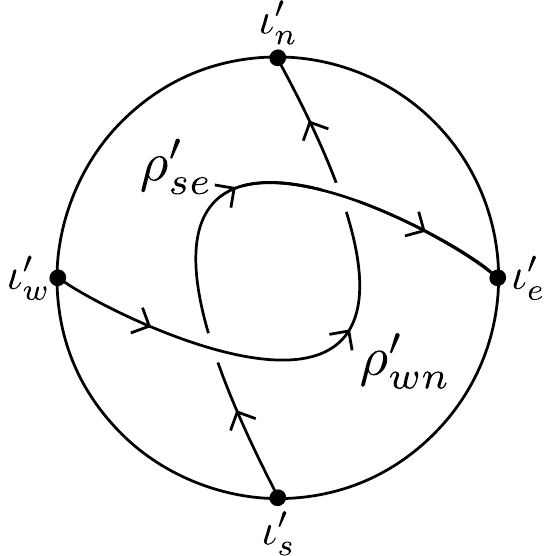}
    \end{tabular}
\end{figure}
Using similar reasoning to the above cases, it is clear that under equality of primed and unprimed intertwiners, such a move will only be permissible if $\rho_{we} = \rho_{se} = \rho'_{se} = \rho'_{wn}$. The geometric covering move is thus given by
\begin{figure}[H]
    \begin{tabular}{>{\raggedleft}m{0.4\textwidth} m{0.05\textwidth} m{0.4\textwidth}}
        \includegraphics[width=0.25\textwidth]{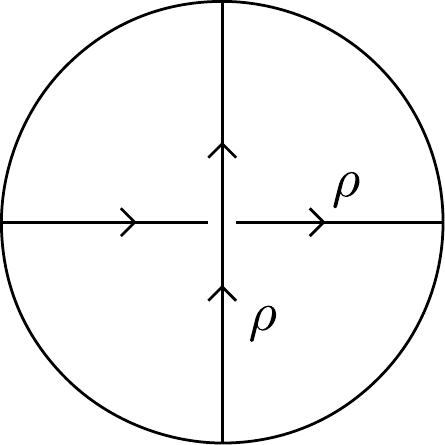} & \center{$\stackrel{C_1}{\leftrightsquigarrow}$} & \includegraphics[width=0.25\textwidth]{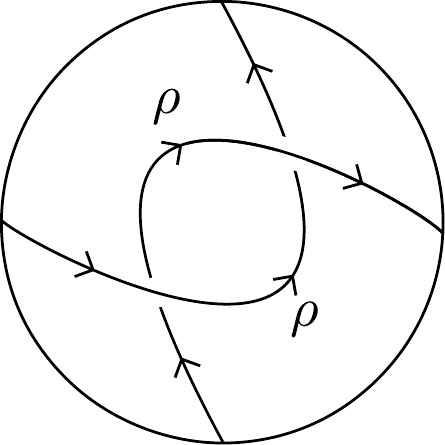}
    \end{tabular}
\end{figure}
In this way, $C_1$ tells us that if two edges labeled by the same representation cross, we can ``redirect'' them by sending each along the other's former path.

\smallskip

The final move, $C_2$, actually has no geometric content: it only changes the sign of the crossing, and as discussed above, only the topological labels are allowed to depend on such features. There are thus no relations between or restrictions on the edge labels. For completeness, this is shown diagrammatically as
\begin{figure}[H]
    \begin{tabular}{>{\raggedleft}m{0.4\textwidth} m{0.05\textwidth} m{0.4\textwidth}}
        \includegraphics[width=0.25\textwidth]{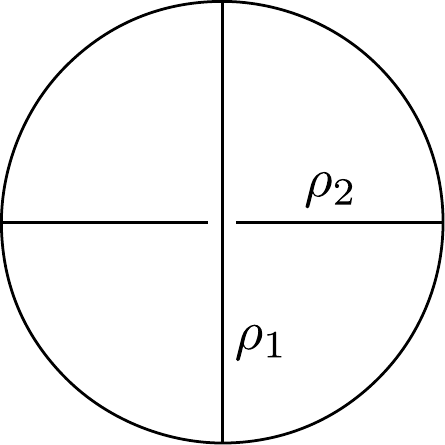} & \center{$\stackrel{C_1}{\leftrightsquigarrow}$} & \includegraphics[width=0.25\textwidth]{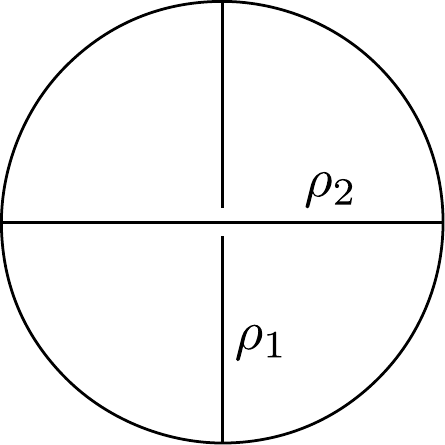}
    \end{tabular}
\end{figure}

\medskip

With these in hand, we now have a complete catalogue of the ways in which two different labeled graphs can represent the same topspin network. That is, equivalence holds if $(\Gamma, \rho, \iota, \sigma)$ can be converted into $(\Gamma', \rho', \iota, \sigma')$ by a finite sequence of covering moves, such that the relations between the topological labels before and after each move are as described in \cref{sec:coveringMoves}, and those between the geometric labels are as described in this section. Further ``trivial'' modifications, in the sense discussed in the opening to this section, can be made to the geometric data as well---but these are not terribly interesting, being on the same level as e.g.\ adding a valence-two vertex with identity intertwiner somewhere in the graph.

\begin{figure}[H]
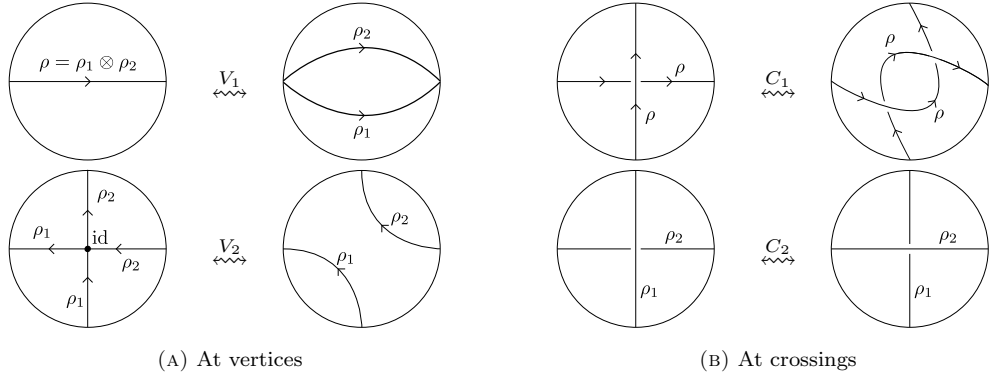

    \centering
    \subfloat[][At vertices]
               {\label{fig:coveringMoves1}
                \begin{tabular}{m{0.16\textwidth} m{0.05\textwidth} m{0.16\textwidth}}
                  \includegraphics[width=0.15\textwidth]{CompatibilityV11Final} & \center{$\stackrel{V_1}{\leftrightsquigarrow}$} & \includegraphics[width=0.15\textwidth]{CompatibilityV12Final} \\
                  \includegraphics[width=0.15\textwidth]{CompatibilityV21Final} & \center{$\stackrel{V_2}{\leftrightsquigarrow}$} & \includegraphics[width=0.15\textwidth]{CompatibilityV22Final} \\
                \end{tabular}
                }
    \hspace{0.05\textwidth}
    \subfloat[][At crossings]
               {\label{fig:coveringMoves2}
                \begin{tabular}{m{0.16\textwidth} m{0.05\textwidth} m{0.16\textwidth}}
                  \includegraphics[width=0.15\textwidth]{CompatibilityC11Final} & \center{$\stackrel{C_1}{\leftrightsquigarrow}$} & \includegraphics[width=0.15\textwidth]{CompatibilityC12Final} \\
                  \includegraphics[width=0.15\textwidth]{CompatibilityC21Final} & \center{$\stackrel{C_2}{\leftrightsquigarrow}$} & \includegraphics[width=0.15\textwidth]{CompatibilityC22Final} \\
                \end{tabular}
                }
    \caption{Geometric covering moves}
\end{figure}

\subsection{Spin foams and topological enrichment}

Spin foams are the natural extensions of spin networks constructed by ``going up a dimension''; that is, we have a finite collection of polygons attached to each other along their edges, with faces labeled by representations and edges labeled by intertwining operators. The construction is such that taking a ``slice'' of the spin foam at a given ``time'' will then produce a spin network. With this in mind, it is intuitively clear how spin foams encode the \textit{dynamics} of loop quantum gravity, while spin networks give the kinematics.

In the following, by a two-complex we mean a simplicial complex with two-dimensional faces, one-dimensional edges, and zero-dimensional vertices, endowed with the usual boundary operator $\partial$, which assigns to a face the formal sum of its boundary edges with positive or negative sign according to whether the induced orientation from the face agrees or not with the orientation of the corresponding edge. An embedded two-complex is a PL embedding of the geometric realization of the two-complex in a PL four-manifold. For a face $f$ in such a two-complex, we let $\partial(f)$ denote the boundary of the face, in the sense described above. We also denote an edge's negative by $\bar{e}$, instead of by $-e$.

Spin foams are then defined by the following data, where we again follow \cite{BaezSpinFoam}.

\begin{defn}
    \label{def:spinFoam}
    Suppose that $\psi = (\Gamma, \rho, \iota)$ and $\psi' = (\Gamma', \rho', \iota')$ are spin networks over $G$, with the graphs $\Gamma$ and $\Gamma'$ embedded in three-manifolds $M$ and $M'$, respectively. A spin foam $\Psi : \psi \to \psi'$ embedded in a cobordism $W$ with $\partial W = M \cup \bar M'$ is then a triple $\Psi = (\Sigma, \tilde\rho, \tilde\iota)$
    consisting of:
    \begin{enumerate}
        \item   an oriented two-complex $\Sigma \subseteq W$, such that $\Gamma \cup \bar \Gamma'$ borders $\Sigma$: that is, $\partial \Sigma = \Gamma \cup \bar \Gamma'$, and there exist cylinder neighborhoods $M_\epsilon = M \times [0,\epsilon)$ and $\bar M'_\epsilon = \bar M' \times (-\epsilon,0]$ in $W$ such that $\Sigma \cap M_\epsilon = \Gamma \times [0,\epsilon)$ and $\Sigma \cap \bar M'_\epsilon = \bar \Gamma' \times (-\epsilon,0]$;
        \item   a labeling $\tilde\rho$ of each face $f$ of $\Sigma$ by a representation $\tilde\rho_f$ of $G$;
        \item   a labeling $\tilde\iota$ of each edge $e$ of $\Sigma$ that does not lie in $\Gamma$ or $\Gamma'$ by an intertwiner
                \begin{equation*}
                    \tilde\iota_e : \bigotimes_{f : e \in \partial(f)} \tilde\rho_f \to \bigotimes_{f' : \bar e \in \partial(f')}
                    \tilde\rho_{f'}
                \end{equation*}
    \end{enumerate}
    with the additional consistency conditions
    \begin{enumerate}
        \item   for any edge $e$ in $\Gamma$, letting $f_e$ be the face in $\Sigma\cap M_\epsilon$ bordered by $e$, we must have that $\tilde\rho_{f_e} = (\rho_e)^*$ if $e\in \partial(f_e)$, or $\tilde\rho_{f_e} = \rho_e$ if $\bar e \in \partial(f_e)$;
        \item   for any vertex $v$ of $\Gamma$, letting $e_v$ be the edge in $\Sigma\cap M_\epsilon$ adjacent to $v$, we must have that $\tilde\iota_{e_v} = \iota_v$ after appropriate dualizations to account for the different orientations of faces and edges as above;
        \item   dual conditions must hold for edges and vertices from $\Gamma'$ and faces and edges in $\Sigma \cap M_\epsilon'$, i.e.\ we would have $\tilde\rho_{f_e} = \rho_e$ if $e \in \partial(f_e)$ and $\tilde\rho_{f_e} = (\rho_e)^*$ if $\bar e \in \partial(f_e)$, and likewise for vertices.
    \end{enumerate}
\end{defn}

Thus, spin foams are cobordisms between spin networks, with compatible labeling of the edges, vertices, and faces. They represent quantized four-dimensional geometries inside a fixed smooth four-manifold $W$.

Notice that, in the context of spin foams, there are some variants regarding what notion of cobordisms between embedded graphs one considers. A more detailed discussion of cobordisms of embedded graphs is given in \cref{cobordSec} below.

When working with topspin networks, one can correspondingly modify the notion of spin foams to provide cobordisms compatible with the topological data, in such a way that one no longer has to specify the four-manifold $W$ with $\partial W = M\cup \bar M'$ in advance: one can work with two-complexes $\Sigma$ embedded in the trivial cobordism $S^3 \times [0,1]$ and the topology of $W$ is encoded as part of the data of a \textit{topspin foam}, just as the  three-manifolds $M$ and $M'$ are encoded in their respective topspin networks.

In the following we say that a three-manifold $M$ corresponds to $(\Gamma,\sigma)$ if it is the branched cover of $S^3$ determined by the representation $\sigma: \pi_1(S^3\smallsetminus \Gamma) \to S_n$, and similarly for a four-manifold $W$ corresponding to data $(\Sigma,\tilde\sigma)$.

\begin{defn}
    \label{def:topspinFoam}
    Suppose that $\psi=(\Gamma,\rho,\iota,\sigma)$ and $\psi'=(\Gamma',\rho',\iota',\sigma')$ are topspin networks over $G$, with topological labels in the same permutation group $S_n$ and with $\Gamma,\Gamma'\subset S^3$. A topspin foam $\Psi: \psi \to \psi'$ over $G$ is a tuple $\Psi=(\Sigma,\tilde\rho,\tilde\iota,\tilde\sigma)$ of data consisting of
    \begin{enumerate}
        \item a spin foam $(\Sigma, \tilde\rho, \tilde\iota)$ between $\psi$ and $\psi'$ in the sense of \cref{def:spinFoam}, with $\Sigma\subset S^3\times [0,1]$;
        \item a representation $\tilde\sigma: \pi_1((S^3\times[0,1])\smallsetminus \Sigma) \to S_n$, such that the manifold $W$ corresponding to the labeled branch locus $(\Sigma, \tilde\sigma)$ is a branched cover cobordism between $M$ and $M'$, where the three-manifolds $M$ and $M'$ are those corresponding to the labeled branch loci $(\Gamma, \sigma)$ and $(\Gamma', \sigma')$ respectively.
    \end{enumerate}
\end{defn}

We recall briefly the analog of the Wirtinger relations for a two-complex $\Sigma$ embedded in $S^4$ (or in $S^3 \times [0,1]$ as in our case below), which are slight variants on the ones given in \cite{AlexPolynomialsAsIsoInvariants}.  One considers a diagram $D(\Sigma)$ obtained from a general projection of the embedded two-complex $\Sigma \subset S^4$ (which one can assume is in fact embedded in $\R^4 = S^4 \setminus \{ \infty \}$) onto a three-dimensional linear subspace $L \subset \R^4$.  The complement of the projection of $\Sigma$ in three-dimensional space $L$ is a union of connected components $L_0\cup \cdots \cup L_N$. One chooses then a point $x_k$ in each component $L_k$ and two points $p$ and $q$ in the two components of $\R^4 \setminus L$.  Denote by $f_i$ the strands of two-dimensional faces in the planar diagram $D(\Sigma)$. Each face of $\Sigma$ corresponds to one or more strands $f_i$ in $D(\Sigma)$ according to the number of undercrossings the projection of the face acquires in the diagram. For each strand $f_\alpha$ one considers a closed curve $\gamma_i = \ell_{p x_i} \cup \ell_{x_i q} \cup \ell_{q x_{i+1}}\cup \ell_{x_{i+1} p}$, where $\ell_{xy}$ denotes a smooth embedded arc in $\R^4$ with endpoints $x$ and $y$, and $x_i$ and $x_{i+1}$ denote the chosen points in the two components of $L \setminus D(\Sigma)$ with the face $f_i$ in their common boundary. We assume that the arcs do not intersect the segments connecting points of $\Sigma$ to their projection on $L$. The curves $\gamma_i$ generate $\pi_1(\R^4 \smallsetminus \Sigma)$. To give then the analog of the Wirtinger relations, one considers overcrossings and undercrossings of faces in the planar diagram $D(\Sigma)$ as well as edges that lie at the intersection of faces. Each crossing gives a relation analogous to the relations for crossings in the case of embedded graphs, namely the generators of the fundamental group associated to the four spatial regions of $L \setminus D(\Sigma)$ surrounding the crossing satisfy
\begin{equation}
    \label{Wirt2D1}
    \gamma_j = \gamma_k \gamma_i \gamma_k^{-1}
    \qquad \text{or}\qquad
    \gamma_j = \gamma_k^{-1} \gamma_i \gamma_k,
\end{equation}
depending on the relative orientations, while at a common edge we have
\begin{equation}
    \label{Wirt2D2}
    \prod \gamma_i ^{\pm 1} =1,
\end{equation}
where the product is over all the regions of $L \setminus D(\Sigma)$ surrounding the edge and bounded by the faces adjacent to the edge, and the $\pm1$ power depends on the relative orientation of the boundaries of the faces and the edge. It then follows that a representation $\tilde \sigma: \pi_1((S^3\times[0,1]) \smallsetminus \Sigma)\to S_n$ is determined by assigning permutations $\tilde\sigma_i$ to the strands in a diagram $D(\Sigma)$ corresponding to the faces of $\Sigma$, with relations as above at the crossings and at edges in the common boundary of different faces. In fact, each permutation $\tilde\sigma_f$ represents the monodromy around the curve $\gamma_i$ linking the strands $f_i$ in the diagram.

\smallskip

The specification of a topspin network, according to \cref{def:topspinFoam}, can then be rephrased in terms of such three-dimensional diagrams in the following way:

\begin{lem}\label{topspinfoamD}
    A topspin foam $\Psi = (\Sigma,\tilde\rho,\tilde\iota,\tilde\sigma)$ over $G$ can be specified by a diagram $D(\Sigma)$ obtained from a three-dimensional projection as above, decorated by
    \begin{enumerate}
        \item   assigning to each one-dimensional strand $e_i$ of $D(\Sigma)$ the same intertwiner $\tilde\iota_e$ assigned to the edge $e$;
        \item   assigning to each two-dimensional strand $f_\alpha$ of $D(\Sigma)$ the same representation $\tilde\rho_f$ of $G$ assigned to the face $f$;
        \item   assigning to each two-dimensional strand $f_\alpha$ of $D(\Sigma)$ a topological label $\tilde\sigma_\alpha \in S_n$ such that taken in total such assignments satisfy the Wirtinger relations
            \begin{align}
                \tilde\sigma_\alpha & = \tilde\sigma_\beta \tilde\sigma_{\alpha'}\tilde\sigma_\beta^{-1}, \label{eq:2DwirtingerRel1}\\
                \tilde\sigma_{\alpha'} & = \tilde\sigma_\beta^{-1} \tilde\sigma_\alpha \tilde\sigma_\beta, \label{eq:2DwirtingerRel2}
            \end{align}
            where $\tilde\sigma_\beta$ is the permutation assigned to the arc of the overcrossing face and $\tilde\sigma_\alpha$ and $\tilde\sigma_{\alpha'}$ are those assigned to the two arcs of the undercrossing face (with the appropriate equation depending on the orientation of the crossing), along with the additional relation at edges in the boundary of different faces
            \begin{equation}
                \label{reledge2D}
                \prod_{\alpha  : e       \in \partial(f_\alpha)   } \hspace{-1em}   \tilde\sigma_\alpha \
                \prod_{\alpha' : \bar{e} \in \partial(f_{\alpha'})} \hspace{-1.3em} \tilde\sigma_{\alpha'}^{-1} = 1.
            \end{equation}
    \end{enumerate}
    The permutations $\tilde\sigma_\alpha$ have the property that, for any strand $e_i$ in the diagram $D(\Gamma)$, letting $f_{\alpha_i}$ be the face in $D(\Sigma)$ bordered by $e_i$, then $\tilde\sigma_{\alpha_i} = (\sigma_i)^*$ if $e_i\in \partial(f_\alpha)$, or $\tilde\sigma_{\alpha_i} = \sigma_i$ if $\bar e_i \in \partial(f_\alpha)$, where $\sigma_i$ are the permutations that label the strands of the diagram $D(\Gamma)$ of the topspin network $\psi$. Similarly, for strands $e_i'$ of $D(\Gamma')$, one has $\tilde\sigma_{\alpha_i} = \sigma'_i$ if $e_i'\in \partial(f_\alpha)$, or $\tilde\sigma_{\alpha_i} = (\sigma'_i)^*$ if $\bar e_i' \in \partial(f_\alpha)$.
\end{lem}

\begin{proof}
    We can consider $\Sigma \subset S^3\times [0,1]$ as embedded in $S^4$, with the four-sphere obtained by cupping $\Sigma \subset S^3\times [0,1]$ with two 3-balls $D^3$ glued along the two boundary components $S^3\times \{0 \}$ and $S^3\times \{1\}$. By removing the point at infinity, we can then think of $\Sigma$ as embedded in $\R^4$ and obtain diagrams $D(\Sigma)$ by projections on generic three-dimensional linear subspaces in $\R^4$, where one marks by overcrossings and undercrossings the strands of $\Gamma$ that intersect in the projection, as in the case of embedded knots and graphs.

    We can then use the presentation of the fundamental group of the complement $\pi_1(S^4 \smallsetminus \Sigma)$ given in terms of generators and Wirtinger relations as above. This shows that the data $\tilde\sigma_\alpha$ define a representation $\tilde\sigma:\pi_1(S^4 \smallsetminus \Sigma) \to S_n$, hence an $n$-fold branched covering of $S^4$ branched along $\Sigma$. The compatibility with $\tilde\sigma_\alpha$ and the $\sigma_i$ and $\sigma_i'$ show that the restriction of this branched covering to $S^3 \times [0,1]$ determines a branched cover cobordism between the covering $M$ and $M'$ of the three-sphere, respectively determined by the representations $\sigma: \pi_1(S^3 \smallsetminus \Gamma) \to S_n$ and $\sigma': \pi_1(S^3 \smallsetminus \Gamma') \to S_n$. We use here the fact that, near the boundary, the embedded two-complex $\Sigma$ is a product $\Gamma \times [0,\epsilon)$ or $(-\epsilon, 0] \times \Gamma'$.
\end{proof}

\subsection{The case of cyclic coverings} \label{cyclcovSec}

A cyclic branched covering $M$ of $S^3$, branched along $\Gamma \subset S^3$, is a branched covering such that the corresponding representation $\sigma$ maps surjectively
\begin{equation}\label{cycliccover}
    \sigma : \pi_1(S^3\smallsetminus \Gamma) \to \Z/n\Z.
\end{equation}
Topspin networks whose topological data define cyclic branched coverings are simpler than the general case discussed above, in the sense that the topological data are directly associated to the graph $\Gamma$ itself, not to the planar projections $D(\Gamma)$.

\begin{prop}\label{cycltopnet}
    Let $\psi=(\Gamma, \rho,\iota,\sigma)$ be a topspin network such that the branched cover $M$ of $S^3$ determined by $(\Gamma,\sigma)$ is cyclic. Then the data $\psi=(\Gamma, \rho,\iota,\sigma)$ are equivalent to a spin networks $(\Gamma, \rho,\iota)$ together with group elements $\sigma_e\in \Z/n\Z$ associated to the edges of $\Gamma$ satisfying the relation
    \begin{equation}\label{Wirtcycl}
        \prod_{s(e_i)=v} \hspace{-0.8em} \sigma_{e_i} \prod_{t(e_j)=v} \hspace{-0.8em} \sigma_{e_j}^{-1} =1.
    \end{equation}
\end{prop}

\begin{proof}
    In terms of the Wirtinger presentation of $\pi_1(S^3\smallsetminus \Gamma)$ associated to the choice of a planar diagram $D(\Gamma)$ of the graph $\Gamma$, we see that, because the range of the representation $\sigma$ is the abelian group $\Z/n\Z$, the Wirtinger relations at crossings \cref{eq:wirtingerRel1,eq:wirtingerRel2} simply give $\sigma_i=\sigma_j$, which means that the group elements assigned to all the strands in the planar diagram $D(\Gamma)$ belonging to the same edge of $\Gamma$ are equal. Equivalently, the topological labels just consist of group elements $\sigma_e \in \Z/n\Z$ attached to the edges of $\Gamma$. The remaining Wirtinger relation \cref{eq:vertexRelation} then gives \cref{Wirtcycl}.
\end{proof}

\smallskip

For such cyclic coverings, we can also introduce a notion of \textit{degenerate} topspin networks. One may regard them as analogous to the degenerate $\Q$-lattices of \cite{ClassTextbook} in our analogy with arithmetic noncommutative geometry discussed in \cref{arNCGsec} below. However, while in the $\Q$-lattices case it is crucial to include the non-invertible (degenerate) $\Q$-lattices in order to have a noncommutative space describing the quotient by commensurability and a dynamical flow on the resulting convolution algebra, in the setting we are considering here one already obtains an interesting algebra and dynamics on it just by restricting to the ordinary (nondegenerate) topspin foams and networks. In fact, the equivalence relation determined by stabilization and covering moves suffices to give rise to a system with properties similar to those of the $\Q$-lattices case.

\begin{defn}\label{def:degenerateTopspinNetworks}
    A possibly-degenerate topspin network over a compact Lie group $G$ is a tuple $(\Gamma, \rho, \iota, \sigma)$ of data consisting of
    \begin{enumerate}
        \item   a spin network $(\Gamma, \rho, \iota)$ in the sense of \cref{def:spinNetwork}, with $\Gamma \subset S^3$;
        \item   a labeling $\sigma$ of each $e$ of $\Gamma$ by a cyclic permutation $\sigma_e \in \Z/n\Z$.
    \end{enumerate}
    The data given by $\sigma$ do not necessarily satisfy the Wirtinger relation \cref{Wirtcycl}.
\end{defn}

Notice that, for more general branched coverings, which are not cyclic coverings, we cannot relax the Wirtinger conditions for the topological data. To see this, note that if the topological data are defined using planar projections $D(\Gamma)$, they need to be consistent with the generalized Reidemeister moves for embedded graphs in order to be well defined, and relaxing the Wirtinger conditions violates Reidemeister invariance. In the case of cyclic coverings this is not a problem, because the topological data are assigned to the graph, not to a planar projection.

\subsection{Embedded graphs and cobordisms}\label{cobordSec}

It is customary to assume in spin foam models \cite{BaezSpinFoam,RovelliQG} that spin foams are representation theoretic data assigned to the faces, edges, and vertices of a cobordism between embedded graphs. There are, however, different versions of cobordism (or concordance, extending the terminology in use for knots and links) in the case of spatial graphs, that is, graphs embedded in the 3-sphere. Some recent discussions of graph cobordisms and resulting concordance groups are given in \cite{MilnorInvariantsSpatialGraphs,KhovanovHomologyEmbeddedGraphs}. In the context of loop quantum gravity, a discussion of cobordisms of embedded graphs is given in \S 9.1 of \cite{RovelliQG}.

In the case of cobordisms between embedded knots and links in $S^3$, realized by smooth embedded surfaces in $S^3\times [0,1]$, the cobordisms can always be constructed out of a series of ``pairs of pants," which can also be described as a series of saddle critical points with respect to the height Morse function $S^3\times [0,1]\to [0,1]$, as in \cref{fig:Saddle}. These can also be described in terms of fusion and fission moves that consist of a surgery that attaches a band $D^1\times D^1$ replacing an $S^0\times D^1$ component of the boundary by a $D^1 \times S^0$ component, see \cite{CobordismBetweenLinks} and the discussion in \cite{KhovanovHomologyEmbeddedGraphs}.

\begin{figure}
    \centering
    \includegraphics[scale=0.95]{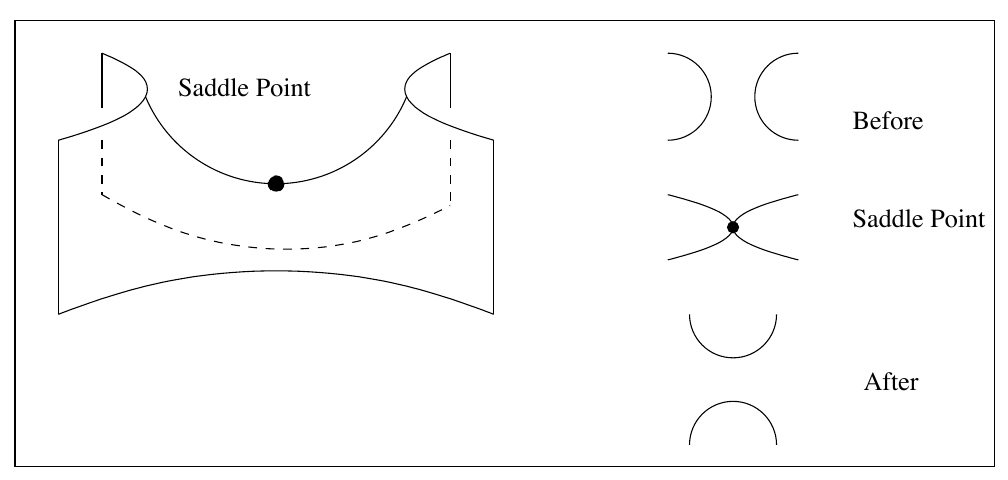}
    \caption{Saddle points in cobordisms by smooth surfaces}
    \label{fig:Saddle}
\end{figure}

The cobordisms $\Sigma$ we consider here are between embedded graphs, not only knots and links, and they are therefore realized by two-complexes $\Sigma$ embedded in $S^3 \times [0,1]$, which are not, in general, smooth surfaces. The basic moves that generate these cobordisms are therefore more general than the saddle points or band attachments that are sufficient for smooth surfaces.

The type of moves that we need to consider for these more general cobordisms include the possibility of contracting an edge of a graph along the cobordism, and of splitting a vertex into a pair of vertices, with the edges incident to the vertex partitioned among the two resulting vertices.  We can apply one or the other move to a graph provided the resulting cobordism will have the property that $\partial \Sigma = \Gamma \cup \bar\Gamma'$ with $\Gamma'$ the graph obtained as a result of the edge contraction or vertex splitting, and where $\partial$ is, as usual, the algebraic boundary operator applied to the two-complex $\Sigma$. We illustrate these two types of moves in \cref{fig:twoComplexCoveringMoves}.

\begin{figure}
    \centering
    \begin{tabular}{m{0.4\textwidth} m{0.1\textwidth} m{0.4\textwidth}}
        \centering
            \subfloat[][Edge contraction]
            {
                \includegraphics[height=2in]{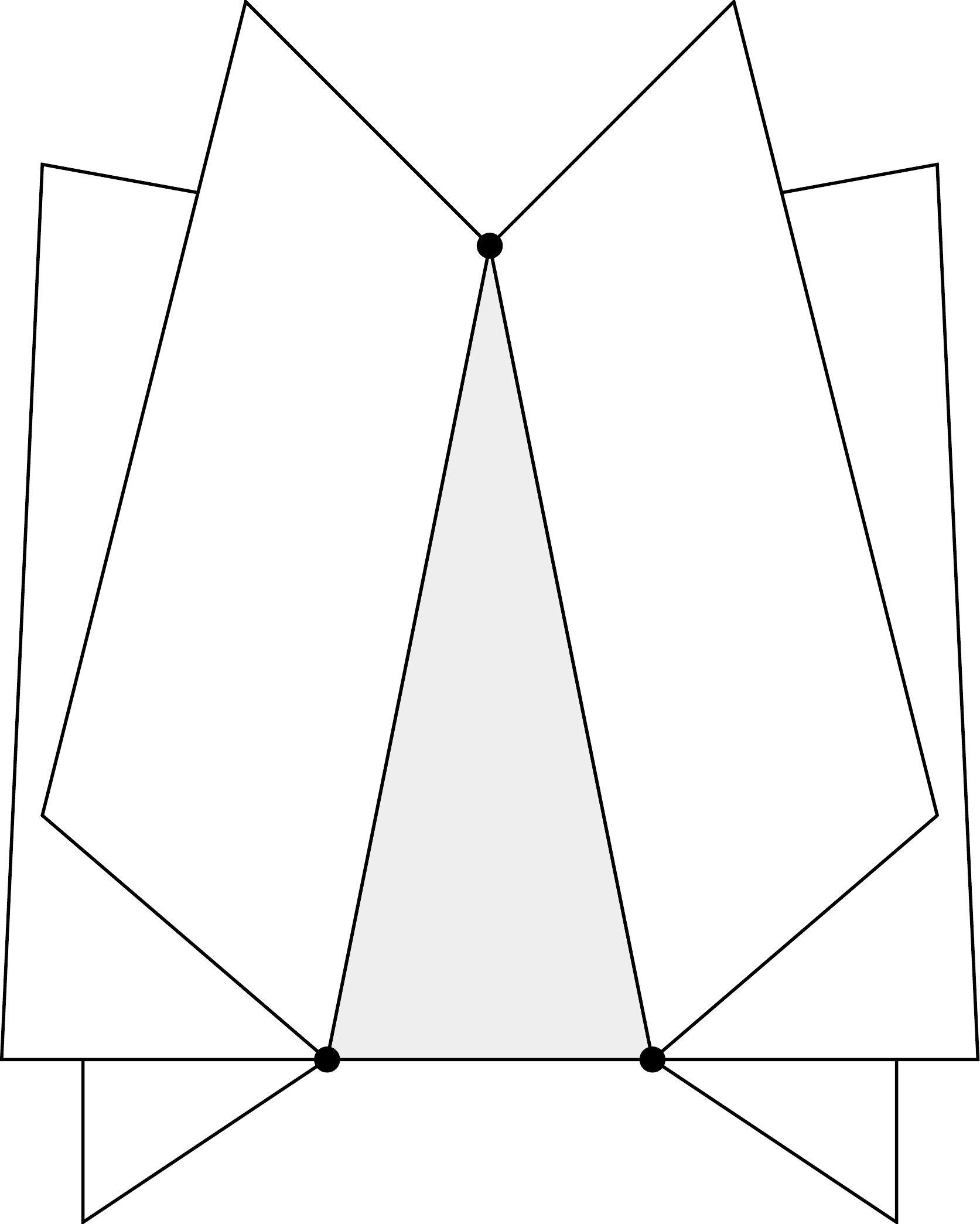}
            } &
        \centering
            $\begin{diagram}[balance,height=1.7in,vcenter]
            +t \\
            \uTo
            \end{diagram}$ &
        \centering
            \subfloat[][Vertex splitting]
            {
                \includegraphics[height=2in]{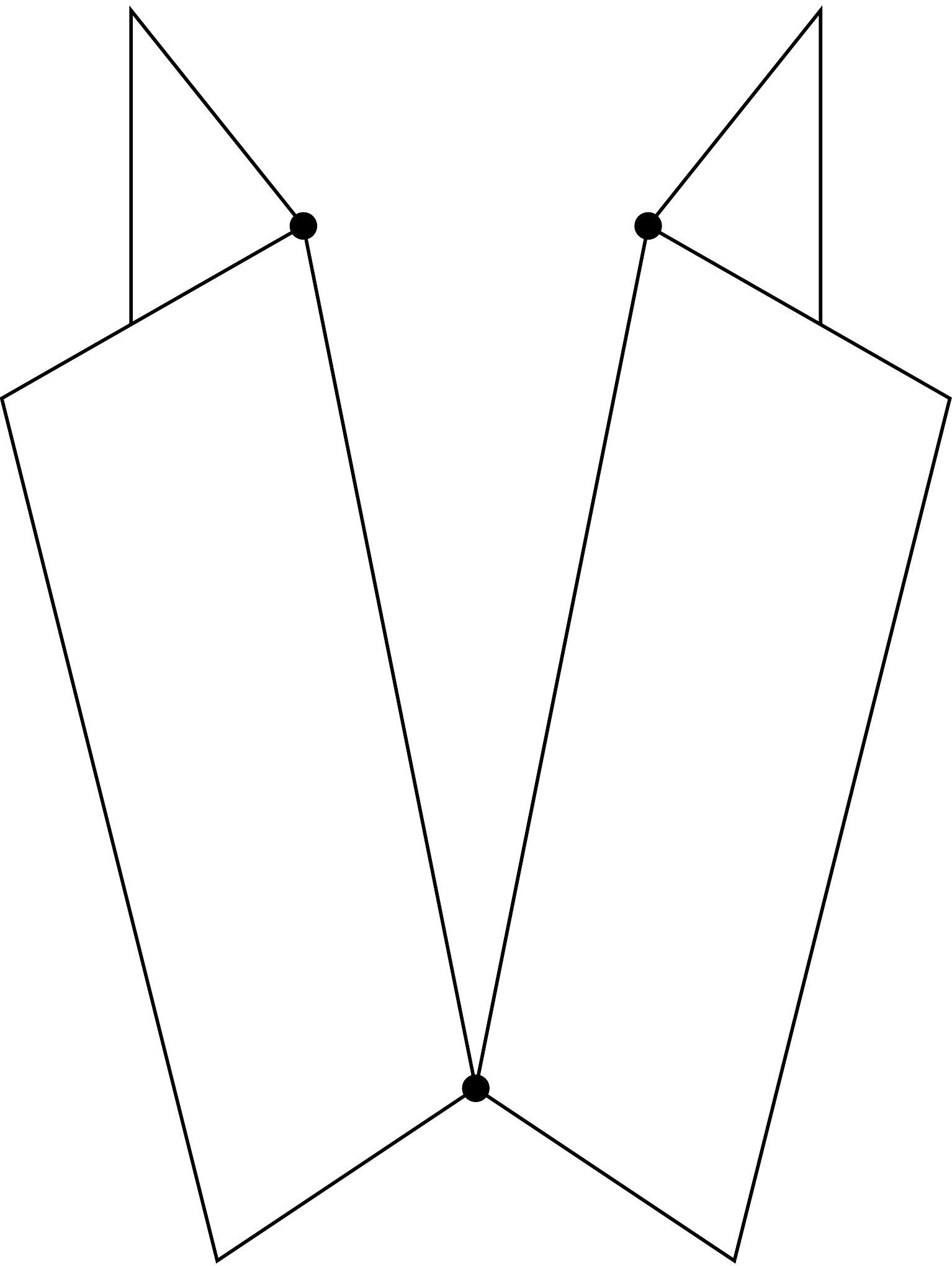}
            }
    \end{tabular}
    \caption{Covering moves for 2-complexes}
    \label{fig:twoComplexCoveringMoves}
\end{figure}

The case of the pair of pants cobordisms in the smooth case can be seen as a particular case of the second procedure (vertex splitting), where two incoming and two outgoing edges at a given vertex are spit in two different ways as two vertices, each with one incoming and one outgoing edge.

The use of sequences of edge contraction and vertex splitting moves as above to construct graph cobordisms includes, as a particular case, the fusion and fission moves for graphs described in \cite{KhovanovHomologyEmbeddedGraphs}, which can be obtained as a sequence of two such moves, one that separates vertices followed by one that contracts an edge. This explains more precisely the difference between the two notions of cobordism of graphs discussed in \cite{KhovanovHomologyEmbeddedGraphs}, both of which recover the usual notion of concordance when applied only to knots and links.

\subsection{Geometries and topologies}

Formally, a path-integral approach to Euclidean quantum gravity \cite{HamberQG,WaveFunctionOfTheUniverse} would formulate the transition amplitude between two given three-dimensional geometries as a path integral involving a sum over four-dimensional geometries, given by cobordisms $W$ with metrics $g$, weighted by the action functional $S_W(g)$ of Euclidean gravity (the Einstein--Hilbert action or a variant thereof):
\begin{equation*}
    \langle (M_1,g_1), (M_2,g_2)\rangle = \hspace{-2.5em}\sum_{\substack{(W, g) : \\ \partial W = M_1 \cup \bar M_2 \\ g|_{M_1} = g_1, \,  g|_{M_2} = g_2}} \hspace{-2em}\int \ee^{\ii S_W(g)} \, \cD[g].
\end{equation*}
The sum is over four-dimensional topologies interpolating via a cobordism between the given three-dimensional manifolds. Even at the purely formal level, it is far from obvious what one should mean by a sum over topologies in this setting. For example, it is well known that in dimension four one has an abundance of topological manifolds which do not admit any smooth structure. One does not expect such topologies to play a physical role in the partition function of quantum gravity, the latter being (at large scales at least) a smooth phenomenon. Moreover, one also has the case of exotic smooth structures, by which a given topological four-manifold that admits smoothings can carry many inequivalent smooth structures. There is growing evidence \cite{ExoticSmoothnessSemiclassicalEQG,ExoticSpacesInQG1} that exotic smooth structures indeed contribute differently to physics (see also \cite{ExoticSmoothnessAndPhysics}) and should be counted in the partition function of Euclidean quantum gravity.

Thus, when one approaches quantum gravity via a discretization of three-dimensional and four-dimensional geometries in terms of spin networks and spin foams, one needs to encode the different topologies and geometries so that the four-dimensional geometries being counted are only the smooth ones, but with all their different exotic structures.

The proposal we make here of using spin networks and spin foams decorated by additional data prescribing the topology of a branched cover addresses both of these issues. In fact, the main result we refer to in the case of four-dimensional geometries, is the description of all compact PL four-manifolds as branched coverings of the four-sphere, obtained in \cite{FourManifoldsAsBranchedCovers}. Since in the case of four-manifolds one can upgrade PL structures to smooth structures, this already selects only those four-manifolds that admit a smooth structures, and moreover it accounts for the different exotic structures.

\section{Noncommutative spaces as algebras and categories}\label{algcatSec}

Noncommutative spaces are often described either as algebras or as categories. In fact, there are several instances in which one can convert the data of a (small) category into an algebra and conversely.

\subsection{Algebras from categories}

We recall here briefly some well known examples in noncommutative geometry which can be described easily in terms of associative algebras determined by categories. We then describe a generalization to the case of 2-categories.

\begin{ex}
    The first example, and prototype for the generalizations that follow, is the well known
    construction of the (reduced) group $C^*$-algebra. Suppose we are given a discrete group
    $G$. Let $\C[G]$ be the group ring. Elements in $\C[G]$ can be described as finitely
    supported functions $f: G \to \C$, with the product given by the convolution product
    $$ (f_1\star f_2)(g) = \sum_{g=g_1g_2} f_1(g_1) f_2(g_2). $$
    This product is associative but not commutative.
    The involution $f^*(g) \equiv \overline{f(g^{-1})}$ satisfies $(f_1 \star f_2)^*= f_2^* \star f_1^*$
    and makes the group ring into an involutive algebra. The norm closure in the
    representation $\pi: \C[G] \to \cB(\ell^2(G))$, given by
    $(\pi(f) \xi)(g)=\sum_{g=g_1 g_2} f(g_1) \xi(g_2)$, defines the reduced group
    $C^*$-algebra $C^*_r(G)$.
\end{ex}

\begin{ex}
    The second example is similar, but one considers a discrete
    semigroup $S$ instead of a group $G$. One can still form the
    semigroup ring $\C[S]$ given by finitely supported functions
    $f: S \to \C$ with the convolution product
    $$  (f_1\star f_2)(s) = \sum_{s=s_1 s_2} f_1(s_1) f_2(s_2). $$
    Since this time elements of $S$ do not, in general, have inverses,
    one no longer has the involution as in the group ring case. One can still
    represent $\C[S]$ as bounded operators acting on the Hilbert space
    $\ell^2(S)$, via $(\pi(f) \xi)(s)=\sum_{s=s_1 s_2} f(s_1) \xi(s_2)$
    This time the elements $s\in S$ act on $\ell^2(S)$ by
    isometries, instead of unitary operators as in the group case. In fact,
    the delta function $\delta_s$ acts on the basis element $\epsilon_{s'}$
    as a multiplicative shift, $\delta_s: \epsilon_{s'} \mapsto \epsilon_{ss'}$.
    If one denotes by $\pi(f)^*$ the adjoint of $\pi(f)$ as an operator on
    the Hilbert space $\ell^2(S)$, then one has $\delta_s^* \delta_s = 1$ but
    $\delta_s \delta_s^* = e_s$ is an idempotent not equal to the identity.
    One can consider then the $C^*$-algebra $C^*_r(S)$, which is
    the $C^*$-subalgebra of $\cB(\ell^2(S))$ generated by the $\delta_s$
    and their adjoints.
\end{ex}

\begin{ex}
    The third example is the groupoid case: a groupoid $\cG=(\cG^{(0)},\cG^{(1)},s,t)$
    is a (small) category with a collection of objects $\cG^{(0)}$ also called the units of
    the groupoids, and with morphisms $\gamma \in \cG^{(1)}$, such that all morphisms
    are invertible. There are source and target maps
    $s(\gamma), t(\gamma) \in \cG^{(0)}$, so that, with the equivalent notation used above
    $\gamma \in \Mor_{\cG}(s(\gamma),t(\gamma))$. One can view $\cG^{(0)} \subset \cG^{(1)}$
    by identifying $x\in \cG^{(0)}$ with the identity morphisms $1_x \in \cG^{(1)}$. The
    composition $\gamma_1\circ \gamma_2$ of two elements
    $\gamma_1$ and $\gamma_2$ in $\cG^{(1)}$ is defined under the condition that
    $t(\gamma_2)=s(\gamma_1)$.  Again
    we assume here for simplicity that $\cG^{(0)}$ and $\cG^{(1)}$ are sets with the discrete
    topology. One considers then a groupoid ring $\C[\cG]$ of finitely supported functions
    $f: \cG^{(1)} \to \C$ with convolution product
    $$ (f_1 \star f_2) (\gamma) = \sum_{\gamma = \gamma_1 \circ \gamma_2} f_1(\gamma_1) f_2(\gamma_2) . $$
    Since a groupoid is a small category where all morphisms are invertible, there is an
    involution on $\C[\cG]$ given again, as in the group case, by $f^*(\gamma) =\overline{f(\gamma^{-1})}$.
    In fact, the group case is a special case where the category has a single object.
    One obtains $C^*$-norms by considering representations
    $\pi_x: \C[\cG] \to \cB(\ell^2(\cG_x^{(1)}))$, where $\cG^{(1)}_x =\{ \gamma \in \cG^{(1)} \, :  \, s(\gamma)=x \}$, given by
    $(\pi_x(f) \xi)(\gamma) = \sum_{\gamma = \gamma_1 \circ \gamma_2} f(\gamma_1) \xi(\gamma_2)$.
    This is well defined since for the composition $s(\gamma)=s(\gamma_2)$. One has
    corresponding norms $\| f \|_x = \| \pi_x(f) \|_{\cB(\ell^2(\cG^{(1)}_x))}$ and $C^*$-algebra completions.
\end{ex}

\begin{ex}
    The generalization of both the semigroup and the groupoid case is then
    the case of a semigroupoid, which is the same as a small category $\cS$. In this case one can
    describe the data $\Obj(\cS)$ and $\Mor_{\cS}(x,y)$ for $x,y\in \Obj(\cS)$ in terms of
    $\cS=(\cS^{(0)}, \cS^{(1)}, s,t)$ as in the groupoid case, but without assuming the invertibility
    of morphisms. The the algebra $\C[\cS]$ of finitely supported functions on $\cS^{(1)}$ with
    convolution product
    $$ (f_1\star f_2)(\phi) = \sum_{\phi = \phi_1\circ \phi_2} f_1(\phi_1) f_2(\phi_2) $$
    is still defined as in the groupoid case, but without the involution. One still has representations
    $\pi_x: \C[\cG] \to \cB(\ell^2(\cS_x^{(1)}))$ as in the groupoid case.
\end{ex}

\begin{ex}
    A simple example of $C^*$-algebras associated to small categories
    is given by the graph $C^*$-algebras. Consider for simplicity a finite oriented
    graph $\Gamma$. It can be thought of as a small category with objects the vertices
    $v\in V(\Gamma)$ and morphisms the oriented edges $e\in E(\Gamma)$. The source
    and target maps are given by the boundary vertices of edges. The semigroupoid algebra
    is then generated by a partial isometry $\delta_e$ for each oriented edge with
    source projections $p_{s(e)}=\delta_e^* \delta_e$. One has the relation
    $p_v = \sum_{s(e) =v} \delta_e \delta_e^*$.
\end{ex}

\begin{ex}
    The convolution algebra associated to a small category $\cS$ is
    constructed in such a way that the product follows the way morphisms can be decomposed in
    the category as a composition of two other morphisms. When one has a category with
    sufficient extra structure, one can do a similar construction of an associative algebra
    based on a decomposition of {\em objects} instead of morphisms. This is possible when
    one has an {\em abelian} category $\cC$, and one associates to it a Ringel--Hall algebra,
    see \cite{EisensteinSeriesAndQAAs}. One considers the set $\Iso(\cC)$ of isomorphism classes of objects and
    functions with finite support $f: \Iso(\cC) \to \C$ with the convolution product
    $$ (f_1 \star f_2) (X) = \sum_{X'\subset X} f_1(X') f_2(X/X'), $$
    with the splitting of the object $X$ corresponding to the exact sequence
    $$ 0 \to X' \to X \to X/X' \to 0. $$
\end{ex}

All these instances of translations between algebras and categories can be
interpreted within the general framework of ``categorification" phenomena.
We will not enter here into details on any of these examples. We give instead
an analogous construction of a convolution algebra associated to a 2-category.
This type of algebras were used both in \cite{CoveringsCorrespondencesNCG} and in \cite{CyclotomyAndEndomotives}; here we give
a more detailed discussion of their properties.

\subsection{2-categories}

In a 2-category $\cC$, one has objects $X\in \Obj(\cC)$, 1-morphisms $\phi\in \Mor_\cC(X,Y)$ for $X,Y \in \Obj(\cC)$, and 2-morphisms $\Phi\in \MorTwo_\cC(\phi,\psi)$ for $\phi,\psi \in \Mor_\cC(X,Y)$.

The composition of 1-morphisms $\circ: \Mor_\cC(X,Y)\times \Mor_\cC(Y,Z)\to \Mor_\cC(X,Z)$, $(\phi,\psi)\mapsto \psi \circ \phi$ is associative. For each object $X\in \Obj(\cC)$ there is an identity morphism
$1_X \in \Mor_\cC(X,X)$, which is the unit for composition.

There are two compositions for 2-morphisms: the vertical and horizontal compositions.
The vertical composition
$$ \bullet: \MorTwo_\cC(\varphi,\psi)\times \MorTwo_\cC(\psi,\eta) \to
\MorTwo_\cC(\varphi,\eta), $$
which is defined for $\varphi,\psi,\eta \in \Mor_\cC(X,Y)$, is associative and has
identity elements $1_\phi \in \MorTwo_\cC(\phi,\phi)$.

The horizontal composition
$$ \circ: \MorTwo_\cC(\varphi,\psi)\times \MorTwo_\cC(\xi,\eta) \to
\MorTwo_\cC(\xi\circ \varphi,\eta \circ \psi), $$
which follows the composition of 1-morphisms and is therefore defined for
$\varphi,\psi \in \Mor_\cC(X,Y)$ and $\xi,\eta \in \Mor_\cC(Y,Z)$, is also
required to be associative. It also has a unit element, given by the identity
2-morphism between the identity morphisms $I_X \in \MorTwo_\cC(1_X,1_X)$.

The compatibility between vertical and horizontal composition is given by
\begin{equation}\label{verthorcompat}
 (\Phi_1 \circ \Psi_1) \bullet (\Phi_2 \circ \Psi_2) = (\Phi_1 \bullet \Phi_2) \circ (\Psi_1\bullet \Psi_2).
\end{equation}

\subsection{Algebras from 2-categories}

The terminology 2-algebras is usually reserved to structures that generalize
Hopf algebras and bialgebras and that are given in terms of a multiplication
and a co-multiplication with some compatibility condition. Here we introduce the
terminology {\em 2-semigroupoid algebra} to denote the algebraic structure
that will be naturally associated to a 2-category in the same way as the convolution algebras
of small categories described above.

\begin{defn}\label{2algA}
A 2-semigroupoid algebra $\cA$ over $\C$ is a $\C$-vector space endowed with two
associative multiplications $\circ$ and $\bullet$, each giving  $\cA$ the
structure of an associative $\C$-algebra with units, $1_\circ$ and $1_\bullet$,
respectively. The two multiplications satisfy the condition
\begin{equation}\label{2prods}
(a_1 \circ b_1)\bullet (a_2 \circ b_2) = (a_1 \bullet a_2) \circ (b_1 \bullet b_2),
\end{equation}
for all $a_1,a_2,b_1,b_2\in \cA$.
\end{defn}

We see then that this algebraic structure arises naturally from 2-categories.
For a 2-category $\cC$ we use the following notation:
$$ \cC^{(0)}= \Obj(\cC), \ \ \ \cC^{(1)}= \bigcup_{x,y\in \cC^{(0)}} \Mor_\cC(x,y), \ \ \
\cC^{(2)} = \bigcup_{\phi,\psi\in \cC^{(1)}} \MorTwo_\cC(\phi,\psi). $$

\begin{lem}\label{2algC}
Let $\cC$ be a small 2-category. Let $\C[\cC]$ be the vector space of finitely supported
functions $f: \cC^{(2)} \to \C$. The product corresponding to the vertical composition
\begin{equation}\label{prodvert}
(f_1 \bullet f_2)(\Phi) = \sum_{\Phi = \Phi_1 \bullet \Phi_2} f_1(\Phi_1) f_2(\Phi_2)
\end{equation}
and the one corresponding to the horizontal composition
\begin{equation}\label{prodhor}
(f_1 \circ f_2)(\Phi)= \sum_{\Phi = \Psi \circ \Upsilon} f_1(\Psi) f_2(\Upsilon)
\end{equation}
give $\C[\cC]$ the structure of a 2-semigroupoid algebra.
\end{lem}

\begin{proof}
    The associativity of both products follows from the associativity of both the
    vertical and the horizontal composition of 2-morphisms in a 2-category. One only
    needs to check that the compatibility condition \cref{2prods} between the two
    products holds. We have
    \begin{align*}
            ((f_1 \circ h_1) \bullet (f_2 \circ h_2)) (\Phi)
        & = \sum_{\Phi = \Phi_1 \bullet \Phi_2} (f_1 \circ h_1)(\Phi_1) (f_2 \circ h_2) (\Phi_2) \\
        & = \sum_{\Phi = \Phi_1 \bullet \Phi_2} \left( \left( \sum_{\Phi_1 =\Psi_1 \circ \Xi_1} f_1(\Psi_1) h_1(\Xi_1) \right) \left(\sum_{\Phi_2=\Psi_2 \circ \Xi_2} f_2(\Psi_2) h_2(\Xi_2) \right) \right) \\
        & = \sum_{\Phi=(\Psi_1 \circ \Xi_1)\bullet (\Psi_2 \circ \Xi_2)} f_1(\Psi_1) h_1(\Xi_1) f_2(\Psi_2) h_2(\Xi_2) \\
        & = \sum_{\Phi=(\Psi_1 \bullet \Psi_2)\circ (\Xi_1 \bullet \Xi_2)} f_1(\Psi_1) f_2(\Psi_2) h_1(\Xi_1)  h_2(\Xi_2) \\
        & = ((f_1 \bullet f_2)\circ (h_1\bullet h_2))(\Phi).
    \end{align*}
\end{proof}

A 2-semigroupoid algebra corresponding to a 2-category of low dimensional geometries
was considered in \cite{CoveringsCorrespondencesNCG}, as the ``algebra of coordinates" of a noncommutative
space of geometries. A similar construction of a 2-semigroupoid algebra coming from surgery
presentations of three-manifolds was considered in \cite{CyclotomyAndEndomotives}.

\section{A model case from arithmetic noncommutative geometry}\label{arNCGsec}

We discuss here briefly a motivating construction that arises in
another context in noncommutative geometry, in applications of
the quantum statistical mechanical formalism to arithmetic of
abelian extensions of number fields and function fields.
We refer the reader to Chapter 3 of the book \cite{ClassTextbook} for
a detailed treatment of this topic.

The main feature of the construction we review below, which is
directly relevant to our setting of spin networks and spin foams, is the
following. One considers a space parameterizing a certain family of
possibly singular geometries. In the arithmetic setting these geometries
are pairs of an $n$-dimensional lattice $\Lambda$ and a group
homomorphism $\phi: \Q^n/\Z^n \to \Q \Lambda /\Lambda$, which
can be thought of as a (possibly degenerate) level structure, a
labeling of the torsion points of the lattice $\Lambda$ in terms of
the torsion points of the ``standard lattice." Among these geometries
one has the ``nonsingular ones," which are those for which the
labeling $\phi$ is an actual level structure, that is, a group isomorphism.
On this set of geometries there is a natural equivalence relation, which
is given by commensurability of the lattices, $\Q\Lambda_1 =\Q \Lambda_2$
and identification of the labeling functions, $\phi_1=\phi_2$ modulo
$\Lambda_1 + \Lambda_2$.  One forms a convolution algebra associated
to this equivalence relation, which gives a noncommutative space
parameterizing the moduli space of these geometries up to commensurability.

The resulting convolution algebra has a natural time evolution, which
can be described in terms of the covolume of lattices. The resulting quantum statistical
mechanical system exhibits a spontaneous symmetry breaking phenomenon.
Below the critical temperature, the extremal low temperature KMS equilibrium
states of the system automatically select only those geometries that are
nondegenerate.

This provides a mechanism by which the correct type of geometries
spontaneously emerge as low temperature equilibrium states.
A discussion of this point of view on emergent geometry can be
found in \S 8 of Chapter 4 of \cite{ClassTextbook}.

The reason why this is relevant to the setting of spin foam models
is that one can similarly consider a convolution algebra that parameterizes
all (possibly degenerate) topspin foams carrying the metric
and topological information on the quantized four-dimensional geometry.
One then looks for a time evolution on this algebra whose low temperature
equilibrium states would automatically select, in a spontaneous symmetry
breaking phenomenon, the correct nondegenerate geometries.

We recall in this section the arithmetic case, stressing the
explicit analogies with the case of spin foams we are considering here.

\subsection{Quantum statistical mechanics}

The formalism of Quantum Statistical Mechanics in the operator algebra setting
can be summarized briefly as follows. (See \cite{OperatorAlgebrasAndQSM1,OperatorAlgebrasAndQSM2} and \S 3 of \cite{ClassTextbook}
for a more detailed treatment.)

One has a (unital) $C^*$-algebra $\cA$ of observables, together with a time evolution---that is, a one-parameter family of automorphisms $\sigma : \R \to \Aut(\cA)$.

A state on the algebra of observables is a linear functional
$\varphi: \cA \to \C$, which is normalized by $\varphi(1)=1$
and satisfies the positivity condition $\varphi(a^*a)\geq 0$
for all $a\in \cA$.

Among states on the algebra, one looks in particular for those that are
equilibrium states for the time evolution. This property is expressed by
the KMS condition, which depends on a thermodynamic parameter
$\beta$ (the inverse temperature). Namely, a state $\varphi$ is a KMS$_\beta$
state for the dynamical system $(\cA,\sigma)$ if for every choice of
two elements $a,b \in \cA$ there exists a function $F_{a,b}(z)$ which is
holomorphic on the strip in the complex plane
$I_\beta =\{ z\in \C \,:\,  0 < \Im(z) <\beta \}$ and extends to a continuous
function to the boundary $\partial I_\beta$ of the strip, with the property that,
for all $t\in \R$,
\begin{align*}
    F_{a,b}(t) & = \varphi(a\sigma_t(b)), \\
    F_{a,b}(t + \ii\beta) & =\varphi(\sigma_t(b)a).
\end{align*}
This condition can be regarded as identifying a class of functionals which
fail to be traces by an amount that is controlled by interpolation via a holomorphic
function that analytically continues the time evolution. An equivalent formulation
of the KMS condition in fact states that the functional $\varphi$ satisfies
$$ \varphi(ab) = \varphi(b \sigma_{\ii \beta}(a)), $$
for all $a,b$ in a dense subalgebra of ``analytic elements".

At zero temperature $T \equiv 1/\beta = 0$, the KMS$_\infty$ states
are defined in \cite{FromPhysicsToNTViaNCG1} as the weak limits of KMS$_\beta$
states $\varphi_\infty (a) = \lim_{\beta\to \infty} \varphi_\beta(a)$.

KMS states are equilibrium states, namely one has $\varphi(\sigma_t(a))=\varphi(a)$
for all $t\in \R$. This can be seen from an equivalent formulation of the KMS condition
as the identity $\varphi(a b)=\varphi(b \sigma_{\ii\beta}(a))$ for all $a,b$ in a dense
subalgebra of analytic elements on which the time evolution $\sigma_t$ admits an analytic
continuation $\sigma_z$, see \cite{OperatorAlgebrasAndQSM2} \S 5.

Given a representation $\pi: \cA \to \cB(\cH)$ of the algebra
of observables as bounded operators on a Hilbert space,
one has a Hamiltonian $H$ generating the time evolution
$\sigma_t$ if there is an operator $H$ (generally unbounded)
on $\cH$ satisfying
$$\pi(\sigma_t(a))=\ee^{\ii tH} \pi(a) \ee^{-\ii tH}$$
for all $a\in \cA$ and for all $t\in \R$.

A particular case of KMS states is given by the Gibbs states
\begin{equation}\label{Gibbs}
 \frac{1}{Z(\beta)}\, \Tr\left( \pi(a)\, \ee^{-\beta H} \right),
\end{equation}
with partition function $Z(\beta)=\Tr\left(\ee^{-\beta H} \right)$.
However, while the Gibbs states are only defined under
the condition that the operator $\exp(-\beta H)$ is trace class,
the KMS condition holds more generally and includes
equilibrium states that are not of the Gibbs form.

\subsection{Lattices and commensurability}

The notion of $\Q$-lattices and commensurability was introduced in
\cite{FromPhysicsToNTViaNCG1} to give a geometric interpretation of
a quantum statistical mechanical system
previously constructed by Bost and Connes in \cite{BostConnesAlgebras}
as the convolution algebra of functions on the
(noncommutative) moduli space of commensurability classes
of one-dimensional $\Q$-lattices up to scaling. This geometric
interpretation gave rise to several generalizations of
the original Bost--Connes system (see \S 3 of \cite{ClassTextbook}).

We recall here briefly the geometry of $\Q$-lattices, because it
will serve as a model for our treatment of spin foams.

An $n$-dimensional $\Q$-lattice $( \Lambda , \phi) $ is a
pair of a lattice $\Lambda\subset \R^n$ together with
a possibly degenerate labeling of the torsion points given by a
group homomorphism
$$ \phi :  \Q^n/\Z^n \longrightarrow \Q\Lambda / \Lambda . $$

The nondegenerate objects consist of those $\Q$-lattices that
are termed invertible, that is, those for which the homomorphism
$\phi$ is in fact an isomorphism.

Two $\Q$-lattices are commensurable if $\Q\Lambda_1=\Q\Lambda_2$ and
$\phi_1 = \phi_2 \mod \Lambda_1 + \Lambda_2$.

The convolution algebra on the space of $\Q$-lattices up to commensurability
consists of functions $f((\Lambda,\phi),(\Lambda',\phi'))$ of pairs of commensurable
lattices $(\Lambda,\phi)\sim (\Lambda',\phi')$, with the convolution product
$$ (f_1\star f_2)((\Lambda,\phi),(\Lambda',\phi')) =
\sum_{(\Lambda'',\phi'')\sim (\Lambda,\phi)} f_1((\Lambda,\phi),(\Lambda'',\phi''))
f_2((\Lambda'',\phi''),(\Lambda',\phi')). $$

In the case of one-dimensional or two-dimensional $\Q$-lattices considered in
\cite{FromPhysicsToNTViaNCG1}, one can similarly consider the convolution algebra for
the commensurability relation on $\Q$-lattices considered up to a scaling
action of $\R^*_+$ or $\C^*$, respectively. One then has on the resulting
algebra a natural time evolution by the ratio of the covolumes of the
pair of commensurable lattices,
$$ \sigma_t(f) ((\Lambda,\phi),(\Lambda',\phi'))  = \left(\frac{\Vol(\R^n/\Lambda')}{\Vol(\R^n/\Lambda)}\right)^{\ii t} \, f((\Lambda,\phi),(\Lambda',\phi')). $$

\subsection{Low temperature KMS states}

The quantum statistical mechanical systems of one-dimensional or two-dimensional
$\Q$-lattices introduced in \cite{FromPhysicsToNTViaNCG1} exhibit a pattern of symmetry breaking
and the low temperature extremal KMS states are parameterized by exactly those
$\Q$-lattices that give the nondegenerate geometries, the invertible $\Q$-lattices.

One considers representations of the convolution algebra on the Hilbert space
$\ell^2(\cC_{(\Lambda,\phi)})$, with $\cC_{(\Lambda,\phi)}$ the commensurability
class of a given $\Q$-lattice,
$$ \pi(f) \xi (\Lambda',\phi') =\sum_{(\Lambda'',\phi'')\sim (\Lambda,\phi)}
f((\Lambda',\phi'),(\Lambda'',\phi'')) \xi(\Lambda'',\phi''). $$
In the case when $(\Lambda,\phi)$ is invertible, the Hamiltonian generating the
time evolution in this representation has positive energy and one has a corresponding
Gibbs KMS state
\begin{equation}\label{GibbsKMS}
\varphi_{L,\beta}(f)=Z(\beta)^{-1} \sum_{m\in \SL_2(\Z)\backslash M_2^+(\Z)}
f(m\rho,m(z)) \,\, \det(m)^{-\beta} ,
\end{equation}
for $L=(\Lambda,\phi)$ an invertible $\Q$-lattice, where we use the parametrization
as in \cite{FromPhysicsToNTViaNCG1} of the invertible $\Q$-lattice $L=(\Lambda,\phi)$ by a pair $(\rho,z)$
of an element $\rho\in GL_2(\hat\Z)$ and a point $z\in \H$, and the $\Q$-lattices in
the commensurability class with elements $(m\rho, m(z))$ with $m\in \SL_2(\Z)\backslash M_2^+(\Z)$.
The partition function is $Z(\beta)=\sum \det(m)^{-\beta} =\zeta(\beta)\zeta(\beta-1)$.
A complete characterization of extremal KMS states in terms of
measures on the space of invertible $\Q$-lattices was given in \cite{FromPhysicsToNTViaNCG1}.

\subsection{The paradigm of emergent geometry}

Without entering into further details, the important observation here is the fact
that the nondegenerate geometries (the invertible $\Q$-lattices) are selected
out of a larger space containing all possibly degenerate geometries (all $\Q$-lattices)
via a dynamical phenomenon of spontaneous symmetry breaking.

This was discussed in \S 8 of Chapter 4 of \cite{ClassTextbook}, in order to propose
a scenario of {\em emergent geometry} according to which spacetime geometry
should arise spontaneously as a symmetry breaking phenomenon. The analogy
described there is between the setting of $\Q$-lattices recalled here above and
the setting of almost-commutative geometries used in the particle physics models
based on noncommutative geometry (see \cite{GravityAndSMWithNM} and Chapter 1 of \cite{ClassTextbook}, as well as \cref{sec:spinFoamsWithMatter} below).
The proposed analog of the possibly degenerate $\Q$-lattices is a degenerate
version of the notion of spectral triple in noncommutative geometry, regarded as
a kind of correspondence, while the nondegenerate geometries that provide
physical models of matter and gravity are expected to arise dynamically through
a symmetry breaking phenomenon as the invertible $\Q$-lattices do.
We present here a similar picture, where the analogy is now between the
case of $\Q$-lattices and that of spin foam models.

In the setting we describe here, we obtain a noncommutative space by considering
topspin foams and networks (our analog of $\Q$-lattices) with the equivalence relation
generated by covering moves, which plays the role of commensurability, and the
dynamics is generated by the lack of invariance under these moves of the spin foam
amplitudes. Other more elaborate constructions based on composition operations
of topspin foams, which we describe in the following sections, also fit into a similar
conceptual picture.

A frequent problem with the construction of spin foam models through
combinatorial data of simplicial complexes with spin labels is that one
often obtains, along with ordinary smooth geometries, also spurious
solutions that do not correspond to manifolds and which are often difficult
to recognize and separate from the ``good solutions." The idea of spontaneous
emergence of geometry as a symmetry breaking phenomenon proposed
in \cite{ClassTextbook} suggests that the correct solutions, or nondegenerate geometries,
should arise as low temperature extremal KMS equilibrium states from a
natural dynamics on a noncommutative space describing the overall moduli space
of all possibly-degenerate geometries. Thus, the type of convolution algebras with
dynamics that we consider here may also have possible generalizations that
address this problem.

\section{Categories of topspin networks and foams}\label{catcovSec}

We give here two natural constructions of categories of topspin networks and foams
and the corresponding algebras.
These will serve as a simpler toy model, before we discuss the
more elaborate construction of the 2-category in \cref{fiberSec} below.

\subsection{The groupoid of topspin networks and covering moves}\label{GgrpdSec}

Let $\cG$ denote the groupoid of the equivalence relation on topspin networks
generated by stabilization and covering moves. The objects of $\cG$ are
topspin networks $\psi=(\Gamma, \rho,\iota,\sigma)$ and the morphisms
are pairs $(\psi,\psi')$ of topspin foams such that the data $(\Gamma,\sigma)$
and $(\Gamma',\sigma')$ are related, after stabilization, by a finite sequence of
covering moves. The data $(\rho,\iota)$ and $(\rho',\iota')$ are related by the
compatibility conditions of \cref{sec:consistencyConditions}.

We consider then the groupoid algebra $\C[\cG]$ of finitely supported
functions $f(\psi,\psi')$ of pairs of covering moves equivalent topspin
networks, with the convolution product
$$ f_1\star f_2 (\psi,\psi')= \sum_{\psi\sim\psi''\sim \psi'} f_1(\psi,\psi'') f_2 (\psi'',\psi'). $$

One can construct a similar groupoid using topspin foams
$\Psi=(\Sigma,\tilde\rho,\tilde\iota,\tilde\sigma)$ and the same
equivalence relation $\Psi\sim \Psi'$ when $(\Sigma,\tilde\sigma)$
and $(\Sigma',\tilde\sigma')$ determine the same four-manifold,
expressed in terms of stabilizations and covering moves in the
four-dimensional setting. The corresponding groupoid algebra of functions $f(\Psi)$
has convolution product as above
$$ f_1\star f_2 (\Psi,\Psi')= \sum_{\Psi\sim\Psi''\sim \Psi'} f_1(\Psi,\Psi'') f_2 (\Psi'',\Psi'). $$

\subsection{The category of topspin networks and foams}\label{ScatSec}

We can also consider the usual category of spin networks and foams used
in loop quantum gravity, which has objects that are spin networks and
morphisms that are spin foams cobordisms. One can enrich it with topological
data and obtain a category, or semigroupoid, $\cS$ whose objects are topspin networks
$\psi=(\Gamma, \rho,\iota,\sigma)$ and whose morphisms ${\rm Mor}_{\cC}(\psi,\psi')$
are topspin foams $\Psi=(\Sigma,\tilde\rho,\tilde\iota,\tilde\sigma)$ cobordisms with
$\partial \Sigma = \Gamma \cup \bar \Gamma'$ and compatible data
$(\tilde\rho,\tilde\iota,\tilde\sigma)$.

Again, one can associate to this category the semigroupoid algebra $\C[\cS]$ of
functions of finite support $f(\Psi)$ with convolution product
$$ f_1\star f_2 (\Psi)= \sum_{\Psi=\Psi_1 \Psi_2} f_1(\Psi_1) f_2(\Psi_2), $$
where the composition of topspin foams is obtained by gluing them along
a common boundary topspin network.

\subsection{Including degenerate geometries}\label{deggeomSec}

A variant of the categories $\cG$ and $\cS$ described above can be obtained
by restricting to cyclic coverings. In that case, we can include among the objects
also the degenerate topspin networks and topspin foams, as in \cref{def:degenerateTopspinNetworks}. The covering moves would still be the
same, and the topspin cobordisms would also be as before, except that among
the cobordisms one also allows for those where the Wirtinger relations at
edges for the data $\tilde\sigma_f$ is not imposed.  This variants allows us to
illustrate, in this toy model, a simple dynamical mechanisms that selects the
nondegenerate geometries among the degenerate ones.

\subsection{Representations}\label{repsGSsec}

We consider the usual representations of $\C[\cG]$ on the Hilbert space spanned by
the arrows in the groupoid with fixed source, $\ell^2(\cG_\psi'^{(1)})$. This is the
space $\ell^2(\cC_{\psi'})$ of the equivalence class $\cC_{\psi'}=\{ \psi \,:\, \psi \sim \psi'\}$
under the equivalence relation determined by stabilizations and covering moves. The
representation $\pi_{\psi'}$ of $\C[\cG]$ on this Hilbert space is then given by
$$ \pi_{\psi'}(f)\xi (\psi) = \sum_{\psi''\sim\psi\sim\psi'} f(\psi,\psi'') \xi(\psi''), $$
where we have identified, in the standard way, an element $\xi\in \ell^2(\cC_{\psi'})$
with a square integrable function $\xi$ on the set of $\psi\in \cC_{\psi'}$.

The case of the groupoid of the covering moves equivalence on topspin foams
is analogous, with
$$ \pi_{\Psi'}(f)\xi (\Psi) = \sum_{\Psi''\sim\Psi\sim\Psi'} f(\Psi,\Psi'') \xi(\Psi''). $$

\smallskip

In a similar way, in the case of $\C[\cS]$ we can consider representations on
a Hilbert space $\cH_\psi$, which in this case is spanned by all the topspin foams
$\Psi$ with $\psi \subset \partial \Psi$. The representation is again given by the
same convolution product
$$ \pi_\psi (f) \xi(\Psi) = \sum_{\Psi=\Psi_1\Psi_2\,:\, \psi\subset \partial\Psi_2}
f(\Psi_1) \xi(\Psi_2). $$

\subsection{Dynamics on algebras of topspin networks and foams}\label{dynGSsec}

In general spin foam models \cite{BaezSpinFoam,RovelliQG}, one assigns to a spin foam $\Psi = (\Sigma,\tilde\rho,\tilde\iota)$ an amplitude of the form
\begin{equation}\label{spinfoamA}
    \bA(\Psi) = \omega(\Sigma) \prod_f \bA_f(\tilde\rho_f) \prod_e \bA_e(\tilde\rho_{F(e)}, \tilde\iota_e) \prod_v \bA_v(\tilde\rho_{F(v)},\tilde\iota_{E(v)}),
\end{equation}
where the weight factor $\omega(\Sigma)$ is a term that depends only on the two-complex $\Sigma$, while the amplitude $\bA_f$ at a given face depends on the representation $\tilde\rho_f$, the amplitude $\bA_e$ at a given edge depends on the representations $\tilde\rho_{F(e)}$ assigned to the faces adjacent to that edge and on the intertwiner $\tilde\iota_e$ at the edge, and the amplitude $\bA_v$ at a vertex depends on the representations $\tilde\rho_{F(v)}$ and intertwiners $\tilde\iota_{E(v)}$ of the faces and edges adjacent to the given vertex. One obtains the transition amplitudes between spin networks $\psi$ and $\psi'$ by summing all the $\bA(\Psi)$ for all the spin foams connecting $\psi$ and $\psi'$.

\smallskip

We introduced topspin networks and topspin foams as a way to encode the topology of the underlying manifold as part of the discrete combinatorial data of quantized geometry. Thus, topspin networks and topspin foams that are related by covering moves, with the correct compatibility conditions on the representation theoretic data, should be regarded as describing the same quantum geometry. However, it is not always true that the spin foam amplitudes are necessarily invariant under covering moves. In a spin foam model where the amplitudes are not covering-moves invariant, this lack of invariance determines a nontrivial dynamics on the algebra of the groupoid of the equivalence relation.

\smallskip

For a choice of the amplitude \eqref{spinfoamA} which is not invariant under covering moves, and for which the amplitudes are positive real numbers, one can define a dynamics on the groupoid algebra of \cref{GgrpdSec} by setting
\begin{equation}\label{Gtsfdyn}
    \sigma_t(f) (\Psi, \Psi') = \left(\frac{\bA(\Psi)}{\bA(\Psi')}\right)^{\ii t} f(\Psi,\Psi').
\end{equation}
It is easy to check that this indeed defines a time evolution on the resulting groupoid algebra.

Moreover, if one includes in the groupoid, in the case of topological data $\sigma$ with values in a cyclic group $\Z/n\Z$, also the degenerate topspin foams of \cref{def:degenerateTopspinNetworks}, then one can include in the time evolution of \cref{Gtsfdyn} a factor that measures how degenerate the geometry is.

Let $\chi: S_\infty \to \mathrm{U}(1)$ be a multiplicative character of the infinite symmetric group $S_\infty = \bigcup_n S_n$.
We define an element $\fW(\psi) \in S_n$ as the elements in the permutation group given by the product of the Wirtinger relations at vertices:
\begin{equation*}
    \fW(\psi) \equiv  \left(\prod_{v \in V(\Gamma)} \prod_{e : v \in \partial(e)} \sigma_e \prod_{e : v \in \bar\partial(e)} \sigma_e^{-1}\right).
\end{equation*}
In fact, for cyclic coverings the elements $\fW(\Psi)$ are contained in the subgroup $\Z/n\Z$ of
cyclic permutations in $S_n$, so it suffices to take a character $\chi$ of $U(1)$ and identify
the elements $\fW(\Psi)$ with roots of unity in $U(1)$.

This element is always $\fW(\psi)=1$ for the nondegenerate geometries, while it can be
$\fW(\psi)\neq 1$ in the degenerate cases.
We can similarly define $\fW(\Psi)$ for a degenerate topspin foams, using the Wirtinger
relation for faces incident to an edge,
\begin{equation}\label{WSigma}
    \fW(\Psi) \equiv  \prod_{f : e \in \partial(f)} \tilde\sigma_f \prod_{f : \bar e \in \partial(f)} \tilde\sigma_f^{-1}.
\end{equation}

With these in hand, we modify the time evolution of \cref{Gtsfdyn} by setting
\begin{equation}\label{timeWirt}
        \sigma_t(f) (\Psi,\Psi') = \left(\frac{\bA(\Psi)}{\bA(\Psi')}\right)^{\ii t}
        \chi(\fW(\Psi) \fW(\Psi')^{-1})^t \, f(\Psi,\Psi').
\end{equation}

In this way, the time evolution of \cref{timeWirt} accounts not only for the normalized amplitude discrepancy between topspin foams related by covering moves, but also for the degeneracy of the geometry, measured by the failure if the topological data to satisfy the Wirtinger relations and therefore define a genuine three-dimensional topology or four-dimensional cobordism. It is possible that more sophisticated examples of this type of dynamics may be constructed using, instead of multiplicative characters of the infinite symmetric group, categorical representations of embedded graphs in the
sense of \cite{YetterKnottedGraphs} and quantum groups.

\medskip

One can treat in a similar way the case of the algebra $\C[\cS]$ of topspin foams
described in \cref{ScatSec}.
We define a normalized amplitude of the form
\begin{equation}\label{normAmp}
    \bA^0(\Psi) = \frac{\bA(\Psi)}{\bA(\psi_0)},
\end{equation}
where $\Psi$ is a spin foams connecting spin networks $\psi_0$ and $\psi_1$, with $\partial\Sigma =\Gamma_0\cup\bar\Gamma_1$, and the amplitude for the spin network $\psi_0$ is defined as
\begin{equation}\label{bA0}
    \bA(\psi_0) = \prod_e \bA_{f_e}(\tilde\rho) \prod_v \bA_{e_v}(\tilde\rho,\tilde\iota),
\end{equation}
where the amplitudes $ \bA_{f_e}(\tilde\rho) $ and $\bA_{e_v}(\tilde\rho,\tilde\iota)$ are those assigned to the spin foam with two-complex $\Sigma\times [0,1]$, with $f_e =e\times [0,1]$ and $e_v = v \times [0,1]$. One can equivalently normalize by the amplitude of the spin network $\psi_1$.

We assume also that the weight factor $\omega(\Sigma)$ satisfies the multiplicative property
\begin{equation}\label{weightmultipl}
    \omega(\Sigma_1 \cup_\Gamma \Sigma_2)= \omega(\Sigma_1) \omega(\Sigma_2),
\end{equation}
for two two-complexes glued together along a boundary graph $\Gamma$.
This assumption is necessary for the result that follows.

Notice how any additive
invariant of embedded surfaces satisfying inclusion--exclusion when gluing
two surfaces along a common boundary will give rise, by exponentiation, to
an invariant satisfying \cref{weightmultipl}. In fact, let $\chi$ be an additive
invariant. The inclusion--exclusion property gives
$$ \chi(\Sigma_1 \cup_\Gamma \Sigma_2)=\chi(\Sigma_1)+\chi(\Sigma_2) -\chi(\Gamma). $$
Then setting
$$ \omega(\Sigma)=\exp(\alpha (\chi(\Sigma)-\chi(\Gamma_1))), $$
where $\partial\Sigma =\Gamma_1 \cup \bar\Gamma$, and for some constant $\alpha$, gives
an invariant with the property \cref{weightmultipl} above.

\begin{lem}\label{multiplAmp}
    Consider a spin foam model with amplitudes \cref{spinfoamA}, where the weight factor satisfies \cref{weightmultipl}. Then the normalized amplitude \cref{normAmp} has the property that, if $\Sigma =\Sigma_1\cup_{\Gamma}\Sigma_2$ is obtained by gluing together two two-complexes along their common boundary, then the normalized amplitudes multiply:
    \begin{equation*}
        \bA^0(\Sigma_1\cup_\Gamma \Sigma_2) =\bA^0(\Sigma_1)\bA^0(\Sigma_2).
    \end{equation*}
    Setting
\begin{equation}\label{sigmaSalg}
\sigma_t (f) (\Psi)= \bA^0(\Sigma)^{\ii t} f (\Psi),
\end{equation}
for $\Psi=(\Sigma,\tilde\rho,\tilde\iota,\tilde\sigma)$, defines a time evolution on
the algebra $\C[\cS]$ introduced in \cref{ScatSec}.
\end{lem}

\begin{proof}
    The faces of $\Sigma$ are faces of either $\Sigma_1$ or $\Sigma_2$, so the factor $\prod_f \bA_f(\tilde\rho)$ in the (normalized) amplitude is a product $\prod_{f_1} \bA_{f_1}(\tilde\rho_1) \prod_{f_2} \bA_{f_2}(\tilde\rho_2)$. The edges and vertices of $\Sigma$ are those of $\Sigma_1$ and $\Sigma_2$, except for those that lie on the common boundary graph $\Gamma$, which are counted a single time instead of two. One then has
    \begin{equation*}
        \prod_e \bA_e(\tilde\rho,\tilde\iota) \prod_v \bA_v(\tilde\rho,\tilde\iota) =\frac{\prod_{e_1,v_1} \bA_{e_1}(\tilde\rho_1,\tilde\iota_1) \bA_{v_1}(\tilde\rho_1,\tilde\iota_1) \prod_{e_2,v_2} \bA_{e_2}(\tilde\rho_2,\tilde\iota_2)\bA_{v_2}(\tilde\rho_2,\tilde\iota_2) }{\prod_{e,v\in \Gamma} \bA_e(\tilde\rho,\tilde\iota) \bA_v(\tilde\rho,\tilde\iota)},
    \end{equation*}
    so that one has, for $\partial \Sigma_1=\Gamma_1 \cup \bar\Gamma$ and $\partial \Sigma_2=\Gamma \cup \Gamma_2$,
    \begin{align*}
            \bA^0(\Psi)
        & = \frac{\omega(\Sigma) \prod_f \bA_f(\tilde\rho) \prod_e\bA_e(\tilde\rho,\tilde\iota) \prod_v \bA_v(\tilde\rho,\tilde\iota)}{\prod_{e,v\in\Gamma_1}  \bA_{f_e}(\tilde\rho) \bA_{e_v}(\tilde\rho,\tilde\iota)} \\
        & = \frac{ \omega(\Sigma_1) \omega(\Sigma_2) \prod_{e_1} \bA_{e_1}(\tilde\rho_1,\tilde\iota_1) \prod_{e_2} \bA_{e_2}(\tilde\rho_2,\tilde\iota_2) }{\prod_{e,v\in\Gamma_1} \bA_{f_e}(\tilde\rho) \bA_{e_v}(\tilde\rho,\tilde\iota) \prod_{e,v\in \Gamma} \bA_{f_e}(\tilde\rho) \bA_{e_v}(\tilde\rho,\tilde\iota)} \\
        & = \bA^0(\Psi_1)\bA^0(\Psi_2).
    \end{align*}
    One can then check directly that \cref{sigmaSalg}
defines a time evolution on the algebra $\C[\cS]$.
\end{proof}

One can again modify the time evolution to include the case of degenerate geometries, as
in the case of the groupoid algebra. In fact, using the same modification of the time evolution by
a term of the form $\chi(\fW(\Psi))^t$ will still give rise to a time evolution.

\begin{cor}\label{timeSalgW}
Setting
\begin{equation}\label{sigmaSalg2}
\sigma_t (f) (\Psi)= \bA^0(\Sigma)^{\ii t} \chi(\fW(\Psi))^t f (\Psi),
\end{equation}
defines a time evolution on the algebra $\C[\cS]$ where degenerate topspin foams
have been included in the category $\cS$.
\end{cor}

\begin{proof}
When gluing two (degenerate) topspin foams along a common boundary, $\Psi=\Psi_1 \cup_\psi \Psi_2$,
the term $\fW(\Psi)$ of \cref{WSigma} spits multiplicatively as $\fW(\Psi)=\fW(\Psi_1)\fW(\Psi_2)$. This
follows from the fact that the spin foams are products $\Gamma \times [0,\epsilon)$ near the boundary,
so that the Wirtinger relation coming from the edges along which the two spin foams are glued
together are trivially satisfied, being just the equality between the $\tilde\sigma_{f_e}$ and
$\tilde\sigma_{f_e'}$ of the two adjacent faces $f_e \supset (-\epsilon,0]\times e$ and $f_e' \supset
e \times [0,\epsilon)$.
\end{proof}

Notice the conceptual difference between the time evolution on the groupoid algebra
$\C[\cG]$ and the one on the semigroupoid algebra $\C[\cS]$ considered here. The
time evolution on the groupoid algebra measures the failure of the spin foam amplitude
to be invariant under covering moves. On the other hand the time evolution on the semigroupoid
algebra $\C[\cS]$ measures how large the amplitude is on different spin foams with the same
spin network boundary.
These very different roles of these two time evolutions explain why, as observed
in \cite{CoveringsCorrespondencesNCG}, when one combines
1-morphisms and 2-morphisms in a 2-category, it is usually a nontrivial problem to
find a time evolution that is simultaneously compatible with both the vertical and the
horizontal composition of 2-morphisms.

\subsection{Equilibrium states}\label{KMSGSsec}

We now consider explicitly the problem of the existence of low temperature KMS
states of Gibbs type for the time evolutions discussed above. We work with the
representations of these algebras described in \cref{repsGSsec}.

\begin{lem}\label{lemHGev}
The infinitesimal generator of the time evolution on the groupoid algebra of
spin foams is the operator $\bH$ acting on the Hilbert space $\cH_\Psi$ by
$\bH \xi(\Psi') = \log \bA(\Psi')\, \xi(\Psi')$, for $\Psi'\sim \Psi$ under the
covering moves equivalence.
\end{lem}

\begin{proof}
One checks that
$\pi_\Psi(\sigma_t(f)) = \ee^{\ii t \bH} \pi_{\Psi}(f) \ee^{-\ii t \bH}$.
\end{proof}

This means that, formally, low temperature KMS states of Gibbs form should be
given by expressions of the form
\begin{equation}\label{KMSGev}
\varphi_\beta (f)=\frac{\Tr(\pi_\Psi(f) \ee^{-\beta \bH})}{\Tr(\ee^{-\beta \bH})}.
\end{equation}
However, this expression makes sense only under the assumption that
$\Tr(\ee^{-\beta \bH})< \infty$ for sufficiently large $\beta$. This brings about
the problem of multiplicities in the spectrum, which we will encounter in
a more dramatic form in the 2-category case we discuss later.

\smallskip

We are working with spin networks and spin foams that are defined, respectively, by embedded graphs in the three-sphere and embedded two-complexes in $S^3\times [0,1]$, but the part of the amplitudes \cref{spinfoamA} consisting of the terms
\begin{equation}\label{amplcomb}
  \bA_{\rm comb}(\Psi):=  \prod_f \bA_f(\tilde\rho_f) \prod_e \bA_e(\tilde\rho_{F(e)}, \tilde\iota_e) \prod_v \bA_v(\tilde\rho_{F(v)},\tilde\iota_{E(v)})
\end{equation}
depends only on the combinatorial structure of the two-complex $\Sigma$, not on the topologically different ways in which it can be embedded in $S^3\times [0,1]$. Thus, it is clear that, if we only
include the terms \eqref{amplcomb} in the amplitude, we would have infinite multiplicities
in the spectrum, as all the possible different topological embeddings of the same combinatorial
$\Sigma$ would have the same amplitude.

\smallskip

The only term that can distinguish topologically inequivalent embeddings, and resolve these infinite multiplicities, is therefore the weight factor $\omega(\Sigma)$. This should be thought of as a generalization of invariants of two-knots embedded in four-dimensional space. This leads
naturally to the following questions on the existence of an invariant with the following properties.

\smallskip

\begin{ques}\label{invariantofgraphs}
Is it possible to construct a topological invariant $\omega(\Gamma)$ of embedded graphs $\Gamma\subset S^3$ with the following properties?
\begin{enumerate}
\item $\omega(\Gamma)$ only depends on the ambient isotopy class of $\Gamma \subset S^3$.
\item The values of $\omega(\Gamma)$ form discrete set of positive real numbers $\{ \alpha(n) \}_{n\in \N}\subset \R^*_+$, which grows exponentially, $\alpha(n) \geq O(\ee^{cn})$, for sufficiently large $n$
and for some $c>0$.
\item The number of embedded graph $\Gamma\subset S^3$ combinatorially equivalent to a given combinatorial graph $\Gamma_0$, with a fixed value of $\omega(\Gamma)$ is finite and grows like
\begin{equation}
\# \{ \Gamma \subset S^3\,\,:\,\, \Gamma \simeq \Gamma_0, \,\, \omega(\Gamma)=\alpha(n) \} \sim O(\ee^{\kappa n }),
\end{equation}
for some $\kappa >0$.
\end{enumerate}
Similarly, is it possible to construct a topological invariant $\omega(\Sigma)$ of embedded two-complexes $\Sigma \subset S^3\times [0,1]$, with the same properties?
\end{ques}

We do not attempt to address this problem in the present paper. However, we see that an invariant $\omega(\Gamma)$, respectively $\omega(\Sigma)$, with the properties listed above will suffice to obtain KMS states of the desired form \eqref{KMSGev}, if the part of the spin foam amplitude \eqref{amplcomb} suffices to distinguish the combinatorics of graphs. More precisely, we have the following situation.

\begin{prop}\label{finmultH}
Consider a time evolution of the form \eqref{Gtsfdyn}, with an amplitude of the form \eqref{spinfoamA}, where the weight factor $\omega(\Sigma)$ has the properties listed in \cref{invariantofgraphs}, and the combinatorial part of the amplitude \eqref{amplcomb} also has the property that the values grow exponentially and that each value
is assumed by only a finite number of combinatorially different $\Sigma$, which grows at most exponentially. Then the operator $\bH$ has the property that, for sufficiently large $\beta >0$,
$\Tr(\ee^{-\beta \bH})<\infty$.
\end{prop}

\begin{proof}
We denote by
$$ N^\omega(n) :=\#\{ \Sigma \subset S^3\times [0,1]\,\,:\,\, \Sigma \simeq \Sigma_0, \,\, \omega(\Gamma)=\alpha(n) \}. $$
Also, if $\{ A_n \}_{n\in \N} \subset \R^*_+$ is an enumeration of the discrete set of values of the
combinatorial amplitude $\bA_{\rm comb}$ of \eqref{amplcomb}, we denote by
$$ N^{\bA_{\rm comb}}(n) :=  \#\{ \Psi=(\Sigma,\tilde\rho,\tilde\iota)\,:\, \bA_{\rm comb}(\Psi)=A_n \}. $$
We are assuming that $A_n \geq O(\ee^{cn})$ for sufficiently large $n$ and that
$N^{\bA_{\rm comb}}(n) \leq O(\ee^{\lambda n})$ for some $\lambda >0$
and $N^\omega(n) \leq O(\ee^{\kappa n})$ for some $\kappa >0$.
We then have in this case that the convergence of  the series computing $\Tr(\ee^{-\beta \bH})$ is dominated by the convergence of
$$ \sum_{n,m}  N^\omega(n) N^{\bA_{\rm comb}}(m) (\alpha(n) A_m)^{-\beta} \leq
\sum_{n,m} \exp( \kappa n + \lambda m -\beta c (n+m)),   $$
which converges for $\beta > \max\{\kappa,\lambda\}/c$. Notice that here we do not
have to worry about the additional presence of the topological labels $\sigma_i$, since
we are looking at the representation on the Hilbert space $\cH_\Psi$ spanned by
topspin foams equivalent to a given $\Psi$, so the $\sigma_i$ on each element in
the equivalence class are uniquely determined and they do not contribute further
multiplicities to the spectrum.
\end{proof}

\section{A 2-category of spin foams and fibered products}\label{fiberSec}

In this section we give a construction of a 2-category of topspin foams where
the 1-morphisms will still be defined by pairs of topspin networks
that are equivalent under covering moves as in \cref{GgrpdSec},
but the horizontal composition will
be different. In particular, the 1-morphisms will no longer be invertible and the
resulting category will no longer be the groupoid of the equivalence relation
generated by covering moves. However, the new composition product
will be constructed in such a way as to correspond to the fibered product of
the branched coverings. This will provide a composition product analogous
to the composition of correspondences used in $KK$-theory. Moreover, the
topspin foam cobordisms as in \cref{ScatSec} will give the 2-morphisms
of this 2-category.

\subsection{$KK$-classes, correspondences, and $D$-brane models}

The construction of \cite{LongIndexTheoremFoliations} of geometric correspondences
realizing KK-theory classes shows that, given manifolds $X_1$ and
$X_2$, classes in $KK(X_1,X_2)$ are realized by geometric data $(Z,E)$
of a manifold $Z$ with submersions $X_1 \leftarrow Z \rightarrow X_2$
and a vector bundle $E$ on $Z$. We write $kk(Z,E)$ for the element
in $KK(X_1,X_2)$ determined by the data $(Z,E)$.
It is shown in \cite{LongIndexTheoremFoliations} that the Kasparov product
$x \circ y \in KK(X_1,X_3)$, for $x=kk(Z,E)\in KK(X_1,X_2)$ and
$y=kk(Z',E')\in KK(X_2,X_3)$, is given by the fibered product
$x\circ y = kk(Z\circ Z',E\circ E')$, where
$$ Z\circ Z'=Z\times_{X_2} Z' \ \ \text{ and } \ \
E\circ E'= \pi_1^* E \times \pi_2^* E'. $$

This geometric construction of KK-theory classes in terms
of geometric correspondences, with the Kasparov product
realized by a fibered product, was recently used extensively
in the context of $D$-brane models in high-energy physics \cite{NCCorrespondencesDBranes}.

The use of $KK$-theory in string theory arises from the fact that
$D$-brane charges are classified topologically by $K$-theory classes,
while $D$-branes themselves are topological $K$-cycles, which define
Fredholm modules and $K$-homology classes.  Thus, $KK$-theory
provides the correct bivariant setting that combines $K$-theory and
$K$-homology. As explained in \cite{NCCorrespondencesDBranes}, the Kasparov product,
realized geometrically as the fibered product of correspondences,
then gives a kind of Fourier--Mukai transform, which is used to provide
a treatment for $T$-duality based on noncommutative geometry. More
generally, for $KK$-classes for noncommutative algebras of open
string fields, the $KK$-theory product should provide the
operator product expansion on the underlying open string vertex
operator algebras.

We are dealing here with a very different type of high-energy physics
model, based on the spin networks and spin foams of loop quantum
gravity. However, we show that here also we can introduce a similar
kind of fibered product on spin foams, in addition
to the usual composition product given by gluing spin foams
along a common boundary spin network.

When we consider spin foams without additional data of matter and
charges, the fibered product will be only of the three-manifolds and
cobordisms defined as branched coverings by the topological
data of the topspin foams and topspin networks. However, if
one further enriches the spin networks and spin foams with
matter data as we outline in the last section, then the fibered
product will also involve bundles over the three-manifolds and
four-manifolds, and will be much more similar to the one considered
in the $D$-branes setting.

\subsection{A 2-category of low dimensional geometries}\label{sec:twoCatOfLowDimGeometries}

We recall here a construction
introduced in \cite{CoveringsCorrespondencesNCG}, for which we give a reinterpretation in terms of
spin networks and spin foams.

In the 2-category of low dimensional geometries constructed in \cite{CoveringsCorrespondencesNCG},
one considers as objects the embedded graphs (up to ambient isotopy) $\Gamma \subset S^3$.

The 1-morphisms between two embedded graphs $\Gamma, \Gamma'\in S^3$
are given by three-manifolds that can be realized as branched coverings of $S^3$
$$\Mor_\cC(\Gamma, \Gamma')=\{ \Gamma  \subset \hat\Gamma \subset S^3 \stackrel{p}{\leftarrow} M \stackrel{p'}{\rightarrow} S^3 \supset \hat\Gamma' \supset \Gamma' \}, $$
with branch locus given by embedded graphs $\hat\Gamma$ and $\hat \Gamma'$, respectively containing $\Gamma$ and $\Gamma'$ as subgraphs.

The composition of 1-morphisms
\begin{equation*}
    \Gamma  \subset \hat\Gamma \subset S^3 \stackrel{p}{\leftarrow} M \stackrel{p_1'}{\rightarrow} S^3 \supset \hat\Gamma_1' \supset \Gamma'
    \qquad \text{and} \qquad
    \Gamma' \subset \hat\Gamma_2' \subset S^3 \stackrel{p_2'}{\leftarrow} M' \stackrel{p''}{\rightarrow} S^3 \supset \hat\Gamma'' \supset \Gamma''
\end{equation*}
is given by the fibered product
\begin{equation*}
    \tilde M = M \times_{S^3} M' = \{ (x,y) \in M \times M' : p_1'(x) = p_2'(y) \},
\end{equation*}
which is also a branched cover
\begin{equation*}
    \Gamma \subset \hat\Gamma \cup p((p_1')^{-1}(\hat\Gamma_2')) \subset S^3 \leftarrow \tilde M
    \to
    S^3 \supset \hat\Gamma'' \cup p''((p_2')^{-1}(\hat\Gamma_1')) \supset \Gamma''.
\end{equation*}

The 2-morphisms are given by branched cover cobordisms
\begin{equation*}
    \Sigma \subset \hat\Sigma \subset S^3\times [0,1] \leftarrow W \to S^3 \times [0,1] \supset \hat \Sigma' \supset \Sigma'
\end{equation*}
with
\begin{equation*}
    \partial W = M_0 \cup \bar M_1, \quad
    \partial \Sigma = \Gamma_0 \cup \bar \Gamma_1, \quad
    \partial \hat\Sigma = \hat\Gamma_1 \cup \overline{\hat\Gamma'_1}, \quad
    \partial \Sigma' = \Gamma_0' \cup \bar \Gamma_1', \quad
    \partial \hat\Sigma' = \hat\Gamma'_0 \cup \overline{\hat\Gamma'_1},
\end{equation*}
where $\Sigma, \Sigma', S, S'$ are two-complexes embedded in $S^3\times [0,1]$. The sphere $S^3$, seen as a trivial covering of itself over the empty graph, gives the unit for composition.

The vertical composition is given by gluing cobordisms along a common boundary,
$$ W_1 \bullet W_2 = W_1 \cup_M W_2, $$
while the horizontal composition is again given by the fibered product
$$ W_1 \circ W_2 = W_1 \times_{S^3 \times [0,1]} W_2, $$
which is a branched cover of $S^3\times [0,1]$, with branch loci
\begin{equation*}
    S \cup p((p_1')^{-1}(S_2')) \qquad \text{and} \qquad S'' \cup p''((p_2')^{-1}(S_1')).
\end{equation*}
The cylinder $S^3\times [0,1]$ with $S=\emptyset$ gives the unit for horizontal composition, while
the identity 2-morphisms for vertical composition are the cylinders $M\times [0,1]$.

All the maps here are considered in the PL category. In dimension three and four, one can
always upgrade from PL to smooth, so we can regard this as a description of smooth
low dimensional geometries.
More precisely, one should also allow for orbifold geometries as well as smooth geometries in the
three-manifolds $M$ and four-manifolds $W$, which can be obtained as fibered products as above.

\subsection{A 2-category of topspin networks and topspin foams}\label{2catLQGsec}

The 2-category  we introduce here is based
on the construction given in \cite{CoveringsCorrespondencesNCG} of a 2-category of low dimensional
geometries, recalled here in \cref{sec:twoCatOfLowDimGeometries} above.
We call the resulting 2-category the \textit{loop quantum gravity 2-category}, because
of the relevance of spin networks and spin foams to loop quantum gravity,
and we denote it with the notation $\cL(G)$, where $G$ stands for the
choice of the Lie group over which to construct spin networks and spin foams.

\begin{thm}\label{2catLQG}
    The 2-category of low dimensional geometries constructed in \cite{CoveringsCorrespondencesNCG}
    extends to a 2-category $\cL(G)$ of  topspin networks and topspin foams
    over a Lie group $G$.
\end{thm}

\begin{proof}
    We show how to extend the definition of objects, 1-morphisms and 2-morphisms, given in \cite{CoveringsCorrespondencesNCG} and recalled in \cref{sec:twoCatOfLowDimGeometries} above, to the rest of the data of spin foams and networks, consistently with compositions and with the axioms of 2-categories.

    \medskip
    \textbf{Objects:}
    Objects in the 2-category $\cL(G)$ are topspin networks $\psi = (\Gamma,\rho,\iota,\sigma)$ over $G$, as in \cref{def:topspinNetwork}.

    \medskip
    \textbf{1-morphisms:}
    Given two objects $\psi$ and $\psi'$, the 1-morphisms in $\Mor_{\cL(G)}(\psi,\psi')$ are pairs of topspin networks $\hat\psi=(\hat\Gamma,\hat\rho,\hat\iota,\hat\sigma)$ and $\hat\psi' = (\hat\Gamma',\hat\rho',\hat\iota',\hat\sigma')$ with the property that $\Gamma \subset \hat\Gamma$ and $\Gamma'\subset\hat\Gamma'$ as subgraphs, with $\hat\rho |_\Gamma =\rho$, $\hat\iota |_\Gamma = \iota$, and for any planar diagrams $D(\hat\Gamma)$ the labels $\hat\sigma_i$ of the strands, determined by the representation
    $\hat\sigma : \pi_1(S^3\smallsetminus \hat\Gamma)\to S_n$ satisfy $\hat\sigma_i =\sigma_i$ for all strands that are in the corresponding diagram $D(\Gamma)$. The conditions relating $\hat\psi'$ and $\psi'$ are analogous. We use the notation $\psi \subset \hat \psi$ and $\psi'\subset \hat\psi'$ to indicate that the conditions above are satisfied.

    An equivalent description of 1-morphisms $\Mor_{\cL(G)}(\psi,\psi')$ is in terms of branched coverings. Namely, a 1-morphism consists of a closed, smooth (PL) three-manifold $M$ with two branched covering maps to $S^3$ branched over graphs $\hat\Gamma$ and $\hat\Gamma'$, respectively containing $\Gamma$ and $\Gamma'$ as subgraphs. We write this as
    \begin{equation*}
        \Gamma \subset \hat\Gamma \subset S^3 \stackrel{p}{\leftarrow} M \stackrel{p'}{\rightarrow} S^3 \supset \hat\Gamma' \supset \Gamma'.
    \end{equation*}
    Additionally, the graphs $\hat\Gamma$ and $\hat\Gamma'$ carry spin network data $(\hat\rho,\hat\iota)$ and $(\hat\rho',\hat\iota')$ compatible with the data on $\psi$ and $\psi'$. In particular, the data $(\hat\Gamma,\hat\rho,\hat\iota,\hat\sigma)$ and $(\hat\Gamma',\hat\rho',\hat\iota',\hat\sigma')$ have to satisfy the compatibility condition under covering moves, possibly after a stabilization to make the two coverings of the same order.

    One can also equivalently describe 1-morphisms as pairs $(\hat\Gamma,\hat\rho,\hat\iota,\hat\sigma)$ and $(\hat\Gamma',\hat\rho',\hat\iota',\hat\sigma')$ related (after stabilization) by covering moves compatibility, each endowed with a marked subgraph, $\Gamma$ and $\Gamma'$ respectively.

    In the following, we use the notation
    \begin{equation*}
        {}_{\psi\subset \hat\psi}M_{\hat\psi'\supset \psi'} \in \Mor_{\cL(G)}(\psi,\psi')
    \end{equation*}
    to indicate 1-morphism. We refer to $\psi$ and $\psi'$ as the source and target topspin networks and to $\hat\psi$ and $\hat\psi'$ as the branching topspin networks. For simplicity of notation, we sometimes write only ${}_\psi M_{\psi'}$ if we wish to emphasize the source and target, or ${}_{\hat\psi} M_{\hat\psi'}$ if we emphasize the branch loci. The distinction in these cases will be clear from the context.

    \medskip
    \textbf{2-morphisms:}
    Given two 1-morphisms
    \begin{equation*}
        \phi_0 = {}_{\hat\psi_0} M_{\hat\psi_0'}
        \quad \text{and} \quad
        \phi_1={}_{\hat\psi_1} M'_{\hat\psi'_1},
    \end{equation*}
    with the same source and target $\psi \subset \hat\psi_i$ and $\psi'\subset \hat\psi_i'$, the 2-morphisms
    $\Phi \in \MorTwo_{\cL(G)}(\phi_0,\phi_1)$ consist of a smooth (PL) four-manifold  $W$ with boundary
    $\partial W = M\cup \bar M'$ with two branched covering maps to $S^3\times [0,1]$ with
    branch loci $\hat\Sigma$ and $\hat\Sigma'$ two embedded two-complexes with
    $$ \partial \hat\Sigma =\hat\Sigma\cap (S^3\times \{ 0,1 \}) = \hat\Gamma_0 \cup \overline{\hat\Gamma_1} \ \ \ \text{ and } \ \ \
    \partial\hat\Sigma'=\hat\Sigma'\cap (S^3\times \{ 0,1 \}) =\hat\Gamma_0' \cup
    \overline{\hat\Gamma_1'} . $$
    The two-complexes $\hat\Sigma$ and $\hat\Sigma'$ are endowed with topspin foam data
    $\Psi=(\hat\Sigma,\tilde\rho,\tilde\iota,\tilde\sigma)$ and $\Psi'=(\hat\Sigma',\tilde\rho',\tilde\iota',\tilde\sigma')$, such that the induced topspin network data $\hat\psi$ and $\hat\psi'$,
    on the boundary graphs $\hat\Gamma$ and $\hat\Gamma'$ respectively, agree with the
    assigned topspin network data $\psi$ and $\psi'$ on the subgraphs $\Gamma$ and $\Gamma'$,
    respectively. The two-complexes $\hat\Sigma$ and $\hat\Sigma'$ contain a marked subcomplex,
    $\Sigma$ and $\Sigma'$, respectively, which is a (possibly nontrivial) cobordism between
    $\Gamma$ and itself, or respectively between $\Gamma'$ and itself. We write this as
    $$ \Sigma\subset \hat\Sigma \subset S^3\times [0,1] \stackrel{q}{\leftarrow} W \stackrel{q'}{\rightarrow}
    S^3\times [0,1] \supset \hat\Sigma' \supset \Sigma' . $$

    Equivalently, a 2-morphism is specified by a pair of topspin foams
    $\Psi=(\hat\Sigma,\tilde\rho,\tilde\iota,\tilde\sigma)$ and
    $\Psi'=(\hat\Sigma',\tilde\rho',\tilde\iota',\tilde\sigma')$ between topspin
    networks $\hat\psi_0$ and $\hat\psi_1$ (respectively $\hat\psi'_0$ and
    $\hat\psi'_1$), which after stabilization
    are related by covering moves for two-complexes and four-manifolds branched coverings,
    together with marked subcomplexes $\Sigma$ and $\Sigma'$ that are coboundaries
    between $\Gamma$ and itself and between $\Gamma'$ and itself.

    We use the notation
    $$ {}_{\Psi} W_{\Psi'} \in \Mor_{\cL}^{(2)}({}_{\psi\subset\hat\psi_0}M_0{}_{\hat\psi'\supset\psi_0'}\,\, , \,\,
    {}_{\psi\subset\hat\psi_1}M_1{}_{\hat\psi'\supset\psi'_1}) $$
    to indicate 2-morphisms.

    \medskip
    \textbf{Horizontal composition:}
    The composition of 1-morphisms and the horizontal composition of 2-morphisms in $\cL(G)$ follows the analogous compositions of branched cover cobordisms by fibered products in \cite{CoveringsCorrespondencesNCG}. Namely, suppose we are given objects $\psi=(\Gamma,\rho,\iota,\sigma)$, $\psi'=(\Gamma',\rho',\iota',\sigma')$, and $\psi''=(\Gamma'',\rho'',\iota'',\sigma'')$ and 1-morphisms given by three-manifolds
    $$ \Gamma\subset\hat\Gamma\subset S^3 \stackrel{p}{\leftarrow} M \stackrel{p_1}{\rightarrow} S^3 \supset \hat\Gamma'_1 \supset \Gamma' $$
    $$ \Gamma'\subset\hat\Gamma'_2\subset S^3 \stackrel{p_2}{\leftarrow} M \stackrel{p'}{\rightarrow} S^3 \supset \hat\Gamma'' \supset \Gamma'', $$
    the composition is given by
    $$  {}_{\psi\subset \hat\psi}M_{\hat\psi'_1\supset \psi'} \circ  {}_{\psi'\subset \hat\psi'_2}M'_{\hat\psi''\supset \psi''} = {}_{\psi \subset \tilde\psi} \tilde M_{\tilde\psi' \supset \psi''}, $$
    where $\tilde\psi$ is the topspin network with embedded graph
    $\tilde\Gamma =\hat\Gamma \cup p p_1^{-1}(\hat\Gamma'_2)$, with labeling $\hat\rho$ on the edges of $\hat\Gamma$ and $\hat\rho'_2$ on the edges of $p p_1^{-1}(\hat\Gamma'_2)$, with
    the convention that on edges common to both graphs one assigns the tensor product of the
    two representations. We use the notation $\hat\rho\cup\hat\rho_2'$ to indicate this.
    The multipliers are assigned similarly. The
    representation of $p_1(S^3 \smallsetminus \tilde\Gamma) \to S_{nm}$, for
    $n$ the order of the covering map $p$ and $m$ the order of $p_2$, is the
    one that defines the branched covering space $\Pi:\tilde M\to S^3$ given by the
    composition of the projection $p_1$ of the fibered product $\tilde M=M\times_{S^3}M'$
    on the first factor, followed by $p$. Similarly, for the data of the spin network
    $\tilde\psi'$. The horizontal composition of the 2-morphisms is defined in the same
    way, using the fibered product of two branched cover cobordisms $W$ and $W'$,
    as in the case of the 2-category of \cite{CoveringsCorrespondencesNCG}, with resulting topspin foams
    $\tilde\Psi$ with two-complex $\hat\Sigma\cup q q_2^{-1}(\hat\Sigma_2')\supset \Sigma$
    and with the remaining data assigned as in the case of topspin networks,
    and similarly for $\tilde\Psi'$ with two-complex $\hat\Sigma''\cup q' q_1^{-1}(\hat\Sigma_1')\supset \Sigma''$.

    \medskip
    \textbf{Vertical composition:}
    The vertical composition of 2-morphisms is the usual
    composition of spin foams along a common boundary spin network, extended
    to include the topological data by performing the corresponding gluing of the
    four-manifolds $W$ along the common boundary three-manifold. In terms of diagrams
    $D(\Sigma)$ one glues together two such diagrams along a common boundary
    $D(\Gamma)$ with matching labels $(\rho,\iota,\sigma)$ at the faces, edges, and
    strands of faces that emerge from the edges, vertices, and strands of edges of
    the diagram $D(\Gamma)$. These are unique by the assumption that spin foams
    are products $\Gamma \times [0,\epsilon)$ near the
    boundary.

    The associativity of both the vertical and the horizontal composition of
    2-morphisms follows as in the corresponding argument given in \cite{CoveringsCorrespondencesNCG}
    for the 2-semigroupoid algebra of the 2-category of low dimensional geometries.
    The unit for composition of 1-morphisms is the empty graph in the three-sphere
    $S^3$, which corresponds to the trivial unbranched covering of the three-sphere by
    itself, and the unit for composition of 2-morphisms is the trivial cobordism
    $S^3\times [0,1]$ with $\Sigma=\emptyset$, the unbranched trivial covering
    of $S^3\times [0,1]$ by itself, as shown in \cite{CoveringsCorrespondencesNCG}.  The horizontal and
    vertical composition of 2-morphisms satisfy the compatibility condition \cref{verthorcompat},
    $$ (W_1\times_{S^3\times [0,1]} W_1') \cup_{M\times_{S^3} M'} (W_2\times_{S^3\times [0,1]}
    W_2') = (W_1\cup_M W_2)\times_{S^3\times [0,1]} (W_1'\cup_{M'} W_2'), $$
    hence the vertical and horizontal products in $\cL(G)$ satisfy the compatibility
    condition \cref{2prods}.
\end{proof}

\subsection{The case of bicategories: loop states and spin networks}\label{sec:bicategories}

When one realizes three-manifolds as branched covers of the three-sphere, the Hilden--Montesinos
theorem \cite{ThreeFoldBranchedCoverings,ThreeManifoldsAsBranchedCovers} ensures that it is in fact always possible to arrange so that the branch locus is an
embedded knot (or link) and the covering is of order $3$. However, when considering
the fibered products as above, one does not necessarily have transversality: even assuming
that the branch covering maps one begins with have branch loci that are knots or links,
only after deforming the maps by a homotopy one can ensure that the branch loci of the
fibered product will still be links. This is the reason for not restricting only to the case of
links in \cite{CoveringsCorrespondencesNCG}. However, if one wishes to only consider branch loci that are links,
or embedded surfaces in the four-dimensional case, one can do so at the cost of replacing
associativity of the composition of 1-morphisms by associativity only up to homotopy.
This corresponds to replacing the 2-category of low dimensional geometries described
above with the weaker notion of a {\em bicategory}.

\smallskip

A bicategory $\cB$, as in the case of a 2-category, has objects $\cB^{(0)}= \Obj(\cB)$,
1-morphisms $\cB^{(1)}= \bigcup_{x,y\in \cB^{(0)}} \Mor_{\cB}(x,y)$ and 2-morphisms
$\cB^{(2)}= \bigcup_{\phi,\psi \in \cB^{(1)}} \MorTwo_{\cB}(\phi,\psi)$. The composition of
1-morphisms $\circ: \Mor_{\cB}(x,y)\times \Mor_{\cB}(y,z) \to \Mor_{\cB}(x,z)$ is
only associative up to natural isomorphisms (the {\em associators})
$$ \alpha_{x,y,z,w} : \Mor_{\cB}(x,z) \times \Mor_{\cB}(z,w) \to
\Mor_{\cB}(x,y)\times \Mor_{\cB}(y,w) $$
$$ (\phi_1 \circ \phi_2)\circ \phi_3 = \alpha_{x,y,z,w} (\phi_1 \circ (\phi_2 \circ \phi_3)) ,$$
with $\phi_1 \in \Mor_{\cB}(z,w)$, $\phi_2\in \Mor_{\cB}(y,z)$, and $\phi_3\in \Mor_{\cB}(x,y)$.
For every object $x\in \cB^{(0)}$ there is an identity morphisms $1_x \in \Mor_{\cB}(x,x)$.
This acts as the unit for composition, but only up to canonical isomorphism,
$$  \alpha_{x,y,1} (\phi \circ 1_x) = \phi =\alpha_{1,x,y}(1_y \circ \phi). $$
The associators
satisfy the same compatibility condition of monoidal categories, namely the
pentagonal identity
$$ \phi \circ (\psi\circ (\eta \circ \xi)) =
1 \circ \alpha \, (\phi \circ ((\psi\circ\eta)\circ \xi)) =
1 \circ \alpha \, ( \alpha \, ( (\phi \circ (\psi\circ \eta))\circ \xi)) = $$ $$
1 \circ \alpha \, ( \alpha \, ( \alpha \circ 1 \, (((\phi \circ \psi)\circ \eta)\circ \xi)))=
\alpha \, ( (\phi\circ\psi)\circ (\eta\circ\xi))  =
\alpha \, ( \alpha \, ( ((\phi\circ \psi)\circ \eta)\circ \xi)) , $$
and the triangle relation
$$ \phi \circ \psi = 1\circ \alpha (\phi \circ (1 \circ \psi)) = 1\circ \alpha( \alpha ( (\phi \circ 1)\circ \psi ))
= \alpha \circ 1 ((\phi \circ 1)\circ \psi ). $$
Vertical composition of 2-morphisms is associative and unital, as in the 2-category case,
while horizontal composition of 2-morphisms follows the same rules as composition of
1-morphism and is associative and unital only up to canonical isomorphisms.

\smallskip

One can correspondingly modify the notion of 2-semigroupoid algebra
discussed above and replace it with a weaker notion of bi-associative algebra,
where the vertical product $\bullet$ is still associative, while the horizontal product
$\circ$ is no longer associative, with the lack of associativity controlled by
associators.

\smallskip

We will not get into more details here on this possible variant, since it is less directly
relevant to the loop quantum gravity context. We remark, however, that a similar
type of nonassociative algebras (without the associative
vertical product), has been already widely used in
the context of $T$-duality in string theory \cite{NonAssocToriTDuality}.

Replacing associativity of the composition of geometric correspondences
by associativity up to homotopy that produce associators is natural if one
wants to work under the assumption of transversality, in the KK-theory approach
developed in \cite{NCCorrespondencesDBranes}.

In loop quantum gravity it is now customary to describe states in terms
of spin networks, though it is known that these can equivalently be described
as linear combinations of ``loop states," the latter being associated to knots
and links, see \S 6.3.2 of \cite{RovelliQG}. In the point of view we describe here,
restricting to the use of loop states corresponds to a loss of associativity in
the convolution algebra of geometries, which corresponds to working with a
bicategory of low dimensional geometries instead of a 2-category. The
main point where the difference between using spin networks as opposed
to loop states arises in our setting is in the composition of correspondences
via the fibered product. The other composition, by gluing cobordisms
along their boundaries, is unaffected by the difference.

\subsection{Convolution algebras of topspin networks and topspin foams}

We associate to the 2-category $\cL(G)$ described in \cref{2catLQGsec} an associative convolution algebra $\cA_{\rm tsn}(G)$ of topological spin networks, which is the semigroupoid algebra of the composition of 1-morphisms, and a 2-semigroupoid algebra $\cA_{\rm tsf}(G)$ of topspin foams, which is the 2-semigroupoid algebra (in the sense of \cref{2algA}) associated to the 2-category $\cL(G)$. From hereon, we suppress for simplicity the explicit dependence on $G$.

The algebra $\cA_{\rm tsn}$ is generated algebraically by finitely supported functions $f=\sum_\phi a_\phi \delta_\phi$ on the set of 1-morphisms $\phi \in \Mor_{\cL(G)}$. Recalling the description of 1-morphisms in terms of branched coverings, we can thus write such functions as
\begin{equation*}
    f\bigl({}_{\psi\subset \hat\psi} M_{\hat\psi'\supset \psi'}\bigr),
\end{equation*}
i.e.\ as functions of three-manifolds with branched covering maps to $S^3$ and with spin network data on the branch loci. The associative, noncommutative convolution product then follows the composition of 1-morphisms:
\begin{equation*}
    (f_1\star f_2)\bigl({}_{\psi \subset \tilde\psi} \tilde M_{\tilde\psi' \supset \psi''}\bigr)
    =
    \sum f_1\bigl({}_{\psi\subset \hat\psi}M_{\hat\psi'_1\supset \psi'}\bigr) \,
         f_2\bigl({}_{\psi'\subset \hat\psi'_2}M'_{\hat\psi''\supset \psi''}\bigr),
\end{equation*}
where the sum is over all the possible ways to write the 1-morphism ${}_{\psi \subset \tilde\psi} \tilde M_{\tilde\psi' \supset \psi''}$ as a composition
\begin{equation*}
          {}_{\psi\subset \hat\psi}M_{\hat\psi'_1\supset \psi'}
    \circ {}_{\psi'\subset \hat\psi'_2}M'_{\hat\psi''\supset \psi''}
    =     {}_{\psi \subset \tilde\psi} \tilde M_{\tilde\psi' \supset \psi''}
\end{equation*}
of two other 1-morphisms, given by the fibered product of the three-manifold as described in \cref{2catLQGsec}.

\medskip

One considers then all representations of this algebra on the Hilbert space $\cH_{\psi'}$ spanned by all the 1-morphisms with range $\psi'$. Equivalently, it is spanned by the topspin networks that are equivalent under branched covering moves (after stabilization) to $\hat\psi' \supset \psi'$, which is to say, by the three-manifolds $M$ with branched covering maps
\begin{equation*}
    \Gamma \subset \hat\Gamma \subset S^3 \stackrel{p}{\leftarrow} M \stackrel{p'}{\rightarrow} S^3 \supset \hat\Gamma' \supset  \Gamma'
\end{equation*}
with spin network data $(\rho,\iota)$ obtained from the data $(\rho',\iota')$ of the given $\psi'$ by the consistency conditions under covering moves. Thus, we can write in the standard way an element
$\xi \in \cH_{\psi'}$, which is a combination of the basis vectors $|{}_{\psi} M_{\psi'}\rangle$, as a
function $\xi ({}_{\psi} M_{\psi'})$. The representation $\pi_{\psi'}$ of the algebra $\cA_{\rm tsn}$ on
$\cB(\cH_{\psi'})$ is then given by
\begin{equation*}
      \pi_{\psi'}(f) \xi({}_{\psi} M_{\psi'})
    = \sum f\bigl({}_{\psi\subset \hat\psi} M_1{}_{\hat\psi''_1\supset\psi''}\bigr) \,
           \xi\bigl({}_{\psi''\subset \hat\psi''_2} M_2 {}_{\hat\psi'\supset\psi'}\bigr),
\end{equation*}
with the sum over all the ways of writing ${}_{\psi} M_{\psi'}$ as a composition of ${}_{\psi} M_1{}_{\psi''}$ and ${}_{\psi''} M_2 {}_{\psi'}$.

\smallskip

The convolution algebra $\cA_{\rm tsn}$ is not involutive, but one can obtain a $C^*$-algebra containing $\cA_{\rm tsn}$ by taking the $C^*$-subalgebra of $\cB(\cH_{\psi'})$ generated by the operators $\pi_{\psi'}(f)$ with $f\in \cA_{\rm tsn}$. This general procedure for semigroupoid algebras was described in \cite{CoveringsCorrespondencesNCG} in terms of creation-annihilation operators.

\medskip

One obtains in a similar way the 2-semigroupoid algebra $\cA_{\rm tsf}$. Namely, one considers functions with finite support $f=\sum a_\Phi \delta_\Phi$, with $\Phi \in \MorTwo_{\cL(G)}$, with two associative convolution products corresponding to the horizontal and vertical composition of 2-morphisms,
\begin{equation*}
      (f_1\circ f_2)({}_\Psi \tilde W_{\Psi''})
    = \sum f_1({}_\Psi W {}_{\Psi'}) \, f_2({}_{\Psi'} W'_{\Psi''}),
\end{equation*}
with the sum taken over all ways of obtaining the 2-morphism $\Phi={}_\Psi \tilde W_{\Psi''}$ as the horizontal composition of 2-morphisms ${}_\Psi W {}_{\Psi'}$ and ${}_{\Psi'} W'_{\Psi''}$, given by the fibered product of the corresponding branched coverings. The other associative product is given by
\begin{equation*}
      (f_1\bullet f_2)({}_\Psi W_{\Psi'})
    = \sum f_1({}_{\Psi_1} W_1{}_{\Psi'_1}) \, f_2({}_{\Psi_2} W_2{}_{\Psi'_2}),
\end{equation*}
with the sum over all decompositions of the 2-morphism ${}_\Psi W_{\Psi'}$ as a vertical composition ${}_\Psi W_{\Psi'}={}_{\Psi_1} W_1{}_{\Psi'_1}\bullet {}_{\Psi_2} W_2{}_{\Psi'_2}$ given by gluing the cobordisms $W_1$ and $W_2$ together along their common boundary.

\medskip

\section{Topspin foams and dynamics}

In \cite{CoveringsCorrespondencesNCG} several examples were given of time evolutions on the algebra of the 2-category of low dimensional geometries that were either compatible with the horizontal or with the vertical time evolution, but not with both simultaneously. We see here that one can, in fact, construct several examples of time evolution on the 2-semigroupoid algebra $\cA_{\rm tsf}$, that are compatible with both the vertical and the horizontal associative products. These time evolutions therefore capture the full underlying algebraic structure coming from the 2-category.

We first present a very basic and simple example of a dynamics on $\cA_{\rm tsn}$ and a variant of the same time evolution on $\cA_{\rm tsf}$, which is compatible with both the vertical and the horizontal composition of 2-morphisms, and then we give more general and more interesting examples.

\begin{prop}\label{tsnevn}
One obtains a time evolution on $\cA_{\rm tsn}$ by setting
\begin{equation}\label{nsigma}
    \sigma_t^{\bO}(f)= \bO^{\ii t} f,
\end{equation}
where $\bO(f)({}_{\psi} M_{\psi'}) = n \, f({}_{\psi} M_{\psi'})$ multiplies by the order $n$ of the covering $p : M \to S^3$ branched over $\hat\Gamma\supset \Gamma$, with $\Gamma$ the graph of the spin network $\psi$. Similarly, one obtains a time evolution on $\cA_{\rm tsf}$, simultaneously
compatible with the vertical and horizontal time evolution, by setting
\begin{equation}\label{nsigmaf}
    \sigma^{\bO}_t (f) = \bO^{F \ii t}\, f,
\end{equation}
where $F$ is the number of faces of the source $\Sigma$, and $\bO$ multiplies by the order of the covering $q:W \to S^3\times [0,1]$ branched along $\hat\Sigma\supset \Sigma$.
\end{prop}

\begin{proof}
    The fact that $f\mapsto \bO^{\ii t} f$ is a time evolution was shown already in \cite{CoveringsCorrespondencesNCG}, and one sees that, under composition of 1-morphisms, if the multiplicities of the covering maps $p$, $p_1$, $p_2$, $p'$ in
    \begin{equation*}
        \Gamma  \subset \hat\Gamma \subset S^3 \stackrel{p}{\leftarrow} M \stackrel{p_1}{\rightarrow} S^3 \supset \hat\Gamma_1 \supset \Gamma'
        \quad \text{and} \quad
        \Gamma' \subset \hat\Gamma_2 \subset S^3 \stackrel{p_2}{\leftarrow} M' \stackrel{p'}{\rightarrow} S^3 \supset \hat\Gamma'' \supset \Gamma''
    \end{equation*}
    are $n$, $n_1$, $n_2$, $n'$, respectively, then the multiplicities of the covering maps for the fibered product
    \begin{equation*}
        \Gamma \subset \hat\Gamma \cup p p_1^{-1}(\hat\Gamma_2) \subset S^3 \leftarrow \tilde M \rightarrow S^3 \supset \hat\Gamma'' \cup p'' p_2^{-1} (\hat\Gamma_1) \supset \Gamma''
    \end{equation*}
    are $n n_2$ and $n_1 n'$, respectively. Thus, one has $\bO^{\ii t} (f_1\star f_2)=\bO^{\ii t}(f_1) \star \bO^{\ii t}(f_2)$.

    The case of $\cA_{\rm tsf}$ is analogous. For the horizontal composition we have the same property that the orders of coverings multiply, so that we obtain
    \begin{equation*}
        \bO^{F \ii t}(f_1 \circ f_2)(\tilde W)= \sum n^{F \ii t} f_1(W)\, n_2^{F \ii t} f_2(W') =( \bO^{F \ii t}(f_1) \circ \bO^{F \ii t} (f_2)) (\tilde W).
    \end{equation*}
    For the vertical product $W=W_1\cup_M W_2$, we have $F(W)=F(W_1)+F(W_2)$, while the order of covering stays the same in the vertical composition. So we obtain
    \begin{equation*}
        \bO^{F \ii t}(f_1 \bullet f_2)(W)= \sum n^{(F(W_1)+F(W_2)) \ii t} f_1(W_1) f_2(W_2)= (\bO^{F \ii t}(f_1) \bullet \bO^{F \ii t}(f_2))(W).
    \end{equation*}
\end{proof}

\subsection{Dynamics from quantized area operators} \label{sec:dynamicsFromQuantizedAreaOperators}

For a three-manifold $M$ realized in two different ways as a covering of $S^3$,
\begin{equation*}
    \Gamma \subset S^3 \stackrel{p}{\leftarrow} M \stackrel{p'}{\rightarrow} S^3 \supset \Gamma',
\end{equation*}
let $S$ and $S'$ be two closed surfaces embedded in $S^3$. Consider the subset of $M$ given by the preimages $p^{-1}(S\cap \Gamma)$ and $(p')^{-1}(S'\cap \Gamma')$.

Suppose we are given two $\mathrm{SU}(2)$-topspin networks $\psi=(\Gamma,\rho,\iota,\sigma)$ and $\psi'=(\Gamma',\rho',\iota',\sigma')$, such that the data $(\Gamma,\sigma)$ and $(\Gamma',\sigma')$ define the same three-manifold $M$ with two branched covering maps to $S^3$ as above. Let $S\subset S^3$ be a closed embedded smooth (or PL) surface. In the generic case where $S$ intersects $\Gamma$ transversely in a finite number of points along the edges of the graph, one has the usual quantized area operator, given by
\begin{equation}\label{qArea}
    A_S \, f({}_\psi M_{\psi'}) = \hbar \left( \sum_{x\in S\cap \Gamma} (j_x (j_x+1))^{1/2} \right) f({}_\psi M_{\psi'}),
\end{equation}
where $f({}_\psi M_{\psi'})$ is a function in the convolution algebra of topspin networks, and $j_x$ is the spin $j_e$ of the $\mathrm{SU}(2)$-representations $\rho_e$ attached to the edge of $\Gamma$ that contains the point $x$.

We can similarly define the operator $\hat A_{S'}$ using labelings of the edges of the graph $\Gamma'$ with the spins $j_{e'}$. The compatibility conditions of \cref{sec:consistencyConditions} for topspin networks $\psi=(\Gamma,\rho,\iota,\sigma)$ and $\psi'=(\Gamma',\rho',\iota',\sigma')$ related by covering moves shows that these two choices are equivalent, when the surfaces $S$ and $S'$ are also related by corresponding moves.

\smallskip

It is convenient to think of the area operator $A_S$ as
\begin{equation*}
    A_S \, f({}_\psi M_{\psi'}) = \hbar \left( \sum_{e \in E(\Gamma)} N^S(e) \, (j_e (j_e+1))^{1/2}  \right) f({}_\psi M_{\psi'}),
\end{equation*}
where $N^S : E(\Gamma) \to \Z$ is a multiplicity assigned to each edge of $\Gamma$ given by the number of points (counted with orientation) of intersection of $e$ with the surface $S$.

This suggests an easy generalization, which gives a quantized area operator associated to each function $N : \bigcup_\Gamma E(\Gamma) \to \Z$ that assigns an integer multiplicity to the edges of all embedded graphs $\Gamma \subset S^3$. The corresponding quantized area operator is given by
\begin{equation}\label{qAreaN}
    A \, f({}_\psi M_{\psi'}) = \hbar \left( \sum_{e \in E(\Gamma)} N(e) \, (j_e (j_e+1))^{1/2}  \right) f({}_\psi M_{\psi'}),
\end{equation}
in analogy with the above. In particular, one can consider such operators associated to
the choice of a subgraph in each graph $\Gamma$, with $N$ the characteristic function $\chi_\Gamma$ of that subgraph.

Recall that, in the definition of the 2-category of low dimensional geometries of \cite{CoveringsCorrespondencesNCG}, the 1-morphisms $\Mor_\cC(\Gamma, \Gamma')$, with $\Gamma$ and $\Gamma'$ embedded graphs in $S^3$, are defined by branched coverings
\begin{equation*}
    \Gamma \subset \hat\Gamma \subset S^3 \stackrel{p}{\leftarrow} M \stackrel{p'}{\rightarrow} S^3 \supset  \hat \Gamma' \supset \Gamma',
\end{equation*}
where the branch loci are graphs $\hat \Gamma$ and $\hat \Gamma'$ that contain the domain and range graphs $\Gamma'$ and $\Gamma'$ as subgraphs. Thus, the corresponding elements of the convolution algebra of topspin networks ${}_\psi M _{\psi'}$ consist of a pair of topspin networks $\psi = (\hat\Gamma,\rho,\iota,\sigma)$ and $\psi' = (\hat\Gamma'',\rho',\iota',\sigma')$, together with a marking of a subgraph $\Gamma \subseteq \hat\Gamma$ and $\Gamma' \subseteq \hat\Gamma'$. This marking can be viewed as a multiplicity function $\chi_\Gamma : E(\hat\Gamma) \to \{ 0,1 \}$, which is the characteristic function of
the subgraph $\Gamma$, and similarly we have a $\chi_{\Gamma'}$ for $\Gamma'$. Thus, a morphism ${}_\psi M_{\psi'}$ in the convolution algebra of topspin networks has two naturally associated quantum area operators, respectively marking the source and target of the morphism. We denote them by $A_\fs$ and $A_\ft$,
\begin{align}
    A_\fs \, f({}_\psi M_{\psi'}) & = \hbar \left( \sum_{e \in E(\hat\Gamma)} \chi_\Gamma(e) \, (j_e (j_e+1))^{1/2}  \right) f({}_\psi M_{\psi'}), \label{qAreaSource} \\
    A_\ft \, f({}_\psi M_{\psi'}) & = \hbar \left( \sum_{e \in E(\hat\Gamma')} \chi_{\Gamma'}(e) \, (j_e (j_e+1))^{1/2}  \right) f({}_\psi M_{\psi'}). \label{qAreaTarget}
\end{align}

Notably, the difference between the eigenvalues of $A_\fs$ and $A_\ft$ on the same eigenfunction measures the difference in quanta of area between the source and target graphs $\Gamma$ and $\Gamma'$ with their labeling by representations $\rho_e$ and $\rho_{e'}$ on the edges. This difference of areas generates a dynamics on the algebra of topspin networks.

\begin{prop}\label{Areatimeev}
    Setting $\sigma_t^A (f) = \exp(\ii t (A_\fs-A_\ft)) f $ defines a time evolution on the     convolution algebra of topspin networks, i.e.\ a one-parameter family $\sigma^A : \R \to \Aut(\cA_{{\rm tsn}})$. In the representation $\pi_{\psi'}$ of $\cA_{{\rm tsn}}$ on the Hilbert space $\cH_{\psi'}$ spanned by 1-morphisms with a given target $\psi'$, the time evolution is generated by $A_\fs$, i.e.\
    \begin{equation*}
        \pi_{\psi'}(\sigma^A_t(f)) = \ee^{\ii t A_\fs} \pi_{\psi'}(f) \ee^{-\ii t A_\fs}.
    \end{equation*}
\end{prop}

\begin{proof}
    We need to check that this time evolution is compatible with the convolution product, i.e.\ that $\sigma_t^A(f_1 \star f_2) = \sigma_t^A(f_1) \star \sigma_t^A(f_2)$. This follows from the fact that the operators $A_\fs$ and $A_\ft$ only depend on the source and range subgraphs $\Gamma$ and $\Gamma'$, and not on the entire branch loci graphs $\hat\Gamma$ and $\hat\Gamma'$. To see this, note that for the composition ${}_\psi \tilde M_{\psi''}$ given by the fibered product of ${}_\psi M_{\psi'}$ and ${}_{\psi'} M'_{\psi''}$, we have branched coverings
    \begin{equation*}
        \Gamma \subset \hat\Gamma \cup p p_1^{-1}(\hat\Gamma_2) \subset S^3 \stackrel{\Pi}{\leftarrow} \tilde M \stackrel{\Pi'}{\rightarrow} S^3 \supset \hat\Gamma'' \cup p'' p_2^{-1}(\hat\Gamma_1) \supset \Gamma'',
    \end{equation*}
    where
    \begin{gather*}
    \Gamma \subset \hat\Gamma \subset S^3 \stackrel{p}{\leftarrow}  M \stackrel{\pi_1}{\rightarrow} S^3 \supset \hat\Gamma_1 \supset \Gamma' \\
    \Gamma' \subset \hat\Gamma_2 \subset S^3 \stackrel{p_2}{\leftarrow}  M' \stackrel{p''}{\rightarrow} S^3 \supset \hat\Gamma'' \supset \Gamma''
    \end{gather*}
    are the branched covering data of the morphisms ${}_\psi M_{\psi'}$ and ${}_{\psi'} M'_{\psi''}$. Thus, we have
    \begin{align*}
            \sigma_t^A(f_1 \star f_2)({}_\psi \tilde{M}_{\psi''})
        & = \exp\left(\ii t \hbar \hspace{-1em} \sum_{e \in \hat\Gamma \cup \hat\Gamma''} (\chi_\Gamma(e) - \chi_{\Gamma''}(e))(j_e (j_e + 1))^{1/2}\right) (f_1 \star f_2) ({}_\psi \tilde{M}_{\psi''}) \\
        & = \sum_{\tilde{M} = M \circ M'} \exp\left(\ii t \hbar \hspace{-1em} \sum_{e \in \hat\Gamma \cup \hat\Gamma_1} (\chi_\Gamma(e) - \chi_{\Gamma'}(e))(j_e (j_e + 1))^{1/2}\right) f_1({}_\psi M_{\psi'}) \\
        & \qquad \qquad \times \exp\left(\ii t \hbar \hspace{-1em} \sum_{e \in \hat\Gamma_2 \cup \hat\Gamma''} (\chi_{\Gamma'}(e) - \chi_{\Gamma''}(e))(j_e (j_e + 1))^{1/2}\right) f_2({}_{\psi'} M_{\psi''}).
    \end{align*}

    Let $\cH_{\psi''}$ denote the Hilbert space spanned by all the 1-morphisms with a given
    target $\psi''$. This means that elements $\xi \in  \cH_{\psi''}$ are square integrable functions
    on the discrete set of all the ${}_{\psi\subset \hat \psi} M_{\hat \psi'' \supset \psi''}$ with
    fixed $\psi''$, corresponding to branched coverings
    $$ \Gamma \subset \hat \Gamma \subset S^3 \stackrel{p}{\leftarrow} M \stackrel{p''}{\rightarrow} S^3
       \supset \hat\Gamma''\supset \Gamma''. $$
    The action of $\cA_{\rm tfn}$ is then given by the representation
    $$ \pi_{\psi''}(f) \xi ({}_{\psi} \tilde M_{\psi''}) = \sum_{\tilde M=M \circ M'} f({}_{\psi\subset \hat\psi}
    M_{\hat \psi'_1 \supset \psi'}) \xi({}_{\psi'\subset \hat\psi'_2} M'_{\hat\psi''\supset \psi''}) . $$
     We then have
    $$ \ee^{-\ii t A_\fs} \xi ({}_{\psi'\subset \hat\psi'_2} M'_{\hat\psi''\supset \psi''}) =
    \ee^{-\ii t \hbar (\sum_{e \in \hat\Gamma_2'}
    \chi_{\Gamma'}(e) (j_e(j_e+1))^{1/2})} \xi ({}_{\psi'\subset \hat\psi'_2} M'_{\hat\psi''\supset \psi''})  , $$
    so that
    $$ \pi_{\psi''}(\sigma^A_t(f)) \xi ({}_{\psi} \tilde M_{\psi''}) =
    \ee^{\ii t (A_\fs-A_\ft)} \pi_{\psi''}(f) \xi ({}_{\psi} \tilde M_{\psi''}) =
    \ee^{\ii t (A_\fs)} \pi_{\psi''}(f) \ee^{-\ii t (A_\fs)}  \xi ({}_{\psi} \tilde M_{\psi''}) . $$
\end{proof}

We now consider the convolution algebra of the 2-category of topspin foams, with both the vertical and the horizontal associative convolution products.

The 2-morphisms in the category of topspin foams are data of a branched cover cobordism $W$, specified by a pair of covering-move equivalent topspin foams $\Psi = (\hat\Sigma,\tilde\rho,\tilde\iota,\tilde\sigma)$ and $\Psi' = (\hat\Sigma',\tilde\rho',\tilde\iota',\tilde\sigma')$, where $\hat\Sigma$ and $\hat\Sigma'$ are two-complexes embedded in $S^3 \times [0,1]$. Since 2-morphisms connect 1-morphisms with same source and target, the spin networks associated to the graphs $\hat\Gamma_0 = \hat\Sigma \cap S^3\times\{0 \}$ and $\hat\Gamma_1 =\hat\Sigma \cap S^3\times\{ 1 \}$ both contain the same marked subgraph $\Gamma$, the domain of the 1-morphisms, and similarly with the range $\Gamma'$ contained in both $\hat\Gamma_0' = \hat\Sigma' \cap S^3\times\{0 \}$ and $\hat\Gamma_1' =\hat\Sigma' \cap S^3\times\{ 1 \}$.

Thus, each branched cover cobordism $W$ which gives a 2-morphism between given 1-morphisms $M$ and $M'$ contains inside the branch loci $\hat\Sigma$ and $\hat\Sigma'$ two subcomplexes $\Sigma$ and $\Sigma'$ which are (possibly nontrivial) cobordisms of $\Gamma$ with itself and of $\Gamma'$ with itself, $\partial \Sigma = \Gamma \cup \bar\Gamma$ and $\partial \Gamma' = \Sigma' \cup \bar \Sigma'$. We can then assign to each 2-morphism $W$ between assigned 1-morphisms $M$ and $M'$ two multiplicity functions that are the characteristic functions of the set of faces of $\Sigma \subset \hat\Sigma$ and of $\Sigma'\subset \hat\Sigma'$. The analog of the quantized area operator in this case becomes
\begin{equation*}
    A_\fs \, f(W) = \hbar \left(\sum_{f \in F(\hat\Sigma)} \chi_{\Sigma}(f) \, (j_f (j_f+1))^{1/2}\right) f(W),
\end{equation*}
where $j_f$ are the spins of the representations $\tilde\rho_f$ assigned to the faces of
$\hat\Sigma$ in the spin foam $\Psi$. The operator $A_\ft$ is defined similarly.

\begin{prop} \label{timeW}
    The time evolution $\sigma^A_t(f) = \ee^{\ii t (A_\fs - A_\ft)} f$ defines a time evolution on the convolution algebra $\cA_{\rm tsf}$ of the 2-category of topspin foams, which is compatible with both the horizontal and with the vertical convolution products.
\end{prop}

\begin{proof}
    The compatibility with the horizontal composition works as in \cref{Areatimeev} above. In fact, given $\Sigma\subset \hat\Sigma$, one has analogously $\Sigma \subset \hat\Sigma \cup q q_1^{-1}(\hat\Sigma_2)$ and $\Sigma''\subset \hat\Sigma'' \cup q'' q_2^{-1}(\hat\Sigma_1)$ as the branch loci of the fibered product $\tilde W$
    \begin{equation*}
        \Sigma\subset \hat\Sigma \cup q q_1^{-1}(\hat\Sigma_2) \subset S^3\times [0,1] \stackrel{Q}{\leftarrow} \tilde W \stackrel{Q'}{\rightarrow} S^3\times [0,1] \supset \hat\Sigma'' \cup q'' q_2^{-1}(\hat\Sigma_1) \supset \Sigma''
    \end{equation*}
    of the branched cover cobordisms
    \begin{gather*}
        \Sigma \subset \hat\Sigma  \subset S^3\times [0,1] \stackrel{q}{\leftarrow} W \stackrel{q_1}{\rightarrow} S^3\times [0,1] \supset \hat\Sigma_1 \supset \Sigma' \\
        \Sigma' \subset \hat\Sigma_2 \subset S^3\times [0,1] \stackrel{q_2}{\leftarrow} W' \stackrel{q'}{\rightarrow} S^3\times [0,1] \supset \hat\Sigma'' \supset \Sigma''.
    \end{gather*}
    The compatibility with the vertical composition comes from the fact that counting faces of $\Sigma$ or $\Sigma'$ with weights given by spins $(j_f(j_f+1))^{1/2}$ is additive under gluing of cobordisms along a common boundary and, as observed in \cite{CoveringsCorrespondencesNCG}, any additive invariant of the cobordisms gives rise to a time evolution compatible with the vertical composition of 2-morphisms.
\end{proof}

\subsection{More general time evolutions from spin foam amplitudes}

We gave in \cref{Areatimeev,timeW} a time evolution on the algebra $\cA_{\rm tsn}$ and on the 2-semigroupoid algebra $\cA_{\rm tsf}$, based on the quantized area operator. In fact, one can generalize this construction and obtain similar time evolutions based on amplitudes in various types of spin foam models.

We consider again the spin foam amplitudes as in \cref{dynGSsec}.
Suppose we are given a spin foam model where the amplitudes $\bA(\Sigma)$ are positive. One then proceeds as above and replaces the time evolution $\sigma_t^A(f) = \exp(\ii t (A_\fs-A_\ft)) f$ with the similar form
\begin{equation}\label{sigmatAmp}
    \sigma_t^{\bA}(f) = \left( \frac{\bA_\fs}{\bA_\ft} \right)^{\ii t} f,
\end{equation}
where, for a 2-morphism ${}_{\Psi \subset \hat\Psi} W_{\hat\Psi'\supset \Psi'} \in \MorTwo_{\cL(G)}$, one sets $\bA_\fs(W) \equiv \bA^0(\Psi)$ and $\bA_\ft(W) \equiv \bA^0(\Psi')$, with $\bA^0$ the normalized amplitude defined in \cref{normAmp}, so that
\begin{equation}\label{timeevAmpl}
    \sigma_t^{\bA}(f) \, \bigl({}_{\Psi \subset \hat\Psi} W_{\hat\Psi'\supset \Psi'}\bigr) = \left( \frac{\bA^0(\Psi)}{\bA^0(\Psi')}\right)^{\ii t}  f \bigl({}_{\Psi \subset \hat\Psi} W_{\hat\Psi'\supset \Psi'}\bigr).
\end{equation}
The quantized area case can be seen as a special case where one takes $\bA(\Sigma)=\exp(A(\Sigma))$, with $\bA_f=\exp( \hbar (j_f(j_f+1))^{1/2} )$.

Then one can show as in \cref{timeW} that one obtains in this way a time evolution on $\cA_{\rm tsf}$, compatible with both the vertical and the horizontal convolution product.

\begin{prop}\label{timeWAmp}
    Consider a spin foam model with positive amplitudes, where the weight factor satisfies \cref{weightmultipl}. Then the transformation \cref{timeevAmpl} defines a time evolution on $\cA_{\rm tsf}$, compatible with both the vertical and the horizontal products.
\end{prop}

\begin{proof}
    The compatibility with the horizontal product can be checked as in \cref{timeW}. Namely, one has
    \begin{align*}
            \sigma_t^{\bA}(f_1\circ f_2)({}_{\Psi} \tilde W_{\Psi''})
        & = \left( \frac{\bA^0(\Psi)}{\bA^0(\Psi'')}\right)^{\ii t}  (f_1\circ f_2)({}_{\Psi} \tilde W_{\Psi''}) \\
        & = \left( \frac{\bA^0(\Psi)}{\bA^0(\Psi'')}\right)^{\ii t} \sum f_1({}_{\Psi} W_{\Psi'}) \, f_2({}_{\Psi'} W'_{\Psi''}) \\
        & = \sum \left( \frac{\bA^0(\Psi)}{\bA^0(\Psi')}\right)^{\ii t}
f_1({}_{\Psi} W_{\Psi'}) \, \left( \frac{\bA^0(\Psi')}{\bA^0(\Psi'')}\right)^{\ii t}
f_2({}_{\Psi'} W'_{\Psi''}),
    \end{align*}
    where the sum is over all decompositions of the 2-morphism ${}_{\Psi} \tilde W_{\Psi''}$ as a horizontal composition (fibered product) of 2-morphisms ${}_{\Psi} W_{\Psi'}$ and ${}_{\Psi'} W'_{\Psi''}$, and we use the fact that the operators $\bA_\fs$ and $\bA_\ft$ only depend upon the source and target $\Psi$ and $\Psi''$ and not on the larger branch loci $\hat\Psi \supset \Psi$ and $\hat\Psi''\supset \Psi''$. The compatibility with the vertical composition follows directly from \cref{multiplAmp}.
\end{proof}

Once again, if we want to use low temperature KMS states as a way to select a class
of geometries as equilibrium states of this dynamics, we run into the problem of the
multiplicities in the spectrum which we have already discussed in the simpler setting of
\cref{KMSGSsec}. The situation here is more complicated, as there are now
three different but related reasons for the occurrence of infinite multiplicities in the
spectrum of the infinitesimal generator of the time evolution.
\begin{enumerate}
 \item The normalized amplitude $\bA^0$ is computed only on a fixed
 subcomplex $\Sigma \subset \hat\Sigma$ of the branch locus, the source of the morphism.
 One needs to resolve the ambiguity due to the possible choices of $\hat\Sigma$ with fixed $\Sigma$.
  \label{item:branchLocusMult}
    \item The amplitudes of \cref{normAmp} do not depend on the additional topological data $\tilde\sigma$ on the source spin foam $\Psi$. This generates an additional degeneracies, due to the fact that, for instance, the order $n$ of the covering can be assigned arbitrarily.    \label{item:topologicalMult}
    \item The combinatorial part of the amplitudes of \cref{amplcomb} depends only on the combinatorial structure of $\Sigma$, not on its different topological embeddings in $S^3 \times [0,1]$. Only the weight factor $\omega(\Sigma)$ may distinguish topologically inequivalent embeddings.         \label{item:embeddingMult}
\end{enumerate}

Trying to resolve this problem leads once again to a question about the existence of
a numerical invariant of embedded two-complexes in $S^3\times [0,1]$ with some
growth conditions. We formulate here the necessary requirements that such a hypothetical
invariant would have to fulfill.

\begin{ques}\label{invariantques}
Is it possible to construct an invariant $\chi(\Sigma,W)$ of
embedded two-complexes $\Sigma \subset S^3\times [0,1]$, which depends on the
data of a branched cover $q:W\to S^3\times [0,1]$, with the following properties?
\begin{enumerate}
\item The values of $\chi(\Gamma,W)$ form discrete set of positive real numbers $\{ \alpha(n) \}_{n\in \N}\subset \R^*_+$, which grows at least linearly, $\alpha(n) \geq c_1 n +c_0 $, for sufficiently large $n$
and for some $c_i >0$.
\item The number of embedded two-complexes $\Sigma \subset S^3\times [0,1]$ such that
$\chi(\Sigma,W)=\alpha(n)$, for fixed branched covering data $W$, grows like
\begin{equation}
\# \{ \Sigma \subset S^3\times[0,1]\,\,:\,\,  \chi(\Gamma,W)=\alpha(n) \} \leq O(\ee^{\kappa n }),
\end{equation}
for some $\kappa >0$, independent of $W$.
\item For $\tilde W$ the fibered product of two branched coverings $W$ and $W'$,
\begin{equation}\label{addchitildeW}
\chi(\Sigma \cup q q_1^{-1}(\Sigma_2),\tilde W) =\chi(\Sigma,W) + \chi(\Sigma_2,W'),
\end{equation}
where
$$ \Sigma \subset S^3 \times [0,1] \stackrel{q}{\leftarrow} W \stackrel{q_1}{\rightarrow} S^3 \times [0,1] \supset \Sigma_1 $$
$$ \Sigma_2 \subset S^3 \times [0,1] \stackrel{q_2}{\leftarrow} W' \stackrel{q'}{\rightarrow} S^3 \times [0,1] \supset \Sigma' $$
$$ \Sigma \cup q q_1^{-1}(\Sigma_2) \subset S^3 \times [0,1] \stackrel{Q}{\leftarrow} \tilde W \stackrel{Q'}{\rightarrow} S^3 \times [0,1] \supset \Sigma' \cup q'q_2^{-1}(\Sigma_1) $$
are the branched covering maps.
\end{enumerate}
\end{ques}

Notice that the property \eqref{addchitildeW} by itself can be interpreted as a kind of
weighted Euler characteristic. For instance, suppose given a PL
branched cover $q: W \to S^3\times [0,1]$ and an embedded two-complex
$\Sigma \subset S^3\times [0,1]$. Assume that, if $\Sigma$ intersect the branch locus,
it does so on a subcomplex. Set
\begin{equation}\label{weightchi}
\tilde\chi(\Sigma,W):= \sum_{f\in \Sigma^{(2)}} n_f  - \sum_{e\in \Sigma^{(1)}} n_e  + \sum_{v\in \Sigma^{(0)}} n_v ,
\end{equation}
where the $n_f$, $n_e$, $n_v$ are the order of the covering map over the faces, edges, and
vertices of $\Sigma$. This computes an Euler characteristic weighted by the multiplicities of the
covering. By the inclusion-exclusion property of the Euler characteristic and its multiplicativity
under covering maps, one sees that $\tilde\chi(\Sigma,W)=\chi(q^{-1}(\Sigma))$ is in fact the
Euler characteristic of the preimage of $\Sigma$ under the covering map. If $\Sigma$ does not
intersect the branched locus, then $\tilde\chi(\Sigma,W)=n \chi(\Sigma)$, where $n$ is the order
of the covering $q: W \to S^3\times [0,1]$. When taking the fibered product $\tilde W$ of $W$
and $W'$, the multiplicities of the covering map $Q: \tilde W\to S^3 \times [0,1]$ are the
product of the multiplicities of the map $q$ and of the first projection $W\times W'\to W$
restricted to $\tilde W\subset W \times W'$. The latter is the same as the multiplicity $n'(x)$ of
the map $q_2: W' \to S^3\times [0,1]$, on each point of the fiber $q_1^{-1}(x)$. Thus, we have
$$ \tilde\chi(\Sigma,\tilde W)= \sum_{f\in \Sigma^{(2)}} n'_{q_1q^{-1}(f)} n_f  - \sum_{e\in \Sigma^{(1)}} n'_{q_1q^{-1}(e)} n_e  + \sum_{v\in \Sigma^{(0)}} n'_{q_1q^{-1}(v)} n_v . $$
Similarly, we have
$$ \tilde\chi(q q_1^{-1}(\Sigma_2),\tilde W)=\sum_{f\in \Sigma_2^{(2)}} n_{qq_1^{-1}(f)} n'_f  - \sum_{e\in \Sigma_2^{(1)}} n_{qq_1^{-1}(e)} n'_e  + \sum_{v\in \Sigma_2^{(0)}} n_{qq_1^{-1}(v)} n'_v . $$
In the generic case where $\Sigma \cap q q_1^{-1}(\Sigma_2) =\emptyset$, we have
$n_{qq_1^{-1}(f)}=n_{qq_1^{-1}(e)} =n_{qq_1^{-1}(v)}=n$ and
$n'_{q_1q^{-1}(f)} =n'_{q_1q^{-1}(e)}=n'_{q_1q^{-1}(v)} =n'$, so that we get
$$  \tilde\chi(\Sigma\cup q q_1^{-1}(\Sigma_2),\tilde W)
= n' \tilde\chi(\Sigma,W)+ n \tilde\chi(\Sigma_2,W'). $$
One then sets $\chi(\Sigma,W)=\tilde\chi(\Sigma,W)/n$ and one obtains
$$ \chi(\Sigma\cup q q_1^{-1}(\Sigma_2),\tilde W)=\tilde\chi(\Sigma\cup q q_1^{-1}(\Sigma_2),\tilde W)/(nn')=\tilde\chi(\Sigma,W)/n+\tilde\chi(\Sigma_2,W')/n' . $$
This invariant does not suffice to satisfy the properties listed in \cref{invariantques},
but it helps to illustrate the meaning of condition \eqref{addchitildeW}.

\subsection{Multiplicities in the spectrum and type II geometry}

The problem of infinite multiplicities in the spectrum, that we encountered in our construction of time evolutions on the convolution algebras of topspin foams and topspin networks, may in fact be related to the occurrence of type II spectral triples associated to spin networks in the work \cite{AGNP3,AGNP4}. Type II spectral triples also typically occur in the presence of infinite multiplicities in the spectrum, and passing to a construction where regular traces are replaced by von Neumann traces provides a way to resolve these multiplicities. A similar occurrence also arises in the context of $\Q$-lattices in \cite{FromPhysicsToNTViaNCG1}, where one considers the determinant part of the quantum statistical mechanical system of 2-dimensional $\Q$-lattices. In that case, again, the partition function of the system is computed with respect to a von Neumann trace instead of the ordinary trace. A similar occurrence of infinite multiplicities and partition function computed with respect to a von Neumann trace also occurs in the quantum statistical mechanical systems considered in \cite{CyclotomyAndEndomotives}. Thus, a different way to approach the problem of infinite multiplicities and of obtaining a well defined and finite partition function for the system at low temperature may be through passing to a type II setting, instead of attempting to construct invariants of embedded two-complexes with the prescribed growth conditions. We do not pursue this line of investigation in the present paper.

\section{Spin foams with matter: almost-commutative geometries} \label{sec:spinFoamsWithMatter}

One of the most appealing features of general relativity is its \textit{geometric} nature. That is, in general relativity the gravitational force manifests as the evolution of the metric structure of spacetime itself. In contrast, quantum field theories are decidedly non-geometric, with their dynamical variables being located in an abstract Hilbert space and interacting according to various representation-theoretic data that specify ``particles."

Traditionally, one places quantum field theories on a background spacetime $M$ by postulating the  existence of a bundle structure over $M$. In this way, one creates an a priori distinction between the base manifold and fibers over it. The former is a geometric object, invariant under diffeomorphisms, while the latter are quantum field-theoretic, and invariant under gauge transformations. While this is more or less the most straightforward way of combining these two theories into one with the desired symmetry group $C^\infty(M, G) \rtimes \Diff(M)$, in the end what we have is a rather unmotivated and inelegant fusion.

The idea that noncommutative geometry could provide a purely geometric interpretation for quantum field theories goes back to the beginnings of the former subject. Since noncommutative geometry gives us access to more general spaces beyond our ordinary ``commutative manifolds," the hope is that we could find some noncommutative space $X$ such that our quantum field theory is given by the evolution of the geometry on this space, just as general relativity is a theory of the evolution of the geometry of a commutative manifold $M$.

Quickly summarizing the current state of the art on this approach, we will simply say that this hope was indeed born out. In \cite{GravityAndSMWithNM}, Chamseddine, Connes, and Marcolli were able to reproduce the Standard Model minimally coupled to gravity and with neutrino mixing in a fully geometric manner---that is, as a theory of \textit{pure gravity} on a suitable noncommutative space. This space is a product of an ordinary four-dimensional spacetime manifold with a small noncommutative space, which is metrically zero-dimensional but K-theoretically six-dimensional, and which accounts for the matter content of the particle physics model. In particular, the analogue of the diffeomorphism group for the noncommutative space in question reproduces the desired symmetry structure $C^\infty(M, G) \rtimes \Diff(M)$, with $G = \mathrm{SU}(3) \times \mathrm{SU}(2) \times \mathrm{U}(1)$ as desired.

The elegant success of noncommutative-geometry techniques in coupling quantum field-theoretic matter to \textit{classical} spacetime geometry then suggests a possible approach for doing the same with the \textit{quantum} spacetime geometry of loop quantum gravity. In what follows, we give a brief flavor of how gauge theories on a classical manifold are constructed in the framework of noncommutative geometry, before surveying existing work on incorporating the spin networks and foams of loop quantum gravity into this same framework. Once we have an idea of how quantum spacetime and matter are separately accounted for by noncommutative geometry, we suggest an approach to combining them, thus adding matter to loop quantum gravity and achieving a unified theory of quantum matter on quantum spacetime.

\subsection{The noncommutative geometry approach to the Standard Model}

The basic primitive of noncommutative differential geometry is known as a \textit{spectral triple}:
\begin{defn}
    A spectral triple is a tuple $(\cA, \cH, D)$ consisting of an involutive unital algebra $\cA$ represented as operators on a Hilbert space $\cH$, along with a self-adjoint operator $D$ on $\cH$ with compact resolvent such that the commutators $[D, a]$ are bounded for all $a \in \cA$.
\end{defn}
Intuitively, $\cA$ gives the algebra of observables over our space, while $\cH$ describes the (fermionic) matter fields on it. The so-called \textit{Dirac operator} $D$ contains both metric information about the space and interaction information for the matter fields on it. It also determines all the boson fields via inner fluctuations, as illustrated below.

The spectral triple formalism is very powerful. We first note that it entirely reproduces the ``commutative case'' of a spin manifold $M$ as the corresponding \textit{Dirac spectral triple}. This is defined by taking as $\cA$ the algebra $C^\infty(M)$ of scalar fields, as $\cH$ the Hilbert space $L^2(M, S)$ of square-integrable spinor fields, and as $D$ the Dirac operator $\slashed{\partial}_M \equiv \ii \gamma^\mu \nabla_\mu^S$ derived from the spin connection over $M$.

We then turn to the case of an \textit{almost-commutative geometry}, described as a product $M \times F$ of a commutative spin manifold and a finite noncommutative space. Here $F$ is finite in the sense that it is given by a spectral triple $(\cA_F,\cH_F,D_F)$, where the algebra $\cA_F$ is finite dimensional. The product geometry is then described by the cup product of the two spectral triples
\begin{equation*}
	(\cA,\cH,D) = (C^\infty(M),L^2(M, S),\slashed{\partial}_M) \cup (\cA_F,\cH_F,D_F),
\end{equation*}
where
\begin{align*}
    \cA & = C^\infty(M)\otimes \cA_F = C^\infty(M, \cA_F), \\
    \cH & = L^2(M, S) \otimes \cH_F = L^2(M, S\otimes \cH_F), \\
    D   & = \slashed\partial_M \otimes 1+ \gamma_5 \otimes D_F.
\end{align*}

The simplest example of an almost-commutative spectral triple is given by
\begin{equation*}
    (\cA_F, \cH_F, D_F) = (M_N(\C), M_N(\C), 0),
\end{equation*}
i.e.\ with $\cA_F$ the algebra of $N \times N$ matrices acting on the Hilbert space of these same matrices, endowed with the Hilbert--Schmidt norm, and with trivial Dirac operator. (To avoid overly complicating the discussion, we omit mention of certain important additional data related to these spectral triples, namely their \textit{real structure} and \textit{grading}; for more details on these, and on the overall construction of these spaces, see \S 1.10 of \cite{ClassTextbook}.)

Remarkably, if one takes the tensor product of these spectral triples, obtaining the almost-commutative geometry given by
\begin{align*}
    \cA & = C^\infty(M, M_N(\C)), \\
    \cH & = L^2(M, S) \otimes M_N(\C), \\
    D   & = \slashed\partial_M \otimes 1,
\end{align*}
one can show that the natural field theory arising from the resulting spectral triple is in fact gravity coupled to a $\mathrm{SU}(N)$ Yang--Mills field over $M$.

This, of course, raises the question of how one determines the ``natural field theory" arising from a given spectral triple. Such a field theory is realized in two parts.

First, one obtains the gauge bosons of the theory from \textit{inner fluctuations of the Dirac operator}. These are generated by the natural notion of isomorphism for noncommutative spaces, which is Morita equivalence of the algebras. When the noncommutative spaces are endowed with the structure of spectral triples, given an algebra $\cA'$ Morita-equivalent to $\cA$, there is a natural choice of Hilbert space $\cH'$, derived from $\cH$ by tensoring with the bimodule $\cE$ that realizes the Morita equivalence, on which $\cA'$ therefore acts; however, in order to obtain a new Dirac operator $D'$, one needs the additional datum of a connection on $\cE$. As a particular case of this procedure, one has the inner fluctuation, which correspond to the trivial Morita equivalence that leaves the algebra and the Hilbert space unchanged, but modifies the Dirac operator $D' = D + A$ by adding a self-adjoint perturbation of the form $A=\sum_i a_i [D,b_i]$, with elements $a_i,b_i\in \cA$. These are called the \textit{inner fluctuations} of the Dirac operator $D$, and are related to the inner automorphisms of $\cA$ that recover the gauge symmetries.  In the example of the product geometry with $\cA_F=M_N(\C)$, one finds that this class is given by operators of the form $A = A_\mu \gamma^\mu$, for $A_\mu \in C^\infty(M, M_N(\C))$ and $\gamma^\mu$ the usual Dirac matrices. We can thus see that, in the same way that gauge symmetries of a Yang--Mills system give rise to bosonic fields, the ``symmetry" of Morita equivalence gives rise to a field $A_\mu$ that we can indeed identify as a gauge boson. In the Standard Model case, one recovers in this way both the gauge boson and the Higgs.

Secondly, one introduces the \textit{Spectral Action Principle}, which is essential a strengthening of the requirement of diffeomorphism invariance, suitably generalized to noncommutative spaces \cite{TheSAP}. It requires that any physical action over our noncommutative space only depend on the spectrum of our Dirac operator $D$. Letting $D_A = D + A$ be a fluctuated version of the Dirac operator, the most general possibility is then
\begin{equation}
    S_D[A, \psi] = \Tr(f(D_A/\Lambda)) + \langle \psi, D_A \psi \rangle,
\end{equation}
with some possible variants on the fermionic term $\langle \psi, D_A \psi \rangle$ involving the real structure $J$ in cases where the metric and K-theoretic dimensions of the noncommutative space $F$ do not agree. Here $\Lambda$ is an energy scale that we eventually take to infinity, and $f$ is a smooth cutoff function whose exact form does not matter, but whose moments determine certain parameters of the model.

The remarkable result is then that, when applied to the example of the product geometry with $\cA_F=M_N(\C)$, the first (``bosonic'') term reproduces both the Einstein--Hilbert action (along with additional gravitational terms), plus the gauge field action term for $A_\mu$ in terms of its curvature $F_{\mu \nu}$:
\begin{equation} \label{eq:bosonicU1Action}
    \Tr(f(D_A/\Lambda)) = \int \Bigl(\tfrac{1}{16 \pi G} R + a_0 C_{\mu \nu \rho \sigma} C^{\mu \nu \rho \sigma} + \tfrac{1}{4 e^2} \Tr(F_{\mu \nu} F^{\mu \nu})\Bigr) \sqrt{g} \, \mathrm{d}^4 x + \mathrm{O}(\Lambda^{-2}).
\end{equation}
And as one would expect, the fermionic term gives the rest of the QED Lagrangian:
\begin{equation}
    \langle \psi, D_A \psi \rangle = \mathrm{i} \bar\psi \gamma^\mu (\partial_\mu - \mathrm{i} A_\mu) \psi - m \bar\psi\psi.
\end{equation}

\medskip

From here, we are only a short conceptual step away from the Noncommutative Standard Model. Indeed, all we need to do is modify the finite part of the product spectral triple to something a bit more complicated. We start by  considering the left-right symmetric algebra
\begin{equation*}
    \cA_{LR} = \C \oplus \H_L \oplus \H_R \oplus M_3(\C),
\end{equation*}
where $\H_L$ and $\H_R$ are two copies of the real algebra of quaternions. The natural Hilbert space to represent this algebra on is the sum of all inequivalent irreducible odd spin representations of $\cA_{LR}$, which we denote by $\cM$. The Hilbert space of our model is then simply three copies---one for each particle generation---of this natural Hilbert space:
\begin{equation*}
    \cH_F = \oplus^3 \cM.
\end{equation*}
As detailed further in Chapter 13 of \cite{ClassTextbook}, all fermions of the Standard Model are found as basis elements of this Hilbert space. The gauge bosons and the Higgs are obtained from the inner fluctuations of the Dirac operator, with the gauge bosons coming, as in the previous example, from fluctuations along the direction of the Dirac operator $\slashed\partial_M$ on the four-dimensional spacetime manifold, and the Higgs coming from the fluctuations in the direction of the Dirac operator $D_F$ on the finite space $F$. The Dirac operator $D_F$ that we choose to act on this space is essentially a large matrix giving all of the coupling constants and masses of our model. The demand that $D_F$ have off-diagonal terms, i.e.\ that there is nontrivial interaction between particles and antiparticles, together with the ``order one condition" for Dirac operators of spectral triples, breaks the left--right symmetry, by selecting as our spectral triple's algebra a subalgebra $\cA_F = \C \oplus \H \oplus M_3(\C)$, where the first two terms embed diagonally in $\cA_{LR}$.

Although the preceding paragraph was somewhat of a haphazard introduction to the Noncommutative Standard Model, the upshot is that when we apply the Spectral Action Principle to the product of the Dirac spectral triple with our finite $(\cA_F, \cH_F, D_F)$, we indeed obtain the complete Standard Model Lagrangian, for an extension of the Minimal Standard Model that includes right handed neutrinos with Majorana mass terms. The bosonic term $\Tr(f(D_A/\Lambda))$ gives us, analogously to the simpler case above, the bosonic and gravitational parts of the action---including the coupling of spin-one bosons to the Higgs as well as the Higgs term itself. And the fermionic term gives us the coupling of both spin-one and Higgs bosons to the fermion sector and all the terms in the Lagrangian involving only the fermions.

Thus, one can interpret the result in the following way: the spectral action functional can be thought of as an action functional for (modified) gravity on noncommutative spaces, which in the commutative case provides a combination of the Einstein--Hilbert action with cosmological term and conformal gravity, through the presence of a Weyl curvature term. When computed on an almost-commutative geometry, this gravity action functional delivers additional terms that provide the other bosonic fields. In particular, the manner in which the Higgs is incorporated gives it a natural interpretation as being part of the gravitational field on our noncommutative space.

In this way, our original hope of realizing particle physics as an entirely geometric phenomenon is borne out: the Standard Model Lagrangian has been reproduced as a natural spectral action over the specified noncommutative space.

\subsection{Spectral triples and loop quantum gravity}

The Noncommutative Standard Model, despite its success, still produces an essentially classical conception of gravity, as seen by the Einstein--Hilbert action embedded in \cref{eq:bosonicU1Action}. Indeed, the authors of \cite{TheSAP} comment on this directly in the context of their discussion of the mass scale $\Lambda$, noting that they do not worry about the presence of a tachyon pole near the Planck mass since, in their view, ``at the Planck energy the manifold structure of spacetime will break down and one must have a completely finite theory.''

Such a view is precisely that embodied by theories of quantum gravity, including of course loop quantum gravity---a setting in which spin networks and spin foams find their home. The hope would be to incorporate such existing work toward quantizing gravity into the spectral triple formalism by replacing the ``commutative part'' of our theory's spectral triple with something representing discretized spacetime. Seen from another point of view, if we can find a way of phrasing loop quantum gravity in the language of noncommutative geometry, then the spectral triple formalism provides a promising approach toward naturally integrating gravity and matter into one unified theory.

This idea of expressing LQG in terms of noncommutative geometry has been investigated recently in a series of papers by Aastrup, Grimstrup, Nest, and Paschke \cite{AGNP1,AGNP2,AGNP3,AGNP4,AGNP5,AGNP6,AGNP7,AGNP8,AGNP9,AGNP10}. Their starting point is to construct a spectral triple from the algebra of holonomy loops. The interaction between this algebra and the spectral triple's Dirac operator reproduces the Poisson bracket of both Yang--Mills theory and of the Ashtekar formulation of general relativity upon which LQG is based. Later papers illuminate the situation by making contact with the semiclassical limit, where they retrieve the Dirac Hamiltonian for fermions coupled to gravity in \cite{AGNP10}. Their approach also resolves a long-standing obstruction to na\"ive transference of the LQG Hilbert space into a spectral triple: namely, that the Hilbert space in question is nonseparable, which creates problems for the spectral triple construction.

\subsection{Spectral correspondences}

In section 12 of \cite{CoveringsCorrespondencesNCG}, the 2-category of low dimensional topologies, which we used here as the basis for the 2-semigroupoid algebra of topspin foams, was enriched by passing to almost-commutative geometries, where the three-manifolds and four-manifolds are described as commutative spectral triples and one takes then a product with a finite noncommutative space, as in the particle physics models described above. The branched cover and branched cover cobordism structure can still be encoded in this almost-commutative framework and so are the composition products by fibered product along the branched covering maps and by  composition of cobordisms, using a suitable notion of spectral triple with boundary proposed by Chamseddine and Connes. This suggests that, if we think in a similar way of spin foams with matter as products of a spectral triple associated to the spin foam and of a finite noncommutative space, then the convolution algebras of spin networks and spin foams considered here will extend to these spectral correspondences with the additional finite noncommutative spectral triples.

One important technical aspect involved in following a similar approach is the fact that the constructions of spectral triples associated to spin networks yield type II spectral geometries. Type II geometries are typically arising in connection to the problem of infinite multiplicities in the spectrum of the Dirac operator, of which we saw an aspect in this paper as well. Thus, one needs to adapt also the spectral action formalism to extend to this case. Roughly, instead of considering a functional of the form $\Tr(f(D_A/\Lambda))$, one needs to replace the ordinary trace by a type II von Neumann trace $\tau$. We do not enter into these aspects in the present paper and reserve them for future investigations.

\subsection{Future work}

The work described above on spectral triples over the algebra of holonomy loops suggests a few further lines of inquiry with regard to spectral triples for quantum gravity.

One of the more pressing questions is the relation between the kinematical Hilbert space of LQG, and that constructed in the above theory. The fact that the spectral triple encodes the Poisson bracket of general relativity implies that it carries information on the kinematical sector of quantum gravity. However, due to differences in construction, it is clear that despite similar starting points with regards to the importance of holonomy variables, in the end the Hilbert space produced is not the same as that of LQG.

Since spin networks arise naturally as the basis for the LQG kinematical Hilbert space, elucidating this relationship might allow us to incorporate the additional background-independence provided by topspin networks into the existing spectral triple theory. It would also help clarify the relation between their construction and the covariant formulation of LQG in terms of spin foams, or the obvious extension to the topspin foams discussed in this paper. This could also provide a hint on how to extend the spectral triple theory to evolving spacetimes; in its existing form, it only contains states that live on static four-manifolds.

Another natural line of inquiry would be to determine the spectral action for this quantum gravity spectral triple. The existing work manages to extract the kinematics via consideration of other aspects of the Dirac operator; namely, its interaction with the algebra to produce the appropriate Poisson bracket structure, or its semiclassical expectation value in order to retrieve the Dirac Hamiltonian. But computation of the spectral action would give a firmer grasp of what type of physical theory arises from this spectral triple. It would also be the first step toward integrating it with a Standard Model--like matter sector in the manner envisioned above, by using it as a replacement for the Dirac spectral triple of classical gravity.

\bigskip

\section*{Acknowledgments}

This work was inspired by and partially carried out during the workshop ``Noncommutative Geometry and Loop Quantum Gravity" at the Mathematisches Forschungsinstitut Oberwolfach, which the first two authors thank for the hospitality and support. The second author is partially supported by NSF grants DMS-0651925 and DMS-0901221. The first author was partially supported by a Richter Memorial Fund Summer Undergraduate Research Fellowship from Caltech.

\bibliographystyle{utphys}
\bibliography{QuantumGravity,NoncommutativeGeometry,Math}

\providecommand{\href}[2]{#2}\begingroup\raggedright\begin{thebibliography}{10}

\bibitem{BaezSpinFoam}
J.~C. Baez, ``Spin foam models,''
  \href{http://dx.doi.org/10.1088/0264-9381/15/7/004}{{\em Class. Quant. Grav.}
  {\bf 15} (1998)  1827--1858}, \href{http://arxiv.org/abs/gr-qc/9709052}{{\tt
  arXiv:gr-qc/9709052}}.

\bibitem{RovelliQG}
C.~Rovelli, {\em Quantum Gravity}.
\newblock Cambridge UP, 2007.

\bibitem{KauffmanInvariants}
L.~H. Kauffman, ``Invariants of graphs in three-space,'' {\em Trans. Amer.
  Math. Soc.} {\bf 311} (1989) no.~2, 697--710.
  \url{http://www.jstor.org/stable/2001147}.

\bibitem{YetterKnottedGraphs}
D.~N. Yetter, ``Category theoretic representations of knotted graphs in
  {$S^3$},'' {\em Advances in Math.} {\bf 77} (1989)  137--155.

\bibitem{AlexanderRiemannSpaces}
J.~W. Alexander, ``Note on {R}iemann spaces,'' {\em Bull. Amer. Math. Soc.}
  {\bf 26} (1920)  370--372.

\bibitem{FoxKnotTheory}
R.~H. Fox, ``A quick trip through knot theory,'' in {\em Topology of
  3-manifolds}, J.~Fort, M.~K., ed., pp.~120--167.
\newblock Prentice-Hall, Englewood Cliffs, NJ, 1962.
\newblock \url{http://www.math.uic.edu/~kauffman/QuickTrip.pdf}.

\bibitem{PraSoIntroToInvariants}
V.~V. Prasolov and A.~B. Sossinsky, {\em Knots, Links, Braids and 3-Manifolds:
  an Introduction to the New Invariants in Low-Dimensional Topology}.
\newblock Amer. Math. Soc., 1996.

\bibitem{ThreeFoldBranchedCoverings}
H.~M. Hilden, ``Three-fold branched coverings of {$S^3$},'' {\em Amer. J.
  Math.} {\bf 98} (1976) no.~4, 989--997.
  \url{http://www.jstor.org/stable/2374037}.

\bibitem{ThreeManifoldsAsBranchedCovers}
J.~M. Montesinos, ``Three-manifolds as 3-fold branched covers of {$S^3$},''
  \href{http://dx.doi.org/10.1093/qmath/27.1.85}{{\em Q. J. Math.} {\bf 27}
  (1976) no.~1, 85--94}.

\bibitem{FourManifoldsAsBranchedCovers}
R.~Piergallini, ``Four-manifolds as 4-fold branched covers of {$S^4$},'' {\em
  Topology} {\bf 34} (1995) no.~3, 497--508.

\bibitem{CoveringsCorrespondencesNCG}
M.~Marcolli and A.~Z. al~Yasry, ``Coverings, correspondences, and
  noncommutative geometry,''
  \href{http://dx.doi.org/10.1016/j.geomphys.2008.07.007}{{\em Journal of
  Geometry and Physics} {\bf 58} (2008) no.~12, 1639--1661},
  \href{http://arxiv.org/abs/0807.2924}{{\tt arXiv:0807.2924 [math-ph]}}.

\bibitem{ClassTopAndCombinGroupTheory}
J.~Stillwell, {\em Classical Topology and Combinatorial Group Theory}.
\newblock Springer, 2~ed., 1993.

\bibitem{AlexPolynomialsAsIsoInvariants}
S.~Kinoshita, ``Alexander polynomials as isotopy invariants. {I},'' {\em Osaka
  Math. J.} {\bf 10} (1958) no.~1, 263--271.
  \url{http://projecteuclid.org/euclid.ojm/1200689436}.

\bibitem{CoveringMovesAndKirbyCalculus}
I.~Bobtcheva and R.~Piergallini, ``Covering moves and {K}irby calculus,''
  \href{http://arxiv.org/abs/arXiv:math/0407032}{{\tt
  arXiv:arXiv:math/0407032}}.

\bibitem{ClassTextbook}
A.~Connes and M.~Marcolli, {\em Noncommutative Geometry, Quantum Fields and
  Motives}.
\newblock AMS, Providence, RI, 2007.
\newblock \url{http://www.its.caltech.edu/~matilde/coll-55.pdf}.

\bibitem{MilnorInvariantsSpatialGraphs}
T.~Fleming, ``{M}ilnor invariants for spatial graphs,'' {\em Topology Appl.}
  {\bf 155} (2008) no.~12, 1297--1305,
  \href{http://arxiv.org/abs/0704.3286}{{\tt arXiv:0704.3286 [math.GT]}}.

\bibitem{KhovanovHomologyEmbeddedGraphs}
A.~Z. {al-Yasry}, ``Khovanov homology and embedded graphs,''
  \href{http://arxiv.org/abs/0911.2608}{{\tt arXiv:0911.2608 [math.AT]}}.

\bibitem{CobordismBetweenLinks}
F.~Hosokawa, ``A concept of cobordism between links,'' {\em Ann. Math.} {\bf
  86} (1967) no.~2, 362--373. \url{http://www.jstor.org/stable/1970693}.

\bibitem{HamberQG}
H.~W. Hamber, {\em Quantum Gravitation: The {F}eynman Path Integral Approach}.
\newblock Springer, 2008.

\bibitem{WaveFunctionOfTheUniverse}
J.~B. Hartle and S.~W. Hawking,
  \href{http://dx.doi.org/10.1103/PhysRevD.28.2960}{``Wave function of the
  universe,''{\em Phys. Rev. D} {\bf 28} (Dec, 1983)  2960--2975}.

\bibitem{ExoticSmoothnessSemiclassicalEQG}
C.~Duston, ``Exotic smoothness in 4 dimensions and semiclassical {E}uclidean
  quantum gravity,'' \href{http://arxiv.org/abs/0911.4068}{{\tt arXiv:0911.4068
  [gr-qc]}}.

\bibitem{ExoticSpacesInQG1}
K.~Schleich and D.~Witt, ``Exotic spaces in quantum gravity. {I}: {E}uclidean
  quantum gravity in seven dimensions,''
  \href{http://dx.doi.org/10.1088/0264-9381/16/7/319}{{\em Class. Quant. Grav.}
  {\bf 16} (1999)  2447--2469}, \href{http://arxiv.org/abs/gr-qc/9903086}{{\tt
  arXiv:gr-qc/9903086}}.

\bibitem{ExoticSmoothnessAndPhysics}
T.~Asselmeyer-Maluga and C.~H. Brans, {\em Exotic Smoothness and Physics:
  Differential Topology and Spacetime Models}.
\newblock World Scientific, 2007.
\newblock \url{http://www.worldscibooks.com/physics/4323.html}.

\bibitem{EisensteinSeriesAndQAAs}
M.~M. Kapranov, \href{http://dx.doi.org/10.1007/BF02399194}{``{E}isenstein
  series and quantum affine algebras,''{\em J. Math. Sci.} {\bf 84} (May, 1997)
   1311--1360}, \href{http://arxiv.org/abs/alg-geom/9604018}{{\tt
  arXiv:alg-geom/9604018}}.

\bibitem{CyclotomyAndEndomotives}
M.~Marcolli, \href{http://dx.doi.org/10.1134/S2070046609030042}{``Cyclotomy and
  endomotives,''{\em P-Adic Numbers, Ultrametric Analysis, and Applications}
  {\bf 1} (September, 2009)  217--263},
  \href{http://arxiv.org/abs/0901.3167}{{\tt arXiv:0901.3167 [math.QA]}}.

\bibitem{OperatorAlgebrasAndQSM1}
O.~Bratteli and D.~W. Robinson, {\em Operator Algebras and Quantum Statistical
  Mechanics}, vol.~1.
\newblock Springer, 2~ed., 1987.

\bibitem{OperatorAlgebrasAndQSM2}
O.~Bratteli and D.~W. Robinson, {\em Operator Algebras and Quantum Statistical
  Mechanics}, vol.~2.
\newblock Springer, 2~ed., 1997.

\bibitem{FromPhysicsToNTViaNCG1}
A.~Connes and M.~Marcolli, ``From physics to number theory via noncommutative
  geometry. part i: Quantum statistical mechanics of $\mathbb{Q}$-lattices,''
  in {\em Frontiers in Number Theory, Physics, and Geometry}, P.~Cartier,
  B.~Julia, P.~Moussa, and P.~Vanhove, eds., vol.~1, pp.~269--347.
\newblock Springer, Berlin, 2005.
\newblock \href{http://arxiv.org/abs/arXiv:math/0404128}{{\tt
  arXiv:math/0404128}}.

\bibitem{BostConnesAlgebras}
J.~B. Bost and A.~Connes, \href{http://dx.doi.org/10.1007/BF01589495}{``Hecke
  algebras, type {III} factors and phase transitions with spontaneous symmetry
  breaking in number theory,''{\em Selecta Mathematica, New Series} {\bf 1}
  (December, 1995)  411--457}.

\bibitem{GravityAndSMWithNM}
A.~H. Chamseddine, A.~Connes, and M.~Marcolli, ``Gravity and the {S}tandard
  {M}odel with neutrino mixing,'' {\em Adv. Theor. Math. Phys.} {\bf 11} (2007)
   991--1089, \href{http://arxiv.org/abs/hep-th/0610241}{{\tt
  arXiv:hep-th/0610241}}.

\bibitem{LongIndexTheoremFoliations}
A.~Connes and G.~Skandalis, ``The longitudinal index theorem for foliations,''
  {\em Publ. Res. Inst. Math. Sci.} {\bf 20} (1984) no.~6, 1139--1183.

\bibitem{NCCorrespondencesDBranes}
J.~Brodzki, V.~Mathai, J.~Rosenberg, and R.~J. Szabo, ``Noncommutative
  correspondences, duality and {D}-branes in bivariant {K}-theory,'' {\em Adv.
  Theor. Math. Phys.} {\bf 13} (2009)  497--552,
  \href{http://arxiv.org/abs/0708.2648}{{\tt arXiv:0708.2648 [hep-th]}}.

\bibitem{NonAssocToriTDuality}
P.~Bouwknegt, K.~Hannabuss, and V.~Mathai, ``Nonassociative tori and
  applications to {T}-duality,''
  \href{http://dx.doi.org/10.1007/s00220-005-1501-8}{{\em Commun. Math. Phys.}
  {\bf 264} (2006)  41--69}, \href{http://arxiv.org/abs/hep-th/0412092}{{\tt
  arXiv:hep-th/0412092}}.

\bibitem{AGNP3}
J.~Aastrup, J.~M. Grimstrup, and R.~Nest, ``On spectral triples in quantum
  gravity {I},'' \href{http://dx.doi.org/10.1088/0264-9381/26/6/065011}{{\em
  Class. Quant. Grav.} {\bf 26} (2009)  065011},
  \href{http://arxiv.org/abs/0802.1783}{{\tt arXiv:0802.1783 [hep-th]}}.

\bibitem{AGNP4}
J.~Aastrup, J.~M. Grimstrup, and R.~Nest, ``On spectral triples in quantum
  gravity {II},'' {\em J. Noncommut. Geom.} {\bf 3} (2009)  47--81,
  \href{http://arxiv.org/abs/0802.1784}{{\tt arXiv:0802.1784 [hep-th]}}.

\bibitem{TheSAP}
A.~H. Chamseddine and A.~Connes, ``The spectral action principle,''
  \href{http://dx.doi.org/10.1007/s002200050126}{{\em Commun. Math. Phys.} {\bf
  186} (1997)  731--750}, \href{http://arxiv.org/abs/hep-th/9606001}{{\tt
  arXiv:hep-th/9606001}}.

\bibitem{AGNP1}
J.~Aastrup and J.~M. Grimstrup, ``Spectral triples of holonomy loops,''
  \href{http://dx.doi.org/10.1007/s00220-006-1552-5}{{\em Commun. Math. Phys.}
  {\bf 264} (2006)  657--681}, \href{http://arxiv.org/abs/hep-th/0503246}{{\tt
  arXiv:hep-th/0503246}}.

\bibitem{AGNP2}
J.~Aastrup and J.~M. Grimstrup, ``Intersecting {C}onnes noncommutative geometry
  with quantum gravity,''
  \href{http://dx.doi.org/10.1142/S0217751X07035306}{{\em Int. J. Mod. Phys.}
  {\bf A22} (2007)  1589}, \href{http://arxiv.org/abs/hep-th/0601127}{{\tt
  arXiv:hep-th/0601127}}.

\bibitem{AGNP5}
J.~Aastrup, J.~M. Grimstrup, and R.~Nest, ``A new spectral triple over a space
  of connections,'' \href{http://dx.doi.org/10.1007/s00220-009-0758-8}{{\em
  Commun. Math. Phys.} {\bf 290} (2009)  389--398},
  \href{http://arxiv.org/abs/0807.3664}{{\tt arXiv:0807.3664 [hep-th]}}.

\bibitem{AGNP6}
J.~Aastrup, J.~M. Grimstrup, and R.~Nest, ``Holonomy loops, spectral triples \&
  quantum gravity,''
  \href{http://dx.doi.org/10.1088/0264-9381/26/16/165001}{{\em Class. Quant.
  Grav.} {\bf 26} (2009)  165001}, \href{http://arxiv.org/abs/0902.4191}{{\tt
  arXiv:0902.4191 [hep-th]}}.

\bibitem{AGNP7}
J.~Aastrup, J.~M. Grimstrup, M.~Paschke, and R.~Nest, ``On semi-classical
  states of quantum gravity and noncommutative geometry,''
  \href{http://arxiv.org/abs/0907.5510}{{\tt arXiv:0907.5510 [hep-th]}}.

\bibitem{AGNP8}
J.~Aastrup, J.~M. Grimstrup, and M.~Paschke, ``Emergent {D}irac hamiltonians in
  quantum gravity,'' \href{http://arxiv.org/abs/0911.2404}{{\tt arXiv:0911.2404
  [hep-th]}}.

\bibitem{AGNP9}
J.~Aastrup and J.~M. Grimstrup, ``Lattice loop quantum gravity,''
  \href{http://arxiv.org/abs/0911.4141}{{\tt arXiv:0911.4141 [gr-qc]}}.

\bibitem{AGNP10}
J.~Aastrup, J.~M. Grimstrup, and M.~Paschke, ``On a derivation of the {D}irac
  {H}amiltonian from a construction of quantum gravity,''
  \href{http://arxiv.org/abs/1003.3802}{{\tt arXiv:1003.3802 [hep-th]}}.

\end{thebibliography}\endgroup

\end{document}